%% file: main_v2.tex
\documentclass[12pt]{article}

\usepackage{graphicx} 
\usepackage{latexsym}
\usepackage{amsmath}
\usepackage{amsfonts}
\usepackage{amssymb}
\usepackage{graphicx}
\usepackage{amsthm}
\usepackage{geometry}
\usepackage{color}
\usepackage{bm}

\usepackage[final]{showlabels}

\usepackage{hyperref}
\hypersetup{
    colorlinks=true,
    linkcolor=blue,
    filecolor=magenta,      
    urlcolor=cyan,
    citecolor=cyan,
    pdftitle={Fournodavlos, Schlue (2024)},
    pdfpagemode=FullScreen,
    }

\usepackage[center,it]{titlesec}
\usepackage[affil-it]{authblk}

\theoremstyle{plain}
\newtheorem{theorem}{Theorem}[section]
\newtheorem{proposition}{Proposition}[section]
\newtheorem{lemma}[proposition]{Lemma}
\newtheorem{corollary}[theorem]{Corollary}

\theoremstyle{definition}
\newtheorem{definition}[proposition]{Definition}

\theoremstyle{remark}
\newtheorem{remark}[proposition]{Remark}

\numberwithin{equation}{section}
\allowdisplaybreaks[4]

\DeclareMathOperator{\Ric}{Ric}
\DeclareMathOperator{\tr}{tr}

\renewcommand{\div}{\mbox{div}}

\newcommand{\ud}{\mathrm{d}}
\newcommand{\vol}[1]{\mathrm{vol}_{#1}}

\def\g{{\bm g}}
\def\Nabla{{\bm \nabla}}
\def\R{{\bm R}}
\def\RIC{{\bm{\mathrm{Ric}}}}

\def\K{\mathcal{K}_{a,m}}
\def\kGamma{{}^{\mathcal{K}}\bm{\Gamma}}

\def\tg{\widetilde{g}}
\def\tk{\widetilde{k}}
\def\tPhi{\widetilde{\Phi}}
\def\tGamma{\widetilde{\Gamma}}
\def\tnabla{\widetilde{\nabla}}

\def\hg{\widehat{g}}
\def\hk{\widehat{k}}
\def\hPhi{\widehat{\Phi}}
\def\hGamma{\widehat{\Gamma}}
\def\hnabla{\widehat{\nabla}}

\newcommand{\Oe}{\mathcal{O}(\widetilde{\varepsilon})}
\newcommand{\Os}[1]{\mathcal{O}(e^{#1 H s})}

\newcommand{\Er}[1]{\mathrm{Error}_{#1}}

\begin{document}

\title{\textbf{\Large Stability of the expanding region\\ of Kerr de Sitter spacetimes}}

\author[$\star$ $\dag$]{Grigorios Fournodavlos}

\author[$\ast$]{Volker Schlue}

\affil[$\star$]{\small University of Crete, Department of Mathematics \& Applied Mathematics, Voutes~Campus,~70013~Heraklion,~Greece\vskip.2pc gfournodavlos@uoc.gr \vskip.2pc \ }

\affil[$\dag$]{\small
Institute of Applied and Computational Mathematics,
FORTH, 70013 Heraklion, Greece\vskip.2pc \ }

\affil[$\ast$]{\small University of Melbourne, School of Mathematics and Statistics, Parkville~VIC~3010, Australia\vskip.2pc  volker.schlue@unimelb.edu.au\vskip.2pc \  }



\maketitle

\begin{abstract}
    We prove the nonlinear stability of the cosmological region of Kerr de Sitter spacetimes. More precisely, we show that solutions to the Einstein vacuum equations with positive cosmological constant arising from data on a cylinder that is uniformly close to the Kerr de Sitter geometry (with possibly different mass and angular momentum parameters at either end) are future geodesically complete and display asymptotically de Sitter-like degrees of freedom.
    The proof uses an ADM formulation of the Einstein equations in parabolic gauge. Together with a well-known theorem of Hintz-Vasy [Acta Math. 220 (2018)], our result yields a global stability result for Kerr de Sitter from Cauchy data on a spacelike hypersurface bridging two black hole exteriors. 
\end{abstract}

\tableofcontents

\section{Introduction}

In the presence of a cosmological constant $\Lambda>0$, the Einstein vacuum equations take the form
\begin{equation}
\label{eq:EVE}
    \RIC[\g]=\Lambda \g\,,
\end{equation}
where $\RIC[\g]$ is the Ricci curvature of the spacetime metric $\g$ on a $1+3$-dimensional Lorentzian manifold $(\mathcal{M},\g)$. Since the sum of the sectional curvatures $\bm{K}_i$ is negative,\footnote{Here $(e_0,e_1,e_2,e_3)$ is an orthonormal frame, $e_0$ is timelike, and $\bm{K}_i=\bm{R}(e_0,e_i,e_0,e_i)$ are the sectional curvatures associated to the planes spanned by $e_0$ and $e_i$.}
\begin{equation*}
    \RIC[\g](e_0,e_0)=\sum_{i=1}^3\, \bm{K}_i=\Lambda \: \g(e_0,e_0)<0\,,
\end{equation*}
solutions to \eqref{eq:EVE} may exhibit \emph{expansion} in all directions. The simplest example of a spacetime with this property is \textbf{de Sitter space}, which models a spatially closed expanding universe with topology $\mathbb{R}\times\mathbb{S}^3$.
The exact solutions to \eqref{eq:EVE}  that motivate this paper are the \textbf{Kerr de Sitter spacetimes} $(\mathcal{M},\g_{\K})$. In addition to black hole interior and exterior regions, they contain a spatially open expanding, or \emph{cosmological region}, with  topology $\mathbb{R}\times\mathbb{S}^2\times\mathbb{R}$.\footnote{Carter gives an excellent discussion of the maximal extension of Kerr de Sitter in \cite{carter}. For an introduction to the global geometry of the cosmological region, specifically in the context of the Cauchy problem see \cite{glw,schlue:weyl}.}

\smallskip
It is known since the work of Friedrich that de Sitter space is \emph{stable} as a solution to \eqref{eq:EVE}: Small perturbations of the initial data on $\mathbb{S}^3$ lead to future geodesically complete spacetimes \cite{friedrich:desitter}. The proof demonstrates in particular the existence of \emph{asymptotic functional degrees of freedom}, but it does \emph{not} apply to Kerr de Sitter spacetimes.\footnote{In \cite{friedrich:desitter} a conformal transformation is used to pass from \eqref{eq:EVE} to the \emph{conformal field equations} which turn out to be regular at the future boundary. In this way, Friedrich was able to reduce the global stability problem for de Sitter to a local in time problem, and identify the asymptotic degrees of freedom with the data on the conformal boundary. In Kerr de Sitter, the desired conformal transformation fails to be regular at $\iota^+$, and  this approach is limited to spatially compact subsets of the cosmological region, and has been applied away from the endpoints in \cite{valiente-kroon:marica,valiente-kroon:edgar}.} The exterior of slowly rotating Kerr de Sitter black holes have been proven to be \emph{asymptotically stable} in a series of influential papers by Hintz and Vasy \cite{vasy:13,hintz:vasy:16:global,hi:vasy:stability}: In the domain bounded by the event and cosmological horizons, small perturbations settle down exponentially fast to a nearby member of the  Kerr de Sitter family. In this paper, we complete the proof of the \emph{global} nonlinear stability of Kerr de Sitter, proving that the cosmological region is stable:

\smallskip
\begingroup
\leftskip 10pt
\rightskip 10pt

\noindent\textsl{Small perturbations on $\mathbb{R}\times\mathbb{S}^2$ which converge exponentially fast at both ends to nearby members of the Kerr de Sitter family lead to future geodesically complete solutions to \eqref{eq:EVE} which display asymptotically de Sitter-like degrees of freedom.}

\endgroup

\bigskip
We proceed with the precise statements.
\vspace{-15pt}
\paragraph{Kerr de Sitter metric.}
Recall the Penrose diagram of the Kerr de Sitter metric $\g_{\K}$ in Fig.~\ref{fig:penrose}, for small angular momentum $a$, $a^2\ll m^2\ll \Lambda^{-1}$. 
The cosmological region $\mathcal{R}$ lies to the future of the black hole exteriors $\mathcal{S}_1$, and $\mathcal{S}_2$, and is separated from these by the \emph{cosmological horizons} $\mathcal{C}_1$ and $\mathcal{C}_2$. The conformal boundary at infinity is denoted by $\mathcal{I}^+$, and the black hole regions $\mathcal{B}_1$, and $\mathcal{B}_2$ are in the complement of its past. 
\begin{figure}[tb]
  \centering
  \includegraphics[scale=1.4]{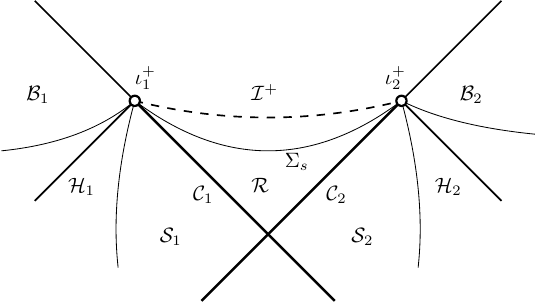}
  \caption{Penrose diagram of Kerr de Sitter geometry.}
  \label{fig:penrose}
\end{figure}
The conformal diagram in Fig.~\ref{fig:penrose} depicts various level sets of a function $r$,
which are timelike, spacelike, or null depending on the values of the polynomial
\begin{equation}
    \Delta_r=(r^2+a^2)\Bigr(1-\frac{\Lambda}{3}r^2\Bigr)-2mr\,,
\end{equation}
which may be positive, negative, or zero, respectively \cite{carter,gibbons:hawking}.
In the cosmological region $r$ is a time-function, and the metric takes the form
\begin{equation}
    \g_{\K}=-\Phi_{\K}^2\ud r^2+(g_{\K})_r\,,\qquad \Phi_{\K}^{-2}=\frac{\Lambda}{3}r^2-1+\mathcal{O}(r^{-1})\,,
\end{equation}
where $(g_{\K})_r$ is a Riemannian metric. Using a reparametrization of time 
\begin{equation}
    r=e^{Hs}\,,\qquad H=\sqrt{\frac{\Lambda}{3}}\,,
\end{equation}
that leads to an expression for the metric that is more common in cosmology,\footnote{In particular, in comparison to FLRW spacetimes; see for instance \cite{rod:speck:FLRW:euler,fournodavlos:FLRW} and references therein.}
we note that $\g_{\K}$ can then be expressed in coordinates so that in $\mathcal{R}$:
\begin{gather}
    \g_{\K}=-\Phi_{\K}^2\ud s^2+(g_{s,\K})_{ij}\ud x^i\ud x^j\,,\\
    \Phi_{\K}^2=1+\Os{-2}\,,\qquad (g_{s,\K})_{ij}=\Os{2}\,;
\end{gather}
for explicit formulas see  Section~\ref{sec:kerr}.  \textbf{Schwarzschild de Sitter} is obtained by setting $a=0$. In fact, $g_{\K}$ is a metric on $\mathbb{R}\times\mathbb{S}^2$ and there are coordinates $(x^1,x^2,x^3)=(t,x^2,x^3)$ such that $(g_{\K})_{ij}$ is independent of the $t$, and $(x^2,x^3)$ are spherical coordinates on $\mathbb{S}^2$. In other words, the metric $g_{\K}$ is translation invariant along the cylinder, and for $a=0$ takes the form:\footnote{See \cite[Section 3, 5 \& Appendix B]{glw} and \cite{gibbons:hawking,glw} for a detailed discussion of the geometry of the cosmological region, in these and comoving coordinates.}

\begin{equation} \label{eq:sds:metric}
    g_{\mathcal{K}_{m,0}}=\Bigl(\frac{\Lambda r^2}{3}+\frac{2m}{r}-1\Bigr)\ud t^2+r^2 \ud \sigma^2_{\mathbb{S}^2}
\end{equation}

\paragraph{Cauchy problem for Kerr de Sitter.}\footnote{For an introduction to the global stability problem for Kerr de Sitter see  \cite{hi:vasy:stability,schlue:weyl}; also \cite[Ch.~6]{dr:clay}.}
Consider a spacelike hypersurface $\Sigma$ in Schwarzschild de Sitter as depicted in Fig.~\ref{fig:cauchy}. The nonlinear stability result of Hintz and Vasy in \cite{hi:vasy:stability} implies that for initial data $(g,k)$ close to the data induced by a Schwarzschild de Sitter metric $\g_{\mathcal{K}_{0,m}}$,  the solution to \eqref{eq:EVE} converges to nearby members of the Kerr de Sitter family in both $\mathcal{S}_1$ and $\mathcal{S}_2$:\footnote{The statement applies independently to $\mathcal{S}_1$  and $\mathcal{S}_2$ by domain of dependence. The result is obtained in a generalized harmonic gauge, which is itself determined dynamically together with the final states $\g_{\mathcal{K}_{a_i,m_i}}, i=1,2$. The specific gauge used  in \cite{hi:vasy:stability} is not relevant for the following, except that the time variable $t$ which is used to express the rate of convergence is comparable to Schwarzschild de Sitter time, $T(t)=1$ where $\mathcal{L}_T \g_{\mathcal{K}_{0,m}}=0$.} 
\begin{theorem}[Stability of the stationary black hole exterior of slowly rotating Kerr de Sitter, Theorem~1.1 in \cite{hi:vasy:stability}] \label{thm:hv}
    For smooth initial data $(g,k)$, close to the data $(g_{\mathcal{K}_{0,m}},k_{\mathcal{K}_{0,m}})$ induced by the Schwarzschild de Sitter metric $\g_{\mathcal{K}_{0,m}}$ on $\Sigma$ (cf.~Figure~\ref{fig:cauchy}) in a higher Sobolev norm,
    there exists a solution to \eqref{eq:EVE} attaining the prescribed initial data on $\Sigma$,
    and parameters $(a_1,m_1)$ and $(a_2,m_2)$ with $\sum_{i=1}^2 |a_i|+|m_i-m|\ll 1$, so that
\begin{equation}
    \label{eq:HV}
    \g-\g_{\mathcal{K}_{a_i,m_i}}=\mathcal{O}(e^{-\alpha_i t}) \quad\text{: in } \mathcal{S}_i\,,\qquad i=1,2\,,
\end{equation}
    for constants $\alpha_i>0$, $i=1,2$ independent of the initial data.
\end{theorem}

In fact, the stability result \eqref{eq:HV} holds on a domain that goes \emph{beyond} the event and cosmological horizons, \emph{uniformly} in $r$; see \cite[Theorem~1.4]{hi:vasy:stability} for precise statements, and Section~3.5 therein for a discussion of the extension beyond the horizon. Therefore the existence of a development is known on a domain (whose future boundary is indicated by the dash dotted lines in Figure~\ref{fig:cauchy}) which  contains a level set $\Sigma_{s_0}\simeq\mathbb{R}\times\mathbb{S}^2$ of $r$, with the property that the geometric data on $\Sigma_{s_0}$ converges along one end to that induced by $\g_{\mathcal{K}_{a_1,m_1}}$, and that induced by $\g_{\mathcal{K}_{a_2,m_2}}$ along the other. This is the initial data for the evolution problem we consider.

\begin{figure}[tb]
  \centering
  \includegraphics[scale=1.4]{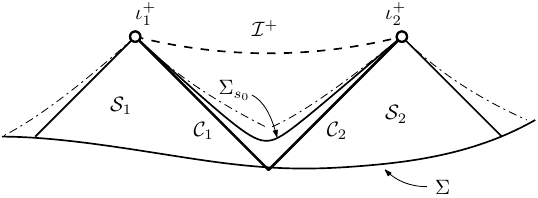}
  \caption{Cauchy problem for Kerr de Sitter.}
  \label{fig:cauchy}
\end{figure}

\paragraph{Geometric set-up for the main theorem.}
We establish the existence of a spacetime $(\mathcal{R},\g)$ which is foliated by the level sets of a time function $s$,
\begin{equation}
    \mathcal{R}=\bigcup_{s\in[s_0,\infty)} \Sigma_s\,, \qquad\Sigma_s\simeq \mathbb{R}\times\mathbb{S}^2\,,
\end{equation}
with each leaf $\Sigma_s$ being diffeomorphic to a cylinder $\mathbb{R}\times\mathbb{S}^2$. We choose coordinates so that
\begin{equation} \label{eq:g:intro}
    \g=-\Phi^2\ud s^2+(g_s)_{ij}\ud x^i \ud x^j\,,
\end{equation}
where $\Phi$ is the lapse function of the foliation, and $g_s$ is a Riemannian metric on $\Sigma_s$; see Section~\ref{sec:ADM}. The foliation is determined by a choice of the lapse function  which we take to be the \emph{solution of a parabolic PDE};\footnote{The PDE satisfied by the lapse function is the consequence of a geometric condition that involves a reference metric; see \eqref{lapse:intro} below. We defer the derivation of the PDE to Section~\ref{sec:evoleqs}. The specific gauge choice is of course central to the global existence proof in this setting. For a broader discussion of the gauge in relation to other works in the literature see Remark~\ref{rmk:gauge} below.} see Section~\ref{sec:evoleqs}.

In fact, in coordinates $(x^1,x^2,x^3)=(t,x^2,x^3)$ on $\mathbb{R}\times\mathbb{S}^2$,
the metric components of the solution $\g$ as well as of any member $\g_{\K}$ of the Kerr de Sitter family can be viewed as functions of $(r,t,x^2,x^3)$ in the fixed charts for the expanding region of Schwarzschild de Sitter.   Then also 
\begin{equation}
    \label{metric.ref.intro}
    \widetilde{\g}=\chi(t)  \g_{\mathcal{K}_{a_1,m_1}}+ (1-\chi(t) ) \g_{\mathcal{K}_{a_2,m_2}}
\end{equation}
is a metric on $\mathcal{R}$, where $\chi$ is a smooth cutoff function on $\mathbb{R}$ with $\chi(t)=0$ for $t\leq -1$, and $\chi(t)=1$ for $t\geq 1$. With this choice a \emph{reference metric}, which in the above coordinates again takes the form $\widetilde{\g}=-\tPhi^2\ud s^2+(\tg_s)_{ij}\ud x^i\ud x^j$ we can define
\begin{equation}
    \hPhi=\Phi-\tPhi\,,\quad\hg=g-\tg\,,\quad\:\hk=k-\tk\,,
\end{equation}
where $g$, $\tg$, and $k$, $\tk$ are the first and second fundamental forms of $\g$, and $\widetilde{\g}$ on $\Sigma_s$, respectively. In these coordinates,
we then also obtain from \eqref{eq:HV} that for $s=s_0$, $\hPhi,\hg,\hk\to 0$ as $t\to\pm\infty$; cf.~Figure~\ref{fig:cylinder}.

\paragraph{Hyperbolicity and Energies.}
In Section~\ref{sec:evoleqs} we cast the Einstein equations \eqref{eq:EVE} as a system of first order variation equations for  $\hg$, and $\hk$, (and the renormalised Christoffel symbols $\hGamma=\Gamma-\tGamma$). 
We show that in the gauge
\begin{equation}
    \label{lapse:intro}
    \Phi-\tPhi=\tr_g k-\tr_{\tg}\tk\,,
\end{equation}
the system of equations is essentially \emph{symmetric hyperbolic}, \footnote{The hyperbolic structure of the system is seen only in conjunction with the \emph{constraint equations}; see in particular Section~\ref{subsec:en.id.discussion}.}
in the sense that we have an \emph{energy identity} for the system derived in Section~\ref{sec:future}. 
The energies we use for the global existence argument are based on standard higher order Sobolev norms on $\Sigma_s$:
\begin{equation}\label{en.def.intro}
    \mathcal{E}_N(s)=\|\widehat g\|^2_{H^N(\Sigma_s,g)}+\|\widehat g^{-1}\|^2_{H^N(\Sigma_s,g)}
+e^{3Hs}\|\widehat\Phi\|^2_{H^N(\Sigma_s,g)}
+e^{2Hs}\big(\|\widehat\Gamma \|^2_{H^N(\Sigma_s,g)}
+\|\widehat k\|^2_{H^N(\Sigma_s,g)}\big)
\end{equation}
While suppressed from the notation,
the norms in the Sobolev spaces $H^N(\Sigma_s,g)=H^N_{\alpha_1,\alpha_2}(\Sigma_s,g)$ are weighted to incorporate exponential decay towards the two ends of the cylinder; see Section~\ref{sec:norms},~\ref{subsec:Boots}.
For fixed $\alpha_1\geq 0$, and $\alpha_2\geq 0$ in the definition of the norms, the Sobolev embedding reads; cf.~Figure~\ref{fig:cylinder}:
\begin{equation}
     \|(e^{\alpha_1 t}+e^{-\alpha_2 t})\mathcal{T}\|_{W^{N,\infty}(\Sigma_s,g)} \leq C  \|\mathcal{T}\|_{H_{\alpha_1,\alpha_2}^{N+2}(\Sigma_s,g)}
\end{equation}
Note that $\mathcal{E}_N(s)$ in \eqref{en.def.intro} refers to the energy of the \emph{renormalised} quantities, and measures the distance from the reference metric $\widetilde{\g}$. 

Moreover $\mathbb{R}\times\mathbb{S}^2$ is endowed with the standard metric on the cylinder, 
and in coordinates $(x^1,x^2,x^3)=(t,\theta,\phi)$,
\[\mathring{g}=\ud t^2+d\theta^2+\sin^2\theta \ud\phi^2\,,\] and corresponding (unweighted) norms $W^{M,\infty}(\mathbb{R}\times\mathbb{S}^2,\mathring{g})$ are defined in Section~\ref{sec:norms}, to measure the size of the asymptotic geometric quantities at infinity.

\paragraph{Main theorem.} We are now ready to state the main result in this paper, which is a \emph{global existence} result for the cosmological region, for initial data close to Kerr de Sitter. For \emph{local existence} in parabolic gauge see Appendix~\ref{sec:app}.
\begin{theorem}[Stability of the expanding region of Kerr de Sitter] \label{thm:fs}
Let $(\mathcal{R},\g)$ be a local in time solution to \eqref{eq:EVE} in parabolic gauge,
expressed relative to a time function $s$ whose level sets satisfy the geometric gauge condition \eqref{lapse:intro}, and 
\[ \mathcal{R}=\bigcup_{s=s_0}^{s_1}\Sigma_s\,, \qquad \Sigma_s\simeq \mathbb{R}\times\mathbb{S}^2\,,\]
for some $s_1>s_0>0$, with $\g$ expressed in local coordinates  as in \eqref{eq:g:intro}.

For $i=1,2$, let $|a_i|\ll m_i \ll \Lambda^{-1/2}$, and $\alpha_i\geq 0$.
Suppose for some $\mathring{\varepsilon}>0$,  the initial data $(g_0,k_0)$ on $\Sigma_{s_0}$
is $\mathring{\varepsilon}$-close to the data induced by a Kerr de Sitter metric $\g_{\K}$ expressed in this gauge, with parameters $(a_1,m_1)$ at one end, and $(a_2,m_2)$ at the other end,
in the sense that for some $N\geq 4$, with the energy defined in \eqref{en.def.intro}:
\begin{equation} \label{eq:thm:fs:small}
    \mathcal{E}_N(s_0)=\mathring{\varepsilon}^2\,.
\end{equation}

\begin{enumerate}
    \item Then, for $|a_1-a_2|$, $|m_1-m_2|$, and $\mathring{\varepsilon}>0$ sufficiently small,
    the solution is \emph{global},
\begin{equation}\label{eq:fs:R}
    \mathcal{R}=\bigcup_{s=s_0}^\infty\Sigma_s\,,\qquad\text{and}\quad \mathcal{E}_N(s)\leq C\: \mathcal{E}_N(s_0)\qquad \text{for all }s\geq s_0\,.
\end{equation}

\item Furthermore,  we have the following asymptotics for the spatial part of the metric \eqref{eq:g:intro}:
\begin{equation} \label{asym.free}
    g_{ij}(s,x)=g^\infty_{ij}(x) e^{2Hs}+h_{ij}(s,x)\,,\qquad g^{\infty}_{ij}=\tg^{\infty}_{ij}+\hg^{\infty}_{ij}\,, \qquad h_{ij}=\widetilde{h}_{ij}+\widehat{h}_{ij}\,,
\end{equation}
    where $\tg^{\infty}$, $g^{\infty},h,\widetilde{h}$ are metrics on $\mathbb{R}\times\mathbb{S}^2$,  with $\tg^{\infty}_{ij},\widetilde{h}_{ij}$ induced by the reference metric, satisfying
    \begin{equation}\label{asym.g.decay.intro}
         \|(e^{\alpha_1 t}+e^{-\alpha_2 t})\hg^{\infty}\|_{W^{N-4,\infty}(\mathbb{R}\times\mathbb{S}^2,\mathring{g})}\leq C \mathring{\varepsilon}\qquad \|(e^{\alpha_1 t}+e^{-\alpha_2 t})\widehat{h}\|_{W^{N-4,\infty}(\mathbb{R}\times\mathbb{S}^2,\mathring{g})}\leq C \mathring{\varepsilon} \,.
    \end{equation}
\item Finally, for $N\ge6$, the lapse function admits the asymptotic expansion
\begin{equation}
\label{asym.free.Phi}
\Phi(s,x)=1+\Phi^\infty(x)e^{-2Hs}+\Psi(s,x)\,,\qquad \Phi^\infty=\widehat{\Phi}^\infty+\widetilde{\Phi}^\infty\,,\qquad\Psi=\widehat{\Psi}+\widetilde{\Psi}\,,
\end{equation}
where $\widetilde{\Phi}^\infty,\widetilde{\Psi}$ are functions induced by the reference metric, satisfying
\begin{align}\label{asym.Phi.decay.intro}
\begin{split}
          \|(e^{\alpha_1 t}+e^{-\alpha_2 t}) \widehat{\Phi}^\infty\|_{W^{N-6,\infty}(\mathbb{R}\times\mathbb{S}^2,\mathring{g})}\leq&\, C\mathring{\varepsilon}\,,\\ 
          \| (e^{\alpha_1 t}+e^{-\alpha_2 t}) \widehat{\Psi} \|_{W^{N-6,\infty}(\mathbb{R}\times\mathbb{S}^2,\mathring{g})} \leq&\, C \mathring{\varepsilon} e^{-4Hs}\,.
          \end{split}
     \end{align}
\end{enumerate}
\end{theorem}
\begin{proof}
The global stability statement \emph{(I)} is proven in Corollary~\ref{cor:global}. The precise asymptotic behavior statements \emph{(II-III)} are the subject of Proposition~\ref{prop:k.Phi.ref.est}.
\end{proof}

\begin{figure}
    \centering
    \includegraphics[]{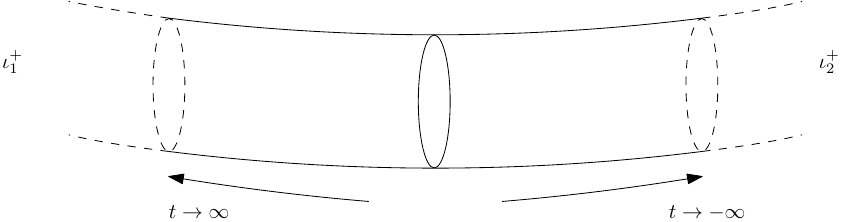}
    \caption{Topology of the the level sets $\Sigma_s$ diffeomorphic to $\mathbb{R}\times\mathbb{S}^2$.}
    \label{fig:cylinder}
\end{figure}

The stability of the expanding region of Kerr de Sitter established in Theorem~\ref{thm:fs} is a self-contained result for the Einstein equations \eqref{eq:EVE}, and independent of the  stability of the stationary black hole exterior given by Theorem~\ref{thm:hv}. Indeed for any smooth Cauchy data on $\Sigma_{s_0}$ sufficiently close to Kerr de Sitter satisfying \eqref{eq:thm:fs:small}, the solution is global to the future as described in Theorem~\ref{thm:fs}. In particular, this result does \emph{not} require the data to converge exponentially to Kerr de Sitter: Theorem~\ref{thm:fs}  holds with $\alpha_1=\alpha_2=0$. Nor does it rely on a smallness assumption for the angular momenta $a_i$, $i=1,2$.\footnote{In particular, if the black hole exteriors $\mathcal{S}_i$ were proven to be stable in the whole subextremal range, then Theorem~\ref{thm:fs} provides an immediate extension to the cosmological region.} However for the global stability problem for \emph{slowly} rotating Kerr de Sitter spacetimes, Theorem~\ref{thm:hv} implies that the assumptions of Theorem~\ref{thm:fs} are satisfied for some $(a_i,m_i)$, $\alpha_i>0$, $i=1,2$, yielding a global nonlinear stability result, in the domain of dependence of a Cauchy hypersurface bridging two black hole exteriors:

\begin{corollary}
    For smooth initial data $(g,k)$ on $\Sigma$ as in Figure~\ref{fig:cauchy},
    close to the initial data induced by the Schwarzschild de Sitter metric as in Theorem~\ref{thm:hv},
    there exists a solution to \eqref{eq:EVE} on $\mathcal{S}_1\cup\mathcal{R}\cup \mathcal{S}_2$,\footnote{Here $\mathcal{S}_i$, $i=1,2$ are diffeomorphic to the domains to the future of $\Sigma$ in Schwarzschild de Sitter, containing a small extension of the stationary black hole exteriors, uniformly in $r$; \cite[Section 3.5]{hi:vasy:stability}. } attaining the prescribed Cauchy data on $\Sigma$, which satisfies the asymptotic stability statement \eqref{eq:HV} in $\mathcal{S}_i$, $i=1,2$, and $(\mathcal{R},g)$ is global in the sense of Theorem~\ref{thm:fs}~(I). Moreover, the solution satisfies the orbital stability (towards $\mathcal{I}^+$) and asymptotic stability (towards $\iota^+_i$, $i=1,2$) statements of Theorem~\ref{thm:fs}~(II,III) with $\alpha_i>0$, $i=1,2$ as in \eqref{eq:HV}.
\end{corollary}

\begin{remark}[Exponential decay]
    The estimates \eqref{asym.g.decay.intro} and  \eqref{asym.Phi.decay.intro} show in particular that exponential decay established in \eqref{eq:HV} along $\mathcal{C}_i$, $i=1,2$, is inherited along every $\Sigma_s$ towards $\iota_i^+$, $i=1,2$.
    Exponential decay in Kerr de Sitter \emph{along} the cosmological horizon is well-established for small angular momentum in various settings: For the linear wave equation \cite{bony:sds,sd:quasi,sd:beyond,mavro:sds}, quasi-linear equations \cite{hi:quasi,hintz:vasy:16:global,hintz:kds:resonance,mavro:quasi}, and -- justifying our assumption here -- for the Einstein equations \cite{hi:vasy:stability,fang:kds:linear,fang:kds:nonlinear}.
\end{remark}

\begin{remark}[Reference metric]
The fact that it is possible to prove global existence with a \emph{fixed} reference metric $\widetilde{\g}$ means that on \emph{every} time slice $\Sigma_s$, $s\geq s_0$, the solution tends to $\g_{\K}$ with the \emph{same} parameters -- $(a_1,m_1)$ towards $\iota^+_1$, and $(a_2,m_2)$ towards $\iota^+_2$, see Fig~\ref{fig:cauchy}.
This is markedly different from the proofs of the major black hole stability theorems \cite{hi:vasy:stability,dhr:linear,ks:schwarzschild,dhrt:schwarzschild,kszeftel:elena,ks:kerr}: In the $\Lambda>0$ case, an iterative scheme for the linearised equations is implemented in \cite{hi:vasy:stability} to make successive gauge corrections and to find the parameters $(a,m)$ of the final state; in the case $\Lambda=0$ the modulation techniques  \cite{gcm:kerr,gcm:effective,shen:kerr} are used in \cite{ks:kerr} both to anchor the gauge, and to determine the parameters $(a,m)$ of the final state.
Heuristically, the reason that the parameters of the Kerr de Sitter metric remain unchanged in the cosmological region is that they are only relevant for convergence along \emph{spacelike} hypersurfaces --- which at their endpoints are unaffected by the perturbation.
\end{remark}

\begin{remark}[Functional degrees of freedom]
The main asymptotic degrees of freedom are \emph{functional} in nature:\footnote{This is already the case for the wave equation \cite{vasy:wave:ds,glw}, and gives rise to a scattering problem \cite{bernhardt}.}
The leading orders of  the solution $\g$ in $s$ are not captured by a member $\g_{\K}$ of the Kerr de Sitter family, but given by a free function, as in \eqref{asym.free}. \footnote{Yet \emph{locally} in the past of a given point on the conformal boundary the solution is asymptotically de Sitter. Cf.~discussions of the \emph{cosmic no hair} conjecture in \cite{hakan:nohair,schlue:weyl}; for a proof in \emph{spherical symmetry} see \cite{costa:nohair}. } These functions are free up to constraints which take a particularly simple form on the conformal boundary \cite{rendall:asymptotics, valiente:book}.
The significance of this freedom for the theory of gravitational radiation has been suggested by Ashtekar et al \cite{ashtekar:PRL}.
We remark, however, that the free functions $g_{ij}^\infty(x)$ in the expansion \eqref{asym.free} capture only part of the asymptotic data set on the conformal boundary. The remaining free functions appear in higher order expansions of the metric, in particular, at order $\mathcal{O}(e^{-Hs})$ for $g_{ij}(s,x)$; (see also Remark~\ref{rmk:expansion}).
In the setting of FLRW spacetimes this is obtained in \cite[Theorem 1.3]{fournodavlos:FLRW}.
More recently, Hintz-Vasy established a Fefferman-Graham type expansion in the present setting, and smoothness of the rescaled metric at the conformal boundary \cite{hintz:vasy:csm:24}.
We refer in this context also to the gluing and scattering constructions \cite{hintz:gluing,hintz:scattering,valiente-kroon:edgar}.
\end{remark}

\begin{remark}[Topology]
    A significant aspect of our theorem is that the topology of the spatial slices $\Sigma_s$ is $\mathbb{R}\times\mathbb{S}^2$ (and thus in the conformal diagram of Figure~\ref{fig:cauchy} the hypersurfaces $\Sigma_s$ extend to the end points $\iota_1^+$ and $\iota_2^+$, see Figure~\ref{fig:cylinder}). Indeed for \emph{compact} subsets $K\subset\Sigma_s$, with $s\geq s_0$ taken sufficiently large,
    a range of results in the literature imply that the domain of dependence $\mathcal{D}^+(K)$ is contained in a perturbation of \emph{de Sitter space}:\footnote{For a longer discussion of the spatially compact setting and its relation to earlier work on de Sitter see \cite[Section 1.6]{schlue:weyl}. The de Sitter solution $(\mathbb{H},h)$ can be realized as a hyperboloid $\mathbb{H}$ in $\mathbb{R}^{4+1}$, with metric $h=m\rvert_{\mathbb{H}}$ induced by the ambient Minkowski metric $m$. See also \cite{vasy:wave:ds} or \cite{schlue:optical}  where the embedding $\mathbb{H}\subset\mathbb{R}^{4+1}$ is used, and its geometric properties are further discussed.} 
    In particular the conformal method applies \cite{friedrich:desitter,valiente-kroon:marica}, and more generally, Ringstr\"om proved geodesic completeness from spatially bounded data irrespective of the topology \cite{ringstrom:invent}.
\end{remark}

\begin{remark}[Parabolic gauge] \label{rmk:gauge}
The gauge \eqref{lapse:intro} is crucial for the global existence proof in this setting and  leads to a foliation that correctly identifies $\mathcal{I}^+$. (In particular, no refoliation is necessary in the continuation argument for Theorem~\ref{thm:fs}.) A related problem occurs in the study of stable spacelike singularities, where related ADM gauges have been used to \emph{synchronize the singularity} \cite{rodspeck:bigbang,fournodavlos:rodspeck,fluk:kasner}; see also \cite{alexakis:greg}.
In the setting of expanding cosmologies, related gauges have been used in \cite{uggla:prl,fournodavlos:FLRW,ftodd}.
\end{remark}

\begin{remark}[Global Penrose diagram]
Since \cite{hi:vasy:stability} provides the stability of the regions $\mathcal{S}_i$ on a domain that extends beyond both the cosmological horizon \emph{and} the event horizon, the result of Hintz-Vasy can in principle be combined with a theorem of Dafermos-Luk on the $C^0$-stability of the Cauchy horizon \cite[Section 1.6]{dluk}.\footnote{For results in \emph{spherical symmetry} with $\Lambda>0$ in this region see \cite{costa:maxwell:I,costa:maxwell:II,costa:maxwell:III}.}
Further to Corollary~1.3, the combination of all three theorems shows in particular that the Penrose diagram in Fig.~\ref{fig:penrose} is dynamically stable.
\end{remark}

    

\begin{description}
    \item[Acknowledgements.] This paper resolves a problem that was first suggested by Mihalis Dafermos during Volker's Ph.D.~thesis. Over the years we have benefited from conversations with many people in the field, and we would like to thank (in alphabetical order) Spyros Alexakis, Abhay Ashtekar, Peter Hintz, Gustav Holzegel,  Jonathan Luk, Hans Ringstr\"om, Jared Speck, Jacques Smulevici, J\'er\'emie Szeftel, and Andr\'as Vasy for their time and interest in this project. We would also like to thank the IHP in Paris, MATRIX in Australia, and the Mittag-Leffler Institute in Sweden for their hospitality on several occasions. G.F.
gratefully acknowledges the support of the ERC starting grant 101078061 SINGinGR, under the European
Union’s Horizon Europe program for research and innovation. 
V.S.~is grateful for the support from the Humboldt Professorship of Gustav Holzegel.

\item[Recent developments.]
    Since the first appearance of this preprint,
    Hintz and Vasy have given an alternative proof of the stability of the cosmological region in generalised harmonic gauge \cite{hintz:vasy:csm:24}, allowing in particular a localisation near timelike infinity.
\end{description}

\section{Covariant \textsc{ADM} formulation of the Einstein equations in parabolic gauge}
\label{sec:eqs}

In this section we derive the Einstein equations for a decomposition of a $1+3$-dimensional Lorentzian manifold $(\mathcal{M},\g)$ with respect to a time-function $s$ which satisfies a parabolic gauge, see \eqref{lapse:intro}.
For a given reference metric $\widetilde{\g}$, this gauge equates the deviation of the corresponding lapse functions to the deviation in the mean curvatures of the leaves $\Sigma_s$ of the foliation.
The Einstein equations then become a first order system for the first and second fundamental forms of $\Sigma_s$, coupled to a parabolic equation for the lapse function.

\subsection{Preliminaries of the \textsc{ADM} decomposition}
\label{sec:ADM}

We recall various geometric notions in the ADM formalism, all defined locally with respect to a differentiable time-function $s$; for an introduction to this formulation see for instance \cite{Christ-notes}.

Given any time-function $s$ (namely a differentiable function $s$ on $\mathcal{M}$ with the property that $V(s)>0$ for any future-directed timelike vectorfield $V$), we denote by $\Sigma_s$ the level sets of $s$  (a family of spacelike hypersurfaces in $\mathcal{M}$), and by $T$ the timelike vectorfield colinear to the normal $N$ defined by $T^\mu=-\g^{\mu\nu}\partial_\nu s$. The \emph{lapse function} $\Phi$ is defined $\Phi^{-2}=-\g(T,T)$, and the unit normal to $\Sigma_s$ is then given by $N=\Phi T$.
    For any choice of coordinates $(x^1,x^2,x^3)$ on a given level set $\Sigma_{s_0}$,
    we can assign to any point $p\in \Sigma_s$ the coordinates $(s,x^1,x^2,x^3)$ if $p=\psi_{t}(q)$, $t=s-s_0$, where $\psi_{t}$ is the $1$-parameter group of diffeomorphism generated by $T$, and $q\in \Sigma_{s_0}$ has coordinates $(x^1,x^2,x^3)$. In these coordinates, the metric takes the form:\footnote{As usual, we use the summation convention and Latin indices range in $\{1,2,3\}$, while Greek indices range in $\{0,1,2,3\}$. }
\begin{equation}\label{metric}
\g=-\Phi^2 \ud s^2+ {g}_{ij}\ud x^i \ud x^j
\end{equation}

We denote the first and second fundamental forms of $\Sigma_s$ by $g_s$, and $k_s$, respectively, and usually suppress the subscript. They are defined by  
    \begin{equation}
  \label{def:ff}
  (g_s)_p=\g_p\Bigr\rvert_{\mathrm{T}_p\Sigma_s}\,,
  \qquad   (k_s)_p(X,Y)=\g_p(\nabla_X N,Y)\,,\qquad X,Y\in\mathrm{T}_p\Sigma_s, \quad p\in\mathcal{M}.
\end{equation}
Here $N=\Phi T$ is the unit normal, 
and we have a coordinate frame $E_i=\frac{\partial}{\partial x^i}$
which is Lie transported by $T=\frac{\partial}{\partial s}$: $[T,E_i]=0$.
In the frame $(E_0=N,E_1,E_2,E_3)$ the metric components are
\begin{equation}
    \g_{00}=-1\,,\qquad \g_{0i}=0\,,\qquad \g_{ij}=g_{ij}\,.
\end{equation}

The first variation formula is
\begin{equation}
    \label{eq:firstvariation}
  \frac{\partial g_{ij}}{\partial s}=2\Phi\, k_{ij}\,,
\end{equation}
and we will express similarly the second variation equation for $\partial_s k_{ij}$ as well as the Gauss-Codazzi equations of the embedding of $\Sigma_s$ in $\mathcal{M}$ in this frame.

While all spacetime quantities are set in bold, we print $\Sigma_s$-tangent tensors in standard font.
For instance, while the Riemann curvature of $(\mathcal{M},\g)$ is $\bm{R}_{\alpha\beta\mu\nu}$,
the components of the Riemann curvature of $(\Sigma_s,g_s)$ are $R_{mnij}$.
The  Einstein vacuum equations are
\begin{equation}
    \label{EVE}
    \RIC[\g]=\Lambda \g\,.
\end{equation}
While we denote the Ricci curvature of $\g$ by $\RIC[\g]$,
and the Ricci curvature of $g$ by $\Ric[g]$,
we often denote the components of the Ricci curvature simply by:
\begin{equation}
    \label{eq:Ric:short}
    \Ric_{ij}=g^{mn}R_{minj}
\end{equation}
Similarly for the Levi-Civita connections of $g,\g$: We denote by $\nabla$ the connection induced by $\Nabla$ on $\Sigma_s$.

The second variation formula is
\begin{equation}
  \label{eq:secondvariationalformula}
  \frac{\partial k_{ij}}{\partial s}=\nabla_i\nabla_j\Phi+\Phi\Bigl\{-\R_{i0j0}+k_i^m k_{mj}\Bigr\}\,.
\end{equation}
The aim is to obtain a closed system of evolution equations for $g_s$, and $k_s$,
and for that purpose we first eliminate the curvature component $\R_{0i0j}$ in \eqref{eq:secondvariationalformula}  using the Einstein equations \eqref{EVE}.

The \textbf{Codazzi equations} are:
\begin{equation}
  \nabla_i k_{jm}-\nabla_j k_{im}=\R_{m0ij}\,,
\end{equation}
and upon contracting and using \eqref{EVE} we obtain:
\begin{equation}
  \label{eq:perppara}
  \nabla^i k_{ji}-\nabla_j\tr k=\RIC_{0j}[\g]=0
\end{equation}

The  \textbf{Gauss equation} reads:
\begin{equation}
  \label{eq:gauss}
  R_{minj}+k_{mn}k_{ij}-k_{mj}k_{ni}=\R_{minj}
\end{equation}
A first contraction yields
\begin{equation}
  \label{eq:gauss:contracted}
  \Ric_{ij}[g]+\tr k\:k_{ij}-k^n{}_{i}\,k_{nj}=\R_{0i0j}+\RIC_{ij}[\g]=\R_{0i0j}+\Lambda g_{ij}\,,
\end{equation}
which gives a formula for $\R_{i0j0}$
which we may substitute into the second variation formula \eqref{eq:secondvariationalformula}:
\begin{equation}
  \frac{\partial k_{ij}}{\partial s}=\nabla_i\nabla_j\Phi-\Phi\Bigl\{\Ric_{ij}[g]+\tr k\, k_{ij}-2k_i{}^mk_{mj}-\RIC_{ij}[\g]\Bigr\}\label{eq:parapara:RIC}
    \end{equation}
A second contraction of \eqref{eq:gauss:contracted} gives:
\begin{equation}
  \label{eq:perpperp}
  \tr_g\Ric+(\tr k)^2-\lvert k\rvert^2=2\RIC_{00}[\g]+{\bm{\mathrm{tr}}}_{\g}\RIC =2\Lambda\,,
\end{equation}
where $R=\tr_g \Ric$ is the scalar curvature  of $g$.
This is  the \emph{Hamiltonian constraint}:
\begin{equation}
\label{Ham.const} R-\lvert k\rvert^2+(\tr k)^2=\,2\Lambda
\end{equation}
Together with \eqref{eq:perppara},
which is also referred to as the \emph{momentum constraint}
\begin{equation}
\label{mom.const} \div_g k - \ud \tr_g k=\,0 \,,   
\end{equation}
these are the \emph{constraint equations} for the first and second fundamental form,
complementing the \emph{evolution equations} \eqref{eq:firstvariation} for $g$ and \eqref{eq:parapara:RIC} for $k$. Combining \eqref{eq:secondvariationalformula}, \eqref{eq:perpperp}, we also have the following evolution equation for the mean curvature:
\begin{align}\label{eq:trk}
\partial_s\mathrm{tr}k=\Delta_g\Phi-\Phi |k|^2-\Phi {\bf Ric}_{00}[{\boldsymbol g}].
\end{align}

For future reference we also record the formula for the Riemann curvature of $g$ in local coordinates:
\begin{equation}
R^a{}_{bmn}=\partial_m\Gamma^a_{nb}-\partial_n\Gamma^a_{mb}+\Gamma^a_{mc}\Gamma^c_{nb}-\Gamma^a_{nc}\Gamma^c_{mb}\,,
\end{equation}
where
\begin{equation} \label{def:Gamma}
    \Gamma_{ic}^a=\frac{1}{2}g^{ab}(\partial_ig_{cb}+\partial_cg_{ib}-\partial_bg_{ic})\,.
\end{equation}
The curvature satisfies the  cyclic identity:
\begin{equation}
  \label{eq:cyclic:identity:components}
R^a{}_{bmn}+R^a{}_{mnb}+R^a{}_{nbm}=0\,,\quad\text{where }R^a{}_{bmn}=g^{ac}R_{cnbm}\,,
\end{equation}
together with the symmetries
$R^a_{\phantom{a}bnm}=-R^a_{\phantom{a}bmn}\,,\ R_{bamn}=-R_{abmn}\,,$
this implies the pair symmetry:
\begin{equation}
  \label{eq:curvature:symmetry:pair:components}
  R_{mnab}=R_{abmn}
\end{equation}
Finally, we have in local coordinates that
\begin{equation}
\label{eq:Ricci:local}
\Ric_{mn}=R^a_{\phantom{a}man}=\partial_a\Gamma^a_{nm}-\partial_n\Gamma^a_{am}+\Gamma_{ac}^a  \Gamma_{nm}^c-\Gamma_{nc}^a \Gamma_{am}^c\,.
\end{equation}
The analogous formulas are valid for the Riemann curvature of $\g$.

\subsection{System of evolution equations in parabolic gauge}
\label{sec:evoleqs}

%
We have already encountered the first variation equation \eqref{eq:firstvariation} for $g_s$,
which can also be expressed as an equation for the components of $g_s^{-1}$:
\begin{subequations}
\begin{align}\label{1stvar}
\partial_s {g}_{ij}=&\,2\Phi k_{ij},\\
\label{1stvar.inv}\partial_sg^{ij}=&-2\Phi k^{ij},
\end{align}
\end{subequations}
where $k^{ij}=g^{im}g^{jn}k_{mn}$.
Moreover we write the second variation equation \eqref{eq:parapara:RIC} as 
\begin{equation}
    \label{2ndvar}\partial_sk_{ij}=\, {\nabla}_i {\nabla}_j\Phi-\Phi( {\Ric}_{ij}+k_{ij}k_l{}^l-2{k_i}^lk_{jl})
+\Phi\Lambda g_{ij}\,.
\end{equation}
In fact, it is convenient to 
view the second fundamental form as a $(1,1)$ tensor and work with $k_i{}^j=g^{cj}k_{ic}$ instead of $k_{ij}$.
  Then the equations \eqref{1stvar}, \eqref{2ndvar} become:
\begin{subequations}
\begin{align}\label{1stvar.2}
\partial_s {g}_{ij}=&\,2\Phi g_{ja}k_i{}^a\\
\label{1stvar.inv.2}\partial_sg^{ij}=&-2\Phi g^{ia}k_a{}^j\\
\label{2ndvar.2}\partial_sk_i{}^j+\Phi k_l{}^lk_i{}^j=&\, \nabla_i \nabla^j\Phi-\Phi \Ric_i{}^j+\Phi\Lambda\delta_i{}^j
\end{align}
\end{subequations}
In Section~\ref{sec:ric} below, the Ricci curvature $\Ric_i{}^j$ is suitably expressed in terms of the Christoffel symbols $\Gamma_{ic}^a$.
We are led to consider, in addition to the \eqref{1stvar.2}  and \eqref{2ndvar.2}, the following evolution equation for $\Gamma_s$:
\begin{align}\label{Gamma.eq}
\partial_s\Gamma_{ic}^a=\nabla_i(\Phi k_c{}^a)
+\nabla_c(\Phi k_i{}^a)-g^{ab}g_{cj}\nabla_b(\Phi k_i{}^j)
\end{align}
This is an immediate consequence of \eqref{1stvar.2}, \eqref{1stvar.inv.2} and the formula \eqref{def:Gamma}:
\begin{equation}
\begin{split}
    \partial_s\Gamma^a_{ic}=&-2\Phi k_d{}^a\Gamma^d_{ic}+g^{ab}\bigl(\partial_i(\Phi k_{cb})+\partial_c(\Phi k_{ib})-\partial_b(\Phi k_{ic})\bigr)\\
    =&\,g^{ab}\bigl(\nabla_i(\Phi k_{cb})+\nabla_c(\Phi k_{ib})-\nabla_b(\Phi k_{ic})\bigr)
\end{split}
\end{equation}

\subsubsection{Reference metric and gauge}

To set up the stability problem, 
we will make the choice of a reference metric in Section~\ref{sec:background}:
\begin{align}\label{metric.ref}
\widetilde{\g} =-\widetilde{\Phi}^2\ud s^2+ \widetilde{g}_{ij}\ud x^i\ud x^j,
\end{align}
defined on the same differentiable manifold $\mathcal{M}$,
and denote by $\widetilde{\RIC}_{\mu\nu}=\RIC[\widetilde{\g}]_{\mu\nu}$ the components of the Ricci curvature of $\widetilde{\g}$.
Also, we denote by $\widetilde{\nabla},\widetilde{\Gamma}_{ij}^a,\widetilde{\Ric}_{ij}=\Ric[\tg]_{ij}$ the Levi-Civita connection, Christoffel symbols, and Ricci curvature associated to $\tg$. 

Define 
\begin{align}\label{diff}
\begin{split}
\widehat\Phi=\Phi-\widetilde\Phi,\qquad \widehat{g}_{ij}=g_{ij}-\widetilde{g}_{ij},\qquad \widehat{g}^{ij}=g^{ij}-\widetilde{g}^{ij},\\
\widehat\nabla=\nabla-\widetilde{\nabla},\qquad\widehat\Gamma_{ij}^a=\Gamma_{ij}^a-\widetilde{\Gamma}_{ij}^a,\qquad\widehat{k}_i{}^j=k_i{}^j-\widetilde{k}_i{}^j,
\end{split}
\end{align}
where  
\begin{equation}
    \widetilde{k}_i{}^j=\widetilde{g}^{ja}\widetilde{k}_{ia}=\frac{1}{2}\widetilde \Phi^{-1}\widetilde g^{ja}\partial_s\widetilde{g}_{ia}\,.
\end{equation}

The remaining gauge freedom is the choice of the time-function $s$.
This is the choice of a lapse function,
and in this work we set
\begin{align}
\Phi-\widetilde{\Phi}=k_l{}^l-\widetilde{k}_l{}^l\,,
\end{align}
or equivalently 
\begin{equation}
    \label{lapse}
    \widehat\Phi=\widehat{k}_l{}^l\,.
\end{equation}

\begin{remark}[Geometric interpretation of differences]
In the definition of the hatted quantities \eqref{diff}, a common identification of the coordinate vectorfields is made.
Strictly speaking, $\g$ and $\widetilde{\g}$ are metrics on different manifolds. However, they are expressed in the same chart.
If we view $(s,x^1=t,x^2,x^3)$ as standard coordinates for $\widetilde{\g}$,
then $-\widetilde{\Phi}^2$, and $\widetilde{g}_{ij}$ refer to the components of $\widetilde{\g}$ with respect to the coordinate basis $(\partial_s, \partial_{x^i})$ in the chart.
However, $-\Phi^2$, and $g_{ij}$ are the components of the metric $\g$
with respect to the basis $(\ud\psi\cdot \partial_s,\ud\psi\cdot\partial_{x^i})$,
where $\psi:(s,x)\mapsto \psi(s,x)\in\mathcal{R}$ is a diffeomorphism, a chart for $(\mathcal{R},\g)$. These are the coordinate vectorfields in the given chart and typically also denoted by $(\partial_s,\partial_{x^i})$. This allows us to view $g_{ij}$ and $\widetilde{g}_{ij}$ as functions of $(s,x^1,x^2,x^3)$ in the same chart.
In other words, the hatted quantities can be viewed as the components of $\psi^\ast \g-\widetilde{\g}$ is the standard coordinate chart of $\widetilde{\g}$.
In particular, with this identification both $g_{ij}$ and $\widetilde{g}_{ij}$ are $\Sigma_s$-tensors.
\end{remark}

\begin{remark}[Contractions with the metric]
The hats in \eqref{diff} do not commute with the metric. For example, $g_{aj}\widehat{k}_i{}^j=\widehat{k}_{ij}\neq k_{ij}-\widetilde{k}_{ij}$, since the raising/lowering of indices of the tilde variables is performed with respect to $\widetilde{g}$. To avoid confusion, we will not change the type of the tensors with hats, that is to say, we will always treat $\widehat{k}$ as a $(1,1)$ tensor, $\widehat{g}$ as a $(0,2)$ etc. 
\end{remark}

\begin{remark}[Elliptic gauge choices]
    A \emph{maximal gauge}, where each level set of the time function has zero mean curvature,
 \begin{equation}
     \label{maximal}
     \tr_g k=0\,,
 \end{equation}
    leads to an \emph{elliptic} equation for the lapse function.
    Here, the choice \eqref{lapse}
leads to a \emph{parabolic} equation for the lapse, which is well-posed in the future direction; see \eqref{Phi.hat.eq} below.
\end{remark}

\begin{remark}[Identity for differences of Christoffel symbols]
Recall that the difference of Christoffel symbols is a $(1,2)$ tensor:
\begin{align}
\label{Gamma.diff}
\widehat{\Gamma}_{ij}^a=\frac{1}{2}\widetilde{g}^{ab}(\nabla_i\widehat{g}_{jb}+\nabla_j\widehat{g}_{ib}-\nabla_b\widehat{g}_{ij})
\end{align}
Also, note that
\begin{align}\label{Gamma.g.rel}
\widetilde{g}_{ac}\widehat\Gamma_{ij}^c+\widetilde{g}_{jc}\widehat\Gamma_{ia}^c=\nabla_i\widehat g_{ja}\,.
\end{align}
\end{remark}

\subsubsection{First and second variation equations for differences}
\label{sec:vareq}

Given a reference metric, we first derive the first variation equations for the differences $\hg_{ij}$.

\begin{lemma}[First variation equations]\label{lem:hat.eq:1st}
The variables $\widehat{g}_{ij},\widehat{g}^{ij}$ satisfy the evolution equations:
\begin{subequations}
\begin{align}
\label{g.hat.eq}\partial_s\widehat{g}_{ij}-2H\widehat{g}_{ij}=&\,2\Phi g_{ja}\widehat{k}_i{}^a+
2\widehat{\Phi}g_{ja}\widetilde{k}_i{}^a
+2\widetilde{\Phi}(\widetilde{k}_i{}^a-H\widetilde{\Phi}^{-1}\delta_i{}^a)\widehat{g}_{ja}\,,\\  \label{g.inv.hat.eq}\partial_s\widehat{g}^{ij}+2H\widehat{g}^{ij}=&-2\Phi g^{ia}\widehat{k}_a{}^j
-2\widehat{\Phi}g^{ia}\widetilde{k}_a{}^j
-2\widetilde{\Phi}(\widetilde{k}_a{}^j-H\widetilde{\Phi}^{-1}\delta_a{}^j)\widehat{g}^{ia}\,,
\end{align}
\end{subequations}
where $H=\sqrt{\frac{\Lambda}{3}}$.
Moreover
\begin{align}
    \label{Gamma.hat.eq}\partial_s\widehat{\Gamma}_{ic}^a=&\,\Phi\nabla_i\widehat{k}_c{}^a
+\Phi\nabla_c\widehat{k}_i{}^a-g^{ab}g_{cj}\Phi\nabla_b\widehat{k}_i{}^j+\mathfrak{G}_{ic}^a\,,\\
\label{frak.G}\mathfrak{G}_{ic}^a=&\,\Phi\widehat\nabla_i\widetilde k_c{}^a
+\Phi\widehat\nabla_c\widetilde k_i{}^a-g^{ab}g_{cj}\Phi\widehat\nabla_b \widetilde k_i{}^j\\
\notag&+\widehat\Phi\widetilde\nabla_i\widetilde k_c{}^a
+\widehat\Phi\widetilde\nabla_c\widetilde k_i{}^a-(\widehat g^{ab}g_{cj}\Phi
+\widetilde g^{ab}\widehat g_{cj}\Phi
+\widetilde g^{ab}\widetilde g_{cj}\widehat\Phi)\widetilde\nabla_b \widetilde k_i{}^j\\
\notag&+k_c{}^a\nabla_i\widehat\Phi
+k_i{}^a\nabla_c\widehat\Phi-g^{ab}g_{cj}k_i{}^j\nabla_b\widehat\Phi\\
\notag&+\widehat k_c{}^a\widetilde\nabla_i\widetilde\Phi
+\widehat k_i{}^a\widetilde\nabla_c\widetilde\Phi-(\widehat g^{ab}g_{cj}k_i{}^j+\widetilde g^{ab}\widehat g_{cj}k_i{}^j+\widetilde g^{ab}\widetilde g_{cj}\widehat k_i{}^j)\widetilde\nabla_b\widetilde\Phi\,.
\end{align}
\end{lemma}
\begin{proof}
To derive the equations \eqref{g.hat.eq}, \eqref{g.inv.hat.eq}, \eqref{Gamma.hat.eq}, we use the fact that the corresponding variables of the reference metric satisfy the equations \eqref{1stvar.2}, \eqref{1stvar.inv.2}, \eqref{Gamma.eq}, and subtract them from the equations satisfied by $g_{ij},g^{ij},\Gamma_{ic}^a$. The computations are straightforward. 
\end{proof}

While the first variation formulas do not rely on the gauge condition, it is used in
the following derivation of the second variation equation for $\hk_i{}^j$  which shows that \eqref{lapse} is a parabolic gauge. 
The derivation uses a specific expression for the Ricci curvature, which we present first.

\begin{lemma}[Ricci curvature]
\label{sec:ric}
The Ricci curvature of $g$ can be expressed in the form
\begin{align}\label{Ricci.3}
\begin{split}
\Ric_i{}^j=&
\frac{1}{3}g^{cj}(\nabla_a\Gamma_{ci}^a-\nabla_c\Gamma_{ia}^a)
+\frac{2}{3}g^{ab}(\nabla_i\Gamma^j_{ab}-\nabla_a\Gamma^j_{bi})\\
&+\frac{1}{3}g^{cj}(\Gamma^a_{cb}\Gamma^b_{ai}-\Gamma^a_{ab}\Gamma_{ci}^b)
+\frac{2}{3}g^{ab}(\Gamma^c_{ib}\Gamma_{ac}^j-\Gamma^c_{ab}\Gamma^j_{ci})\,.
\end{split}
\end{align}
Here $\nabla \Gamma$ is interpreted tensorially, e.g. $$\nabla_a\Gamma_{ji}^a:=\partial_a\Gamma_{ji}^a+\Gamma^a_{ab}\Gamma^b_{ij}-\Gamma_{aj}^b\Gamma_{bi}^a-\Gamma_{ai}^b\Gamma_{jb}^a\,.$$  
\end{lemma}

\begin{proof}
 Starting from the expression \eqref{eq:Ricci:local}
we can expand  the  expression for the Ricci curvature:
\begin{align}\label{Ricci}
\Ric_i{}^j=&\,g^{cj}(\partial_a\Gamma^a_{ci}-\partial_c\Gamma^a_{ia}+\Gamma^a_{ab}\Gamma_{ci}^b-\Gamma^a_{cb}\Gamma^b_{ai})\\
\notag=&\,g^{cj}(\nabla_a\Gamma_{ci}^a-\nabla_c\Gamma_{ia}^a+\Gamma^a_{cb}\Gamma^b_{ai}-\Gamma^a_{ab}\Gamma_{ci}^b)
\end{align}
Alternatively, we write
using the pair symmetry of the curvature tensor:
\begin{align}\label{Ricci.2}     
\notag \Ric_i{}^j=&R^{b\phantom{ib}j}_{\phantom{b}ib}=g^{ba}R_{aib}^{\phantom{aib}j}=g^{ab}R_{b\phantom{j}ai}^{\phantom{b}j}\\
    =&g^{ab}\bigl(\partial_i \Gamma_{ab}^j -\partial_a \Gamma_{ib}^j+\Gamma_{ic}^j \Gamma^c_{ab}-\Gamma_{ac}^j \Gamma^c_{ib}\bigr)\\
\notag=&\,g^{ab}(\nabla_i\Gamma^j_{ab}-\nabla_a\Gamma^j_{bi}+\Gamma^c_{ib}\Gamma_{ac}^j-\Gamma^c_{ab}\Gamma^j_{ci})
\end{align}
Combining \eqref{Ricci} and \eqref{Ricci.2} gives \eqref{Ricci.3}.
\end{proof}

The motivation for these manipulations will be discussed in Remark~\ref{rem:symm}. 
In fact, we have already seen in \eqref{Gamma.hat.eq} that the evolution equation for $\widehat{\Gamma}$ contains derivatives of $\widehat{k}$.
Now returning to \eqref{2ndvar.2},
we will see that the evolution equation for $\widehat{k}$ can be expressed 
in terms of derivatives of $\widehat{\Gamma}$,
in such a way that in the derivation of the energy estimate, which couples these evolution equations, several terms cancel. This uses the formulas for the Ricci curvature in Lemma~\ref{sec:ric}.

\begin{lemma}[Second variation equations]
\label{lem:hat.2nd}
The variables $\widehat{\Phi},\widehat{k}_i{}^j$ satisfy the evolution equations:
  \begin{align}
\label{k.hat.eq}\partial_s\widehat{k}_i{}^j+3H\widehat{k}_i{}^j=&
\,g^{cj}\nabla_i\nabla_c\widehat\Phi-\widetilde{\Phi}(\widetilde{k}_l{}^l-3H\widetilde{\Phi}^{-1})\widehat{k}_i{}^j+\mathfrak{K}_i{}^j+(\widetilde{\mathfrak{I}}_k)_i{}^j\\
\notag&+\frac{1}{3}\Phi g^{cj}(\nabla_c\widehat\Gamma_{ia}^a-\nabla_a\widehat\Gamma_{ci}^a)
+\frac{2}{3}\Phi g^{ab}(\nabla_a\widehat\Gamma^j_{bi}-\nabla_i\widehat\Gamma^j_{ab})\,,
  \end{align}  
  and
  \begin{equation}
\label{Phi.hat.eq}\partial_s\widehat{\Phi}-\Delta_g\widehat{\Phi}+2H\widehat{\Phi}=\mathfrak{F}+\widetilde{\mathfrak{I}}_\Phi\,,
  \end{equation}
where $H=\sqrt{\frac{\Lambda}{3}}$, and
\begin{align}
\label{frak.K}\mathfrak{K}_i{}^j=&-\widehat{\Phi}k_l{}^lk_i{}^j-\widetilde{\Phi}\widehat{\Phi}k_i{}^j+\Lambda\delta_i{}^j\widehat{\Phi}
+\widehat{g}^{cj}\partial_i\partial_c\widetilde{\Phi}
-\widehat{g}^{cj}\widetilde{\Gamma}_{ic}^a\partial_a\widetilde{\Phi}-g^{cj}\widehat{\Gamma}_{ic}^a\partial_a\widetilde{\Phi}\\
\notag&+\frac{1}{3}\Phi g^{cj}(\widehat{\nabla}_c\widetilde\Gamma_{ia}^a-\widehat\nabla_a\widetilde\Gamma_{ci}^a)
+\frac{2}{3}\Phi g^{ab}(\widehat\nabla_a\widetilde\Gamma^j_{bi}-\widehat\nabla_i\widetilde\Gamma^j_{ab})\\
\notag&+\frac{1}{3}(\widehat{\Phi} g^{cj}+\widetilde{\Phi}\widehat{g}^{cj})(\widetilde{\nabla}_c\widetilde\Gamma_{ia}^a-\widetilde\nabla_a\widetilde\Gamma_{ci}^a)
+\frac{2}{3}(\widehat{\Phi} g^{ab}+\widetilde{\Phi}\widehat{g}^{ab})(\widetilde\nabla_a\widetilde\Gamma^j_{bi}-\widetilde\nabla_i\widetilde\Gamma^j_{ab})\\
\notag&+\frac{1}{3}\Phi g^{cj}(\widehat\Gamma^a_{ab}\Gamma_{ci}^b-\widehat\Gamma^a_{cb}\Gamma^b_{ai})+\frac{1}{3}\Phi g^{cj}(\widetilde\Gamma^a_{ab}\widehat\Gamma_{ci}^b-\widetilde\Gamma^a_{cb}\widehat\Gamma^b_{ai})\\
\notag&+\frac{1}{3}(\widehat\Phi g^{cj}+\widetilde\Phi\widehat{g}^{cj})(\widetilde\Gamma^a_{ab}\widetilde\Gamma_{ci}^b-\widetilde\Gamma^a_{cb}\widetilde\Gamma^b_{ai})
+\frac{2}{3}\Phi g^{ab}(\widehat\Gamma^c_{ab}\Gamma^j_{ci}-\widehat\Gamma^c_{ib}\Gamma_{ac}^j)\\
\notag&+\frac{2}{3}\Phi g^{ab}(\widetilde\Gamma^c_{ab}\widehat\Gamma^j_{ci}-\widetilde\Gamma^c_{ib}\widehat\Gamma_{ac}^j)
+\frac{2}{3}(\widehat\Phi g^{ab}+\widetilde\Phi\widehat{g}^{ab})(\widetilde\Gamma^c_{ab}\widetilde\Gamma^j_{ci}-\widetilde\Gamma^c_{ib}\widetilde\Gamma_{ac}^j)\,,
\\
\label{frak.F}\mathfrak{F}=&-2\widehat\Phi \widehat{k}_i{}^j\widetilde{k}_j{}^i
-\Phi \widehat{k}_i{}^j\widehat{k}_j{}^i
-\widehat{\Phi}(\widetilde{k}_i{}^j\widetilde{k}_j{}^i-\Lambda)-\widetilde\Phi \widehat{k}_i{}^j(\widetilde{k}_j{}^i
-H\widetilde{\Phi}^{-1}\delta_j{}^i)\\
\notag&
+\widehat{g}^{ab}\partial_a\partial_b\widetilde{\Phi}
-g^{ab}\widehat{\Gamma}_{ab}^c\partial_c\widetilde{\Phi}
-\widehat{g}^{ab}\widetilde{\Gamma}_{ab}^c\partial_c\widetilde{\Phi}\,.
\end{align}
The terms
$(\widetilde{\mathfrak{I}}_k)_i{}^j,\widetilde{\mathfrak{I}}_\Phi$
only contain variables of the reference metric $\widetilde{\g}$ and are equal to:
\begin{equation}
    \label{eq:J:k.Phi}
    (\widetilde{\mathfrak{I}}_k)_i{}^j=-\widetilde{\Phi}\Bigl(\widetilde{\RIC}_i{}^j-\Lambda \delta_i{}^j\Bigr),\qquad
\widetilde{\mathfrak{I}}_\Phi=\tPhi\Bigl(\widetilde{\RIC}_{00}+\Lambda\Bigr)\,.
\end{equation}
\end{lemma}

\begin{proof}
    For \eqref{k.hat.eq}, recall that $\widetilde{\g}$ is not an exact solution of the Einstein equations.
The starting point here is \eqref{eq:parapara:RIC},
which we can write as:
\begin{equation}\label{k.ref.eq}
\begin{split}
\partial_s\widetilde k_i{}^j+\widetilde\Phi \widetilde k_l{}^l\widetilde k_i{}^j=&\, \widetilde{g}^{cj}\widetilde\nabla_i \widetilde\nabla_c\widetilde\Phi-\widetilde\Phi  \Ric[\tg]_i{}^j
+\widetilde{\Phi}\RIC[\widetilde{\g}]_i{}^j\\
=&\widetilde{g}^{cj}\widetilde\nabla_i \widetilde\nabla_c\widetilde\Phi-\widetilde\Phi \widetilde \Ric_i{}^j + \widetilde{\Phi} \Lambda \delta_i{}^j -(\widetilde{\mathfrak{I}}_k)_i{}^j
\end{split}
\end{equation}
Subtracting it from \eqref{2ndvar.2} and using \eqref{lapse} then results in 
\begin{multline}
\partial_s\widehat{k}_i{}^j+3H\widehat{k}_i{}^j+\hPhi k_l{}^l k_i{}^j+\tPhi\hPhi k_i{}^j+\tPhi\bigl(\tk_l{}^l-3H\tPhi^{-1}\bigr)\hk_i{}^j =\\
    =\nabla_i\nabla^j\hPhi+\hg^{jc}\tnabla_{i}\partial_c\tPhi+g^{jc}\hGamma_{ic}^b\partial_b\tPhi
    -\Phi \Ric_i{}^j+\tPhi\widetilde{\Ric}_i{}^j+\hPhi\Lambda\delta_i{}^j+(\widetilde{\mathfrak{I}}_k)_i{}^j\,.
\end{multline}
This already accounts for all terms in the first line of \eqref{k.hat.eq} together with the first line in \eqref{frak.K}.
It remains to compute the difference of the Ricci curvatures.
In view of \eqref{Ricci.3} we have:
\begin{equation}
    \begin{split}
        \Phi \Ric_i{}^j-\tPhi\widetilde{\Ric}_i{}^j=&\phantom{+}\frac{1}{3}\Phi g^{cj}(\nabla_a\hGamma_{ci}^a-\nabla_c\hGamma_{ia}^a)
+\frac{2}{3} \Phi g^{ab}(\nabla_i\hGamma^j_{ab}-\nabla_a\hGamma^j_{bi})\\
&+\frac{1}{3}\Phi g^{cj}(\hnabla_a\tGamma_{ci}^a-\hnabla_c\tGamma_{ia}^a)
+\frac{2}{3} \Phi g^{ab}(\hnabla_i\tGamma^j_{ab}-\hnabla_a\tGamma^j_{bi})\\
&+\frac{1}{3} \Phi g^{cj}(\hGamma^a_{cb}\Gamma^b_{ai}-\hGamma^a_{ab}\Gamma_{ci}^b)
+\frac{2}{3} \Phi g^{ab}(\hGamma^c_{ib}\Gamma_{ac}^j-\hGamma^c_{ab}\Gamma^j_{ci})\\
&+\frac{1}{3} \Phi g^{cj}(\tGamma^a_{cb}\hGamma^b_{ai}-\tGamma^a_{ab}\hGamma_{ci}^b)
+\frac{2}{3} \Phi g^{ab}(\tGamma^c_{ib}\hGamma_{ac}^j-\tGamma^c_{ab}\hGamma^j_{ci})\\
&+\frac{1}{3}\bigl(\hPhi \tg^{cj}+\Phi\hg^{cj}\bigr)(\tnabla_a\tGamma_{ci}^a-\tnabla_c\tGamma_{ia}^a)
+\frac{2}{3} \bigl(\hPhi \tg^{ab}+\Phi \hg^{ab}\bigr)(\tnabla_i\tGamma^j_{ab}-\tnabla_a\tGamma^j_{bi})\\
&+\frac{1}{3} \bigl(\hPhi \tg^{cj}+\Phi \hg^{cj}\bigr)(\tGamma^a_{cb}\tGamma^b_{ai}-\tGamma^a_{ab}\tGamma_{ci}^b)
+\frac{2}{3} \bigl(\hPhi \tg^{ab}+\Phi \hg^{ab}\bigr)(\tGamma^c_{ib}\tGamma_{ac}^j-\tGamma^c_{ab}\tGamma^j_{ci})
    \end{split}
\end{equation}

For the equation \eqref{Phi.hat.eq}, we consider \eqref{eq:trk} for a solution to \eqref{eq:EVE}:
\begin{equation}\label{trk.eq}
\partial_sk_l{}^l=\Delta_g\Phi+\Lambda\Phi-\Phi |k|^2
\end{equation}
The corresponding equation for the mean curvature of the reference metric is found by contracting \eqref{k.ref.eq}: 
\begin{equation}    \partial_s\tk_l{}^l+\tPhi\bigl(\tk_l{}^l\bigr)^2=\widetilde{\Delta}\tPhi-\tPhi\widetilde{R}+3\Lambda\tPhi-(\widetilde{\mathfrak{J}}_k)_l{}^l
\end{equation}
and using the twice contracted Gauss equation:
\begin{equation}
    \widetilde{R}+(\widetilde{k}_l{}^l)^2-\tk_i{}^j\tk_j{}^i=2\widetilde{\RIC}_{00}+\widetilde{\R}=2\Lambda+\tPhi^{-1}\widetilde{\mathfrak{J}}_\Phi-\tPhi^{-1}(\widetilde{\mathfrak{J}}_k)_j{}^j
\end{equation}
%
%
Thus we have:
\begin{equation}
    \label{trk.ref.eq}
    \partial_s\widetilde k_l{}^l=\widetilde\Delta\widetilde\Phi-\widetilde\Phi \widetilde k_i{}^j\widetilde k_j{}^i+ \Lambda \tPhi-\widetilde{\mathfrak{J}}_\Phi
\end{equation}
 Subtracting \eqref{trk.ref.eq} from \eqref{trk.eq} gives
 \begin{equation}
     \partial_s \hk_l{}^l=\Delta\Phi-\widetilde{\Delta}\tPhi-\Phi k_{i}{}^j k_j{}^i+\tPhi \tk_{i}{}^j \tk_j{}^i+\Lambda\hPhi+\widetilde{\mathfrak{J}}_\Phi\,.
     \end{equation}
 Since
 \begin{equation}
     \begin{split}
         \Phi k_{i}{}^j k_j{}^i-\tPhi \tk_{i}{}^j \tk_j{}^i =& \,\hPhi k_{i}{}^j k_j{}^i+\tPhi \hk_{i}{}^j k_j{}^i+\tPhi \tk_{i}{}^j \hk_j{}^i\\
         =&\,\hPhi \hk_{i}{}^j \hk_j{}^i+2\hPhi \hk_{i}{}^j \tk_j{}^i+\hPhi \tk_{i}{}^j \tk_j{}^i+\tPhi \hk_{i}{}^j \hk_j{}^i+2\tPhi \hk_{i}{}^j \tk_j{}^i\\
         =&\,\Lambda\hPhi+2H\hk_j{}^j\\
         &+\Phi \hk_{i}{}^j \hk_j{}^i+2\hPhi \hk_{i}{}^j \tk_j{}^i+\hPhi \Bigl(\tk_{i}{}^j \tk_j{}^i-\Lambda\Bigr)+2\tPhi \hk_{i}{}^j \Bigl(\tk_j{}^i-H\tPhi^{-1}\delta_j{}^i\Bigr)\,,
     \end{split}
 \end{equation}
 and also 
 \begin{equation}
 \begin{split}
     \Delta\Phi-\widetilde{\Delta}\tPhi=&g^{ij}\nabla_i\partial_j\Phi+\hg^{ij}\tnabla_i\partial_j\tPhi-g^{ij}\tnabla_i\partial_j\tPhi\\
=&\Delta\hPhi-g^{ij}\hGamma_{ij}^k\partial_k\tPhi+\hg^{ij}\tnabla_i\partial_j\tPhi\,,
 \end{split}
      \end{equation}
      we obtain  that
      \begin{equation}
          \partial_s\hk_l{}^l+2H\hk_l{}^l=\Delta\hPhi+\mathfrak{F}+\widetilde{\mathfrak{J}}_\Phi
      \end{equation}
      which finally implies \eqref{Phi.hat.eq},
      by virtue of the gauge condition \eqref{lapse}. 
\end{proof}

Finally we turn to the constraint equations for differences.
\begin{lemma}[Constraint equations]
\label{lemma:constraint:h}
\begin{align}\label{mom.const.hat}
\nabla_j\widehat k_i{}^j=&\,\partial_i\widehat{\Phi}-\widehat{\Gamma}_{jc}^j(\widetilde{k}_i{}^c-H\delta_i{}^c)+\widehat{\Gamma}_{ji}^c(\widetilde{k}_c{}^j-H\delta_c{}^j)
+\widetilde{\mathfrak{C}}_i\\
\label{mom.const.hat.ii}
     g^{im}\nabla_m\hk_i{}^j=&\,g^{jc}\partial_c \hPhi+\hg^{jc}\tnabla_c(\tk_l{}^l-3H)-\hg^{im}\tnabla_m(\tk_i{}^j-H\delta_i{}^j)\\
     \notag&-g^{im}\widehat{\Gamma}_{mc}^j(\widetilde{k}_i{}^c-H\delta_i{}^c)+g^{im}\widehat{\Gamma}_{mi}^c(\widetilde{k}_c{}^j-H\delta_c{}^j)+\widetilde{\mathfrak{C}}^j
\end{align}
    The terms $\widetilde{\mathfrak{C}}_i,\widetilde{\mathfrak{C}}^j$ only contain variables of the reference metric $\widetilde{\g}$ and are equal to:
    \begin{equation}\label{frak.C}
    \widetilde{\mathfrak{C}}_i=-\widetilde{\RIC}_{0i}\,,\qquad \widetilde{\mathfrak{C}}^j=\widetilde{g}^{ij}\widetilde{\mathfrak{C}}_i\,.
    \end{equation}
\end{lemma}

\begin{proof}
From the Codazzi equations \eqref{eq:perppara} we know that
\begin{equation}
 \tnabla_j\tk_i{}^j-\tnabla_i\tk_l{}^l=\RIC[\widetilde{\g}]_{0i}\,,
\end{equation}
which we subtract from the momentum constraint \eqref{mom.const.hat} to get
\begin{equation}
    \nabla_j\hk_i{}^j-\nabla_i\hk_l{}^l=-\hnabla_j\tk_i{}^j+\widetilde{\mathfrak{C}}_i=-\widehat{\Gamma}_{jc}^j(\widetilde{k}_i{}^c-H\delta_i{}^c)+\widehat{\Gamma}_{ji}^c(\widetilde{k}_c{}^j-H\delta_c{}^j)+\widetilde{\mathfrak{C}}_i\,.
\end{equation}
In view of the gauge condition \eqref{lapse} this is \eqref{mom.const.hat}.
Alternatively, we can also write \eqref{eq:perppara} as
\begin{equation}
    \tg^{im}\tnabla_m \tk_i{}^j-\tg^{jc}\tnabla_c \tk_l{}^l=\tg^{jc}\widetilde{\RIC}_{0c}\,,
\end{equation}
to obtain after subtracting that
\begin{equation}
    g^{im}\nabla_m\hk_i{}^j-g^{jc}\nabla_c \hk_l{}^l=-g^{im}\hnabla_m\tk_i{}^j-\hg^{im}\tnabla_m\tk_i{}^j+\hg^{jc}\tnabla_c\tk_l{}^l+\widetilde{\mathfrak{C}}^j\,.
\end{equation}
In view of the gauge condition, this gives \eqref{mom.const.hat.ii} after expanding the first term on the RHS.
\end{proof}

\begin{remark}\label{rem:symm}
The equations \eqref{k.hat.eq}, \eqref{Gamma.hat.eq} are not symmetric hyperbolic in $\widehat{k}_i{}^j,\widehat\Gamma_{ic}^a$, due to the presence of the terms $\frac{1}{3}\Phi g^{cj}\nabla_c\widehat\Gamma_{ia}^a$, $-\frac{2}{3}\Phi g^{ab}\nabla_i\widehat\Gamma_{ab}^j$ in the RHS of \eqref{k.hat.eq}. However, the latter terms can be treated in the energy estimates by integrating by parts and using the constraint equations \eqref{mom.const.hat}, see Section \ref{subsec:en.id.discussion}.
\end{remark}

The system of equations derived in this section are summarized in Appendix~\ref{app.eqs}.

\section{The background reference metric}
\label{sec:background}

In Section~\ref{sec:ADM}, we have considered a general spacetime $(\mathcal{M},\g)$ foliated by the level sets of a time function $s$. We have also introduced coordinates $(s,x)$. In these coordinates, we will now consider a class of reference metrics of the form \eqref{metric.ref} which are constructed from the family of Kerr de Sitter metrics.
\begin{definition}
\label{def:Os}
We write
\begin{equation*}
    f=\mathcal{O}(\eta e^{m Hs})\,,
\end{equation*}
for some $m\in\mathbb{Z}$ and $\eta>0$, if $f(s,x)$ is a smooth function depending only on the Kerr de Sitter metrics considered, 
with the property that 
\begin{equation}\label{big.O.est}
 |\partial_s^i\partial^\alpha_xf|\leq C_{i,\alpha} \eta e^{mHs},
\end{equation}
for any $i,\alpha$ and $(s,x)\in \mathcal{M}$.
\end{definition}

\subsection{Kerr de Sitter metric}
\label{sec:kerr}

In Boyer Lindquist coordinates $(t,r,\theta,\phi)$, the 
Kerr de Sitter metric reads 
\begin{align}\label{Kerr.metric}
\g_{\mathcal{K}_{a,m}}=\frac{\rho^2}{\Delta_r}\ud r^2
+\frac{\rho^2}{\Delta_\theta} \ud\theta^2
+\sin^2\theta\frac{\Delta_\theta}{\rho^2}\bigl(a\ud t-\frac{r^2+a^2}{\Delta_0}\ud\phi\bigr)^2-\frac{\Delta_r}{\rho^2}\bigl(\ud t-\frac{a\sin^2\theta}{\Delta_0}\ud\phi\bigr)^2,
\end{align}
where we adopt the convention \cite{carter}:
\begin{subequations}
\begin{gather}\label{Kerr.conv}
\rho^2=r^2+a^2\cos^2\theta,\qquad\Delta_r=(r^2+a^2)(1-\frac{\Lambda}{3}r^2)-2mr,\\
\Delta_\theta=1+\frac{\Lambda}{3}a^2\cos^2\theta,\qquad
\Delta_0=1+\frac{\Lambda}{3}a^2.
\end{gather}
\end{subequations}
The cosmological region is the domain  $\Delta_r<0$, where $r$ is a time function. With the following reparametrization of the time function $r$, to
\begin{align}\label{s.ref}
s=H^{-1}\ln r\quad\Leftrightarrow\quad r=e^{Hs}\,,\qquad H=\sqrt{\frac{\Lambda}{3}}\,,
\end{align}
the Kerr de Sitter metric \eqref{Kerr.metric} then takes the form
\begin{align}\label{Kerr.metric.s}
\g_{\mathcal{K}_{a,m}}=-(\Phi_{\mathcal{K}_{a,m}})^2\ud s^2
+(g_{\mathcal{K}_{a,m}})_{ij}\ud x^i \ud x^j,\qquad x^1=t,\,x^2=\theta,\,x^3=\phi\,,
\end{align}
where 
\begin{align}\label{Kerr.metric.lead}
\notag \Phi_{\mathcal{K}_{a,m}}=&\,1+\mathcal{O}(e^{-2Hs})\,,\\ 
(g_{\mathcal{K}_{a,m}})_{11}=&\,H^2 e^{2Hs}+\mathcal{O}(1)\,,\qquad
(g_{\mathcal{K}_{a,m}})_{13}=\,-\frac{a\sin^2\theta}{\Delta_0} H^2 e^{2Hs}+\mathcal{O}(1)\,,\\
\notag(g_{\mathcal{K}_{a,m}})_{22}=&\,\frac{e^{2Hs}}{\Delta_\theta}+\mathcal{O}(1)\,,\qquad
(g_{\mathcal{K}_{a,m}})_{33}= \,\frac{\sin^2\theta\Delta_\theta+H^2 a^2\sin^4\theta}{\Delta_0^2} e^{2Hs}+\mathcal{O}(1)\,.
\end{align}

 Recall here Definition~\ref{def:Os} for our use of the notation $\Os{m}$, for $m\in\mathbb{Z}$.

This transformation is motivated by the standard form of FLRW spacetimes; see \cite{fournodavlos:FLRW}.

\subsection{Partition of Kerr de Sitter}

The Kerr de Sitter metrics $\g_{\K}$ can be viewed as a 2-parameter family of metrics on the fixed differentiable structure of the cosmological region of a Schwarzschild de Sitter spacetime $\g_{m,0}$. In other words, we can view $(\g_{\K})_{ss}$, and $(\g_{\K})_{ij}$ as a family of metric components in a \emph{fixed} chart 
\begin{equation}
    \label{eq:chart}
    (s,t,\phi,\theta)\in (s_\mathcal{C},\infty)\times \mathbb{R}\times (0,2\pi)\times(0,\pi)\,.
\end{equation}

Let $(a_1,m_1)$, $(a_2,m_2)$ be two (possibly different) pairs of Kerr de Sitter parameters, which are sufficiently close to each other
\begin{align}\label{ai.mi}
|a_1-a_2|+|m_1-m_2|<\widetilde{\varepsilon}\,,
\end{align}
and choose $m$ any value between $m_1$ and $m_2$.

We define a metric $\widetilde{\g}$ in the above chart as a smooth transition from $\g_{\mathcal{K}_{a_1,m_1}}$ to $\g_{\mathcal{K}_{a_2,m_2}}$:
\begin{equation}
\begin{split}
    \widetilde{\g}=&\,(1-\chi) \g_{\mathcal{K}_{a_1,m_1}}+\chi \g_{\mathcal{K}_{a_2,m_2}}\\
=&\,\g_{\mathcal{K}_{a_1,m_1}}+\chi(\g_{\mathcal{K}_{a_2,m_2}}-\g_{\mathcal{K}_{a_1,m_1}}) \label{metric.ref.def}
\end{split}
\end{equation}
where $\chi:\mathcal{R}\to[0,1]$ is a smooth function represented by 
\begin{align}\label{chi}
\chi(s,t,\theta,\phi)=\left\{\begin{array}{cc}
    0, & t\leq-1 \\
    1, & t\ge1
\end{array}\right.,\qquad |\partial_\alpha \chi|\leq C_\alpha,
\end{align}
for any coordinate derivative and multi-index $\alpha$. 
We call the reference metric \eqref{metric.ref.def} a \emph{partition of Kerr de Sitter}. 
In the fixed coordinate chart we have
\begin{subequations}
\label{eq:metric:ref:close}
\begin{equation}
\widetilde{\g} =-\widetilde{\Phi}^2\ud s^2+ \widetilde{g}_{ij}\ud x^i\ud x^j 
\end{equation}
where
\begin{align} 
    \tPhi^2=&\,(1-\chi)\Phi_{\mathcal{K}_{a_1,m_1}}^2+\chi \Phi_{\mathcal{K}_{a_2,m_2}}^2
    \\ \tg_{ij}=&\,(1-\chi)(g_{\mathcal{K}_{a_1,m_1}})_{ij}+\chi (g_{\mathcal{K}_{a_2,m_2}})_{ij}\,.
\end{align}
\end{subequations}

\begin{remark}\label{rem:coord}
In view of the use of polar coordinates on the cylinder $(x^1,x^2,x^3)=(t,\theta,\phi)$ the spacetime manifold is not covered by a single chart \eqref{eq:chart}.
However, an atlas can be constructed from two of these charts. 
For the pointwise bounds below it is understood that the chart is restricted away from the points where the metric degenerates in these coordinates. 
\end{remark}

The main  properties are recorded in the following Lemma. 
\begin{lemma}
\label{lem:ref.properties}
    The components of the reference metric \eqref{metric.ref.def} satisfy
        \begin{gather}\label{ref.metric.prop}
         \widetilde{\Phi}-1=\mathcal{O}(e^{-2Hs}),\qquad 
         \widetilde{g}_{ij}=\mathcal{O}(e^{2Hs}),
         \qquad         \widetilde{g}^{ij}=\mathcal{O}(e^{-2Hs}), \\
         \widetilde{\Gamma}_{ic}^a=\mathcal{O}(1),\qquad\widetilde{k}_i{}^j-H\delta_i{}^j=\mathcal{O}(e^{-2Hs}),
        \end{gather}
    where the $\mathcal{O}(e^{mHs})$ terms satisfy \eqref{big.O.est}.
\end{lemma}

\begin{proof}
The statement for $\widetilde{\Phi}$ follows from \eqref{eq:metric:ref:close}, since $\Phi_{\mathcal{K}_{a_1,m_1}},\Phi_{\mathcal{K}_{a_2,m_2}}$ have the same property, see \eqref{Kerr.metric.lead}. In fact, from \eqref{Kerr.metric.lead} and \eqref{eq:metric:ref:close}, in the region where each coordinate chart is regular (cf. Remark \ref{rem:coord}), we also have 
\begin{align*}
C^{-1} e^{2Hs}\leq \widetilde{g}_{ij}\leq C e^{2Hs},\qquad |\partial^\alpha\widetilde{g}_{ij}|\leq C_\alpha e^{2Hs}\,, 
\end{align*}
for all $(s,x)\in\mathcal{M}$. This implies both the statements for $\widetilde{g}_{ij}$ and $\widetilde{g}^{ij}$. 
Then the statement for $\widetilde{\Gamma}_{ic}^a=\frac{1}{2}\widetilde{g}^{al}(\partial_i\widetilde{g}_{cl}+\partial_c\widetilde{g}_{il}-\partial_l\widetilde{g}_{ic})$ becomes obvious.

For $\widetilde{k}_i{}^j$ we compute
\begin{equation*}
    2\Phi_{\K} (k_{\K})_{ij}=\partial_s(g_{\K})_{ij}= 2H (g_{\K})_{ij}+\mathcal{O}(1)\,,
\end{equation*}
and hence
\begin{equation}
\begin{split}
       (k_{\K})_i{}^j=&(g_{\K}^{-1})^{jm}(k_{\K})_{im}=\Phi_{\K}^{-1} H \delta_i{}^j+\mathcal{O}(e^{-2Hs})    \\
       =&H \delta_i{}^j+\mathcal{O}(e^{-2Hs}) \,.\label{eq:k:kerr}
\end{split}
\end{equation}
Therefore again
\begin{equation*}
    2\tPhi\tk_{ij}=\partial_s\tg_{ij}= 2H \tg_{ij}+\mathcal{O}(1)\,,
\end{equation*}
and 
\begin{equation*}
       \tk_i{}^j=\tg^{jm}\tk_{im}=\tPhi^{-1} H \delta_i{}^j+\mathcal{O}(e^{-2Hs})   
       =H \delta_i{}^j+\mathcal{O}(e^{-2Hs})\,.
\end{equation*}
This completes the proof of the lemma.   
\end{proof}

\subsection{Approximate solution}

A partition of Kerr de Sitter is not a solution to the Einstein equations (in the transition region $|t|<1$). However, all that is needed for our purposes is that it is an \emph{approximate solution}. It turns out that the following is sufficient.
\begin{proposition}\label{prop:approx.sol} 
The partition metric \eqref{eq:metric:ref:close} satisfies, on one hand:
\begin{align}\label{approx.sol}
\left\{\begin{array}{ccc}\big|\partial_x^\alpha\big(\widetilde\RIC_{00}+\Lambda\big)\big|&\leq\quad C_\alpha e^{-2Hs}\\
\big|\partial_x^\alpha \widetilde\RIC_0{}^j\big|&\leq\quad C_\alpha e^{-4Hs}\\
\big|\partial_x^\alpha\big(\widetilde\RIC_i{}^j+\Lambda\delta_i{}^j\big)\big|&\leq\quad C_\alpha e^{-2Hs}
\end{array}\right.,\qquad\text{for $|t|<1$},
\end{align}
and on the other:
\begin{equation}
    \big|\partial_x^\alpha\big(\widetilde\RIC_{00}+\Lambda\big)\big| + \big|\partial_x^\alpha \widetilde\RIC_0{}^j\big| + \big|\partial_x^\alpha\big(\widetilde\RIC_i{}^j+\Lambda\delta_i{}^j\big)\big| \leq C_\alpha \widetilde{\varepsilon}\,, \qquad\text{for $|t|<1$}. 
\end{equation}
Moreover, 
\begin{align}\label{approx.sol.supp}
\widetilde\RIC_{00}+\Lambda=\widetilde\RIC_0{}^j=\widetilde\RIC_i{}^j+\Lambda\delta_i{}^j=0,\qquad\text{for $|t|\ge1$}.
\end{align}
\end{proposition}
\begin{proof}
\eqref{approx.sol.supp} is immediate from the definition of the partition metric \eqref{metric.ref.def}, since for $|t|\ge1$, $\widetilde{\g}$ coincides with one of the two Kerr de Sitter metrics $\g_{\mathcal{K}_{a_1,m_1}},\g_{\mathcal{K}_{a_2,m_2}}$. 

The bounds \eqref{approx.sol} are proven in two steps. For the decay statement, the specific expression of the Kerr de Sitter metric is actually not used. Instead, we show that any metric which satisfies Lemma \ref{lem:ref.properties} has this property. On the other hand, the $\widetilde{\varepsilon}$ smallness of the relevant terms follows from the assumption \eqref{ai.mi} and the precise formula \eqref{Kerr.metric}.

{\it Step 1. Decay.} The normal vectorfield $N=\tPhi^{-1}\partial_s$ satisfies
$\Nabla_N N=\tPhi^{-1}\nabla\tPhi$ hence
\begin{gather*}
    \Nabla_{\partial_s}\partial_s = \tPhi^{-1}\partial_s\tPhi \partial_s+\tPhi{\nabla}\tPhi\,,\\
\bm{\tGamma}_{ss}^s=\tPhi^{-1}\partial_s\tPhi=\Os{-2}\qquad \bm{\tGamma}_{ss}^i=\tPhi \tg^{ij}\partial_j\tPhi=\Os{-4}\,.
\end{gather*}
Moreover, we know $\Nabla_{\partial_i}N=k_i{}^j\partial_j$,
\begin{gather*}   \Nabla_{\partial_i}\partial_s=\tPhi^{-1}\partial_i\tPhi\partial_s+\tPhi \tk_i{}^j\partial_j\,,\\
\bm{\tGamma}_{is}^s=\tPhi^{-1}\partial_i\tPhi=\Os{-2}\qquad\bm{\tGamma}_{is}^j=\tPhi \tk_i{}^j=H\delta_i{}^j+\Os{-2}\,.
\end{gather*}
Since $\Nabla_{\partial_i}\partial_j=\bm{\tGamma}_{ij}^s \partial_s+\tGamma_{ij}^k \partial_k$ where $\tGamma_{ij}^k$ are the connection coefficients of $\tg$, we compute
\begin{equation*}
    \bm{\tGamma}_{ij}^s=-\tPhi^{-1}\widetilde{\g}(\Nabla_{\partial_i}\partial_j,\tPhi^{-1}\partial_s)=\tPhi^{-1}\tk_{ij}=H G_{ij}(x) e^{2Hs}+\mathcal{O}(1)\,,
\end{equation*}
    and
    \begin{align*}
        \tGamma_{ij}^{k}=\frac{1}{2}G^{kl}\bigl(\partial_i G_{jl}+\partial_j G_{il}-\partial_lG_{ij}\bigr)+\Os{-2}\,.
    \end{align*}
    
Let us now compute the components of the Ricci curvature. We start with $\widetilde{\RIC}_{00}$.
In view of the expression \eqref{eq:Ricci:local}, we compute
\begin{align*}
    \partial_\alpha \bm{\tGamma}_{ss}^\alpha=&\Os{-2}\,,\qquad
    \partial_s\bm{\tGamma}^\alpha_{\alpha s}=\Os{-2}\,,\\
    \bm{\tGamma}_{\alpha\gamma}^\alpha  \bm{\tGamma}_{ss}^\gamma=&\bm{\tGamma}_{\alpha s}^\alpha  \bm{\tGamma}_{ss}^s+\bm{\tGamma}_{\alpha i}^\alpha  \bm{\tGamma}_{ss}^i=\Os{-2}\,,\\
    \bm{\tGamma}_{s\gamma}^\alpha \bm{\tGamma}_{\alpha s}^\gamma=&\bm{\tGamma}_{s j}^i \bm{\tGamma}_{i s}^j+\Os{-2}=3H^2+\Os{-2}\,,
\end{align*}
and therefore
\begin{equation}
    \widetilde{\R}_{00}=\widetilde{\Phi}^{-2}\widetilde{\R}_{ss}=-\Lambda+\Os{-2},
\end{equation}
since $\widetilde{\Phi}^{-2}=1+\Os{-2}$.

Now compute $\widetilde{\RIC}_{0j}=\widetilde{\Phi}^{-1}\widetilde{\RIC}_{sj}$:
\begin{align*}
    \partial_\alpha\bm{\tGamma}^\alpha_{j s}&=\Os{-2}\,,\qquad
    \partial_j\bm{\tGamma}^\alpha_{\alpha s}=\Os{-2}\,,\\
    \bm{\tGamma}_{\alpha\gamma}^\alpha  \bm{\tGamma}_{j s}^\gamma&=\bm{\tGamma}_{i k}^i  \bm{\tGamma}_{j s}^k+\Os{-2}=H\tGamma_{ij}^i+\Os{-2}\,,\\
    \bm{\tGamma}_{j\gamma}^\alpha \bm{\tGamma}_{\alpha s}^\gamma &= \bm{\tGamma}_{ji}^s\bm{\tGamma}^i_{ss}+\bm{\tGamma}_{j k}^i \bm{\tGamma}_{i s}^k+\Os{-2}=H\tGamma^i_{ji}+\Os{-2}\,,
\end{align*}
and so by symmetry we have a cancellation
\begin{equation}
    \widetilde{\RIC}_{0j}=\Os{-2}\quad\Rightarrow\quad \widetilde{\RIC}_0{}^j=\widetilde{g}^{aj}\widetilde{\RIC}_{0a}=\mathcal{O}(e^{-4Hs})\,.
\end{equation}

So it remains to compute $\widetilde{\RIC}_{ij}$:
\begin{align*}
    \partial_\alpha\bm{\tGamma}^\alpha_{ji}&=2H^2 G_{ij} e^{2Hs}+\mathcal{O}(1)\,,\qquad
    \partial_j{\bm\tGamma}^\alpha_{\alpha i}=\mathcal{O}(1)\,,\\
    \bm{\tGamma}_{\alpha\gamma}^\alpha  {\bm\tGamma}_{ji}^\gamma&=\bm{\tGamma}_{k s}^k  {\bm\tGamma}_{ji}^s+\mathcal{O}(1)=3H^2 G_{ij} e^{2Hs}+\mathcal{O}(1)\,,\\
    \bm{\tGamma}_{j \gamma}^\alpha {\bm\tGamma}_{\alpha i}^\gamma&=\bm{\tGamma}_{j k}^s {\bm\tGamma}_{s i}^k+\bm{\tGamma}_{j s}^k {\bm\tGamma}_{k i}^s+\mathcal{O}(1)=2 H^2G_{ij} e^{2Hs}+\mathcal{O}(1)\,.
\end{align*}
Therefore,
\begin{equation}
    \widetilde{\RIC}_{ij}=3 H^2 G_{ij}e^{2Hs}+\mathcal{O}(1)=\Lambda\tg_{ij}+\mathcal{O}(1)\,,
\end{equation}
and thus in view of \eqref{ref.metric.prop},
\begin{equation}
\widetilde{\RIC}_{i}{}^j-\Lambda \delta_{i}{}^j=\tg^{jk}\Bigl(\widetilde{\RIC}_{ik}-\Lambda\tg_{ik}\Bigr)=\Os{-2}\,.
\end{equation}

{\it Step 2. Smallness.} Denote by $\kGamma_{\alpha\beta}^\gamma,{}^\mathcal{K}\RIC_{\mu\nu}$ the Christoffel symbol and Ricci curvature of the Kerr de Sitter metric with either $(a_1,m_1)$ or $(a_2,m_2)$ parameters. It follows directly from \eqref{Kerr.metric} and \eqref{ai.mi}
that
\begin{align*}
\bm{\tGamma}_{ss}^s=&\,\tPhi^{-1}\partial_s\tPhi={\kGamma}_{ss}^s+\Oe\\ \bm{\tGamma}_{ss}^i=&\,\tPhi \tg^{ij}\partial_j\tPhi=\kGamma_{ss}^i+e^{-2Hs}\Oe\\
\bm{\tGamma}_{is}^s=&\,\tPhi^{-1}\partial_i\tPhi=\kGamma_{is}^s+\Oe \\\bm{\tGamma}_{is}^j=&\,\tPhi \tk_i{}^j=\Phi_{\K}(k_{\K})_i{}^j+\Oe=\kGamma_{is}^j+\Oe\\
\bm{\tGamma}_{ij}^s=&\,\tPhi^{-1}\tk_{ij}=\kGamma_{ij}^s+e^{2Hs}\Oe\\
        \bm{\tGamma}_{ij}^{k}=&\,\kGamma_{ij}^k+\Oe\,.
    \end{align*}
Therefore,
\begin{multline*}
    \widetilde{\RIC}_{ss}=\partial_\alpha \bm{\tGamma}_{ss}^\alpha-\partial_s\bm{\tGamma}^\alpha_{\alpha s}
    +\bm{\tGamma}_{\alpha\gamma}^\alpha  \bm{\tGamma}_{ss}^\gamma
    -\bm{\tGamma}_{s\gamma}^\alpha \bm{\tGamma}_{\alpha s}^\gamma\\
    ={}^\mathcal{K}\RIC_{ss}+\Oe=\Lambda (\g_{\K})_{ss}+\Oe=\Lambda \widetilde{\g}_{ss}+\Oe,
\end{multline*}
because $\widetilde{\g}_{ss}=-\tPhi^2=- \Phi_{\K}^2+\Oe$. Hence, $\widetilde{\RIC}_{00}+\Lambda=\mathcal{O}(\widetilde{\varepsilon})$.

Similarly for $\widetilde{\RIC}_{0j}=\Phi^{-1}\widetilde{\RIC}_{sj}$:
\begin{equation*}
    \widetilde{\RIC}_{sj}= \partial_\alpha\bm{\tGamma}^\alpha_{j s}-\partial_j\bm{\tGamma}^\alpha_{\alpha s}+\bm{\tGamma}_{\alpha\gamma}^\alpha  \bm{\tGamma}_{j s}^\gamma-\bm{\tGamma}_{j\gamma}^\alpha \bm{\tGamma}_{\alpha s}^\gamma
    = {}^\mathcal{K}\RIC_{sj}+\Oe=\Oe\,,
\end{equation*}
where we have used that $\bm{\tGamma}_{ji}^s\bm{\tGamma}_{ss}^i=\Oe$.

It remains to compute $\widetilde{\RIC}_{ij}$:
\begin{align*}    \partial_\alpha\bm{\tGamma}^\alpha_{ji}&=\partial_\alpha\kGamma^\alpha_{ji}+e^{2Hs}\Oe\,,\qquad
    \partial_j{\bm\tGamma}^\alpha_{\alpha i}=\partial_j\kGamma^\alpha_{\alpha i}+\Oe\,,\\
    \bm{\tGamma}_{\alpha\gamma}^\alpha  {\bm\tGamma}_{ji}^\gamma&=\kGamma_{\alpha\gamma}^\alpha  \kGamma_{ji}^\gamma+e^{2Hs}\Oe\,,\qquad
    \bm{\tGamma}_{j \gamma}^\alpha {\bm\tGamma}_{\alpha i}^\gamma=\kGamma_{j \gamma}^\alpha \kGamma_{\alpha i}^\gamma+e^{2Hs}\Oe\,.
\end{align*}
Therefore
\begin{equation*}
    \widetilde{\RIC}_{ij}=\RIC[\g_{\K}]+e^{2Hs}\Oe=\Lambda \widetilde{\g}_{ij}+e^{2Hs}\Oe\,,
\end{equation*}
and thus in view of \eqref{approx.sol},
\begin{equation*}
\widetilde{\RIC}_{i}{}^j-\Lambda \delta_{i}{}^j=\Oe\,.
\end{equation*}
This completes the proof of the proposition.
\end{proof}

\section{The bootstrap argument}
\label{sec:bootstrap}

In this section we prove the global existence statement of Theorem~\ref{thm:fs} (I).

\subsection{Weighted norms and energy}
\label{sec:norms}

For the spacetimes we consider, $\Sigma_s$ is diffeomorphic to $\mathbb{R}\times\mathbb{S}^2$.
While $g=g_s$ is a metric on $\Sigma_s$,
we endow $(\mathbb{R}\times\mathbb{S}^2,\mathring{g})$ with the standard metric $\mathring{g}$ on the cylinder:
\begin{align}\label{g.0}
\mathring{g}=\ud t^2+\mathring{\gamma},\qquad \mathring{\gamma}=\ud\theta^2+\sin^2\theta \ud\phi^2.
\end{align}
The coordinates on $\Sigma_s$ are denoted by $(t,\theta,\phi)$; recall that by construction these are coordinates on each $\Sigma_s$, $s\geq s_0$, and can also be identified with coordinates on a chart for $\mathbb{R}\times\mathbb{S}^2$.
Given a $\Sigma_s$-tangent $(n,m)$ tensor $\mathcal{T}$, we define 
\begin{align*}
(\nabla^{(\ell)}\mathcal{T})^{j_1\ldots j_n}_{a_1\ldots a_\ell i_1\ldots i_m}=\nabla_{a_1}\ldots\nabla_{a_\ell}\mathcal{T}^{j_1\ldots j_n}_{i_1\ldots i_m}
\end{align*}
and 
\begin{align}\label{T.norm}
|\mathcal{T}|_g^2=g^{i_1i_1'}\ldots g^{i_mi_m'}g_{j_1j_1'}\ldots g_{j_nj_n'}\mathcal{T}^{j'_1\ldots j_n'}_{i_1'\ldots i_m'}\mathcal{T}^{j_1\ldots j_n}_{i_1\ldots i_m}.
\end{align}
We define $\mathring{\nabla}^{(\ell)}\mathcal{T}$ and $|\mathcal{T}|_{\mathring{g}}^2$ similarly, using the covariant derivative $\mathring{\nabla}$ of $\mathring{g}$, instead of $\nabla$, and contracting indices with $\mathring{g}$ instead of $g$.

\begin{definition}
Let $W^{M,\infty}(\Sigma_s,g)$ be the space of $\Sigma_s$-tangent tensors with $M$ bounded spatial derivatives  with respect to 
\begin{align}\label{WM.infty}
\|\mathcal{T}\|_{W^{M,\infty}(\Sigma_s,g)}=\sum_{\ell\leq M}\text{ess\,sup}_{p\in\Sigma_s}e^{\ell Hs}|\nabla^{(\ell)}\mathcal{T}|_g(p)\,.
\end{align}
In particular, $L^\infty(\Sigma_s,g)=W^{0,\infty}(\Sigma_s,g)$ with norm
\begin{equation}\label{Linfty}
\|\mathcal{T}\|_{L^\infty(\Sigma_s,g)}=\text{ess\,sup}_{p\in\Sigma_s}|\mathcal{T}|_g(p)\,.
\end{equation}
\end{definition}

\begin{definition} \label{def:norms}
Let $H^M(\Sigma_s,g)$ be the Sobolev space of $\Sigma_s$-tangent tensors with $M$ square integrable spatial derivatives with respect to the \emph{weighted} norm 
\begin{align}\label{Hw}
\|\mathcal{T}\|_{H^M(\Sigma_s,g)}^2=\sum_{\ell\leq M}\int_{\Sigma_s} f^2(t) e^{2\ell Hs}|\nabla^{(\ell)} \mathcal{T}|^2_g\,e^{-3Hs}\mathrm{vol}_g,
\end{align}
where $\mathrm{vol}_g=\sqrt{|g|}\:\ud t\wedge\ud \theta\wedge \ud \phi$ is the volume form of $(\Sigma_s,g)$ in local coordinates and 
\begin{equation} \label{weight}
    f(t)=e^{\alpha_1 t}+e^{-\alpha_2 t}
\end{equation}
is  a weight function, for some
$\alpha_1,\alpha_2\ge0$. 
 In particular,
 \begin{equation}
\label{L2}\|\mathcal{T}\|^2_{L^2(\Sigma_s,g)}=\int_{\Sigma_s}f^2(t)\:|\mathcal{T}|_g^2\: e^{-3Hs}\vol{g}\,.
\end{equation}
\end{definition}

We will also sometimes use the $W^{M,\infty}(\mathbb{R}\times\mathbb{S}^2,\mathring{g}),H^M(\mathbb{R}\times\mathbb{S}^2,\mathring{g})$ norms, for $\Sigma_s$-tangent tensors, defined as follows:
\begin{align}\label{WM.infty.mathring.g}
\|\mathcal{T}\|_{W^{M,\infty}(\mathbb{R}\times\mathbb{S}^2,\mathring{g})}=\sum_{\ell\leq M}\text{ess\,sup}_{p\in \mathbb{R}\times\mathbb{S}^2}|\mathring{\nabla}^{(\ell)}\mathcal{T}|_{\mathring{g}}(p)\,,\\
\label{HM.mathring.g}
\|\mathcal{T}\|_{H^M(\mathbb{R}\times\mathbb{S}^2,\mathring{g})}^2=\sum_{\ell\leq M}\int_{\mathbb{R}\times\mathbb{S}^2} f^2(t) |\mathring{\nabla}^{(\ell)} \mathcal{T}|^2_{\mathring{g}}\mathrm{vol}_{\mathring{g}}
\end{align}

\begin{remark}
Note that in \eqref{WM.infty} and \eqref{Hw}, each extra spatial derivative $\nabla$ comes at a cost of a weight $e^{Hs}$.
Moreover in \eqref{Hw} the volume form is \emph{renormalised},
to the effect that in this setting,
\begin{align}\label{vol.g}
C^{-1} \leq e^{-3Hs}\sqrt{|g|}\leq C\,,
\end{align}
by virtue of the bootstrap assumptions \eqref{Boots} on $g$ below.
\end{remark}

\begin{remark}
The exponential rates $\alpha_1,\alpha_2$
are related to the exponential decay of the perturbed solution towards the endpoints $t=\pm\infty$.
They are nonnegative, and may not be equal; they can also be set to $0$, when no exponential decay along the cosmological horizons is imposed initially. While the norm $\|\cdot\|_{H^M}$ does depend on $\alpha_1,\alpha_2$, we typically supress this in the notation. To simplify notation, we often drop $(\Sigma_s,g)$ from subscript to the norms \eqref{WM.infty}, \eqref{Linfty}, \eqref{Hw}, and \eqref{L2}.
\end{remark}

Next, we define the overall energy for the variables $\widehat g_{ij},\widehat{g}^{ij},\widehat{k}_i{}^j,\widehat\Gamma_{ic}^a,\widehat\Phi$.
\begin{definition}
    If $N\in\mathbb{N}$ denotes the total number of derivatives we are commuting the main equations with, then let
\begin{align}\label{en.def}
\begin{split} 
\mathcal{E}_N(s)=&\,\|\widehat g\|^2_{H^N(\Sigma_s,g)}+\|\widehat g^{-1}\|^2_{H^N(\Sigma_s,g)}
+e^{3Hs}\|\widehat\Phi\|^2_{H^N(\Sigma_s,g)}\\
&+e^{2Hs}\big(\|\widehat\Gamma \|^2_{H^N(\Sigma_s,g)}
+\|\widehat k\|^2_{H^N(\Sigma_s,g)}\big)\,.
\end{split}
\end{align}
\end{definition}
\begin{remark}
In terms of the $e^{Hs}$ weights,
boundedness of the energy \eqref{en.def} is optimal for $\widehat{g},\widehat{g}^{-1},\widehat{\Gamma}$, since the corresponding Kerr de Sitter variables themselves do not behave better. However, the $e^{Hs}$ weights in the norms of $\widehat{\Phi},\widehat{k}$ in \eqref{en.def} are sub-optimal relative to the expected behavior of these variables ($e^{4Hs}$ would be optimal for both instead of $e^{3Hs},e^{2Hs}$). For  technical reasons (hyperbolicity, boundedness of error terms etc.), we cannot propagate optimal estimates for all variables at the same time. Nevertheless, once we have completed our bootstrap argument (see Sections \ref{subsec:Boots}, \ref{subsec:glob.stab}), the precise asymptotic behavior of all components of the perturbed solution can be derived (see Section \ref{sec:prec.asym}).
\end{remark}

\subsection{Bootstrap assumptions and basic consequences}\label{subsec:Boots}

In view of the local well-posedness statement in Appendix~\ref{sec:app},
there exists $\varepsilon > 0$ and a maximal time of existence 
$s_b\in(s_0,+\infty)$ such that the following inequalities hold:
\begin{align}\label{Boots}
\|\widehat{g}\|_{W^{2,\infty}(\mathbb{R}\times\mathbb{S}^2,\mathring{g})}\leq\varepsilon e^{2Hs} ,\qquad \|\widehat{g}^{-1}\|_{W^{2,\infty}(\mathbb{R}\times\mathbb{S}^2,\mathring{g})}\leq \varepsilon e^{-2Hs},\qquad
\mathcal{E}_N(s)\leq \varepsilon^2,
\end{align}
for all $s\in[s_0,s_b)$, and some $N\ge 4$.

These are our bootstrap assumptions, and the statement that such a bootstrap time $s_b\in (s_0,+\infty)$  exists follows from classical Cauchy stability for the locally well-posed system of Appendix~\ref{sec:app}~\&~\ref{app.eqs}. 
In Section~\ref{subsec:glob.stab} we will show that for data sufficiently close to Kerr de Sitter,
these bounds are not saturated at $s=s_b$, and the solution can hence be continued, using again the local existence statement of Appendix~\ref{sec:app}, to yield the global existence result of Theorem~\ref{thm:fs}~(I).

\subsection{Preliminary estimates}
\label{sec:Boots.est}

The bootstrap assumptions \eqref{Boots} have certain basic implications, which will be useful in deriving energy estimates below. First, we compare the norm \eqref{T.norm} to the components.
\begin{lemma}\label{lem:T.norm}
Let $\mathcal{T}$ be a $\Sigma_s$-tangent $(n,m)$ tensor. Then the following inequalities hold:
\begin{align}\label{T.norm.ineq}
C^{-1}\min_{\substack{i_1,\ldots,i_m,\\j_1,\ldots,j_n}}|\mathcal{T}^{j_1\ldots j_n}_{i_1\ldots i_m}|\leq e^{(m-n)Hs}|\mathcal{T}|_g\leq C \max_{\substack{i_1,\ldots,i_m,\\j_1,\ldots,j_n}}|\mathcal{T}^{j_1\ldots j_n}_{i_1\ldots i_m}|
\end{align}
for all $(s,x)\in(s_0,s_b)\times\Sigma_s$.
\end{lemma}
\begin{proof}
The first two bounds in \eqref{Boots} imply that 
\begin{align*}
C^{-1}e^{2Hs}\leq g_{ij}\leq Ce^{2Hs},\qquad    
C^{-1}e^{-2Hs}\leq g^{ij}\leq Ce^{-2Hs},
\end{align*}
in a given regular coordinate patch, from which the desired inequalities readily follow. 
\end{proof}
Next, we derive the Sobolev embedding for the weighted norm \eqref{Hw}.
\begin{lemma}\label{lem:Hw}
Let $\mathcal{T}$ be a $\Sigma_s$-tangent $(n,m)$ tensor. Then the following inequality holds:
\begin{align}\label{Sob}
\|f(t)\mathcal{T}\|_{W^{M,\infty}(\Sigma_s,g)}\leq C\|\mathcal{T}\|_{H^{M+2}(\Sigma_s,g)}
\end{align}
for all $s\in[s_0,s_b)$.
\end{lemma}
\begin{proof}
The classical Sobolev embedding in $\mathbb{R}\times\mathbb{S}^2$ implies that 
\begin{align*}
f^2(t)|\mathcal{T}^{j_1\ldots j_n}_{i_1\ldots i_m}|^2\leq C\sum_{\ell\leq2}\int_{\Sigma_s}|\mathring{\nabla}^{(\ell)}[f(t) \mathcal{T}]|^2_{\mathring{g}}\,\mathrm{vol}_{\mathring{g}}\leq C\|\mathcal{T}\|_{H^2(\mathbb{R}\times\mathbb{S}^2,\mathring{g})}^2,
\end{align*}
since $|\mathring{\nabla}^{(\ell)} f(t)|\leq Cf(t)$.
Recall \eqref{vol.g}
and use Lemma \ref{lem:T.norm} to deduce that
\begin{align}\label{pre.Sob}
\|f(t)\mathcal{T}\|_{L^\infty(\Sigma_s,g)}^2\leq C\int_{\Sigma_s}f^2(t)\big[|\mathcal{T}|^2_g +e^{2Hs}|\nabla(\mathcal{T})|^2_g+e^{4Hs}|\nabla( \nabla(\mathcal{T}))|^2_g\big]e^{-3Hs}\mathrm{vol}_g,
\end{align}
where schematically
\begin{align*}
\nabla(\mathcal{T})=&\,\nabla\mathcal{T}+\Gamma\star\mathcal{T}\\
\nabla(\nabla(\mathcal{T}))=&\,\nabla\nabla\mathcal{T}+\Gamma\star\nabla\mathcal{T}+\Gamma\star\Gamma\star\mathcal{T}+\nabla\Gamma\star\mathcal{T} \,.
\end{align*}
By the first two bounds in the bootstrap assumptions \eqref{Boots}, the correction terms satisfy
\begin{align*}
    |\Gamma\star\mathcal{T}|_g^2\leq &\,Ce^{-2Hs}|\mathcal{T}|^2_g,\\ |\Gamma\star\nabla\mathcal{T}
    +\Gamma\star\Gamma\star\mathcal{T}
    +\nabla\Gamma\star\mathcal{T}|^2_g
    \leq &\,Ce^{-2Hs}|\nabla \mathcal{T}|^2_g+Ce^{-4Hs}|\mathcal{T}|^2_g.
\end{align*}
Inserting the former identities into \eqref{pre.Sob} and using the latter bounds gives \eqref{Sob} for $M=0$. The proof for $M>0$ is the same, replacing $\mathcal{T}$ by $\nabla^{(\ell)}\mathcal{T}$, for $\ell\leq M$.
\end{proof}
An immediate consequence of the bootstrap assumptions \eqref{Boots} and the previous lemma is the following.
\begin{lemma}\label{lem:WN-2.infty}
The variables $\widehat{g},\widehat{g}^{-1},\widehat{\Phi},\widehat{\Gamma},\widehat{k}$ satisfy the $W^{N-2,\infty}(\Sigma_s,g)$ bound:
\begin{equation}\label{WN-2.infty.est}
\begin{split}
&\|f(t)\widehat g\|_{W^{N-2,\infty}(\Sigma_s,g)}+\|f(t)\widehat g^{-1}\|_{W^{N-2,\infty}(\Sigma_s,g)}
+e^{\frac{3}{2}Hs}\|f(t)\widehat\Phi\|_{W^{N-2,\infty}(\Sigma_s,g)}\\
\notag&+e^{Hs}\big(\|f(t)\widehat\Gamma \|_{W^{N-2,\infty}(\Sigma_s,g)}
+\|f(t)\widehat k\|_{W^{N-2,\infty}(\Sigma_s,g)}\big)\leq C\varepsilon,
\end{split}
\end{equation}
for all $s\in[s_0,s_b)$. 
\end{lemma}
Next, we compare the norms defined relative to $g$ and $\mathring{g}$.
\begin{lemma}\label{lem:norm.equiv}
Let $\mathcal{T}$ be a $\Sigma_s$-tangent $(n,m)$ tensor. Then for $M\leq N-1$:
\begin{align}\label{norm.equiv}
C^{-1}\|\mathcal{T}\|_{W^{M,\infty}(\mathbb{R}\times\mathbb{S}^2,\mathring{g})}\leq e^{(m-n)Hs}\|\mathcal{T}\|_{W^{M,\infty}(\Sigma_s,g)}\leq C\|\mathcal{T}\|_{W^{M,\infty}(\mathbb{R}\times\mathbb{S}^2,\mathring{g})}
\end{align}
and for $M\leq N$:
\begin{align}\label{HM.equiv}
C^{-1}\|\mathcal{T}\|_{H^M(\mathbb{R}\times\mathbb{S}^2,\mathring{g})}\leq e^{(m-n)Hs}\|\mathcal{T}\|_{H^M(\Sigma_s,g)}\leq C\|\mathcal{T}\|_{H^M(\mathbb{R}\times\mathbb{S}^2,\mathring{g})}.
\end{align}
\end{lemma}
\begin{proof}
It follows from the bootstrap assumptions \eqref{Boots} and Lemmas \ref{lem:T.norm}, \ref{lem:WN-2.infty}.
\end{proof}

\subsection{Global stability}\label{subsec:glob.stab}

The main energy estimates that we derive in Section \ref{subsec:en.est} prove the following.
\begin{theorem}\label{thm:Boots}
Assume that the bootstrap assumptions \eqref{Boots} are valid for some $N\ge 4$. Then the perturbed solution satisfies the energy estimate:
\begin{multline}\label{Boots:en.est}
\mathcal{E}_N(s)\leq C\mathcal{E}_N(s_0)+C \int_{s_0}^s\,e^{-\frac{1}{2}H\tau}\mathcal{E}_N(\tau)\ud \tau\\ 
+C\int_{s_0}^s e^{\frac{7}{2}H\tau}\big\{\|\widetilde{\mathfrak{J}}_\Phi\|_{W^{N,\infty}(\mathbb{R}\times\mathbb{S}^2,\mathring{g})}^2
+\|(\widetilde{\mathfrak{J}}_k)_i{}^j\|_{W^{N,\infty}(\mathbb{R}\times\mathbb{S}^2,\mathring{g})}^2
+\|\widetilde{\mathfrak{C}}_i\|_{W^{N,\infty}(\mathbb{R}\times\mathbb{S}^2,\mathring{g})}^2\big\} \ud \tau
\end{multline}
for all $s\in[s_0,s_b)$.
\end{theorem}
\begin{proof}
In view of the definition of the overall energy \eqref{en.def}, this inequality follows directly from the main energy estimates in differential form in Proposition~\ref{prop:main.en.est}.
Note that by adding up the inequalities \eqref{main.en.est.g}-\eqref{main.en.est.Gamma.k}, after multiplying the equation \eqref{main.en.est.Phi} for $\hPhi$ by a suitably large constant, the term involving $\nabla^{(N+1)}\hPhi$ on the RHS of \eqref{main.en.est.Gamma.k} is absorbed by the positive term on the LHS of \eqref{main.en.est.Phi}.
\end{proof}
The previous theorem, combined with a standard continuation argument, implies that 
\begin{corollary}\label{cor:global}
The perturbed solution exists in all of $[s_0,+\infty)\times\Sigma_s$, satisfying the global estimate:
\begin{align}\label{glob:en.est}
\|\widehat{g}\|_{W^{2,\infty}(\mathbb{R}\times\mathbb{S}^2,\mathring{g})}\leq C\mathring{\varepsilon} e^{2Hs} ,\qquad \|\widehat{g}^{-1}\|_{W^{2,\infty}(\mathbb{R}\times\mathbb{S}^2,\mathring{g})}\leq C\mathring{\varepsilon} e^{-2Hs},\qquad
\mathcal{E}_N(s)\leq C\mathring{\varepsilon}^2, 
\end{align}
for all $s\in[s_0,+\infty)$, where $\mathring{\varepsilon}^2:=\mathcal{E}_N(s_0)$.
\end{corollary}
\begin{proof}
Applying Gronwall's inequality to \eqref{Boots:en.est} gives
\begin{align*}
\mathcal{E}_N(s)\leq 
C\bigg[\mathcal{E}_N(s_0)+\int^{s_b}_{s_0}
e^{\frac{7}{2}Hs}\big\{\|\widetilde{\mathfrak{J}}_\Phi\|_{W^{N,\infty}(\mathbb{R}\times\mathbb{S}^2,\mathring{g})}^2
+\|(\widetilde{\mathfrak{J}}_k)_i{}^j\|_{W^{N,\infty}(\mathbb{R}\times\mathbb{S}^2,\mathring{g})}^2
+\|\widetilde{\mathfrak{C}}_i\|_{W^{N,\infty}(\mathbb{R}\times\mathbb{S}^2,\mathring{g})}^2\big\}\ud s\bigg]\,.
\end{align*}
In view identities \eqref{eq:J:k.Phi}, \eqref{frak.C} and Proposition \ref{prop:approx.sol}, we have 
\begin{align*}
\mathcal{E}_N(s)\leq &\,C\bigg[\mathring{\varepsilon}^2+\int^{s_*}_{s_0}
e^{\frac{7}{2}Hs}\big\{\|\widetilde{\mathfrak{J}}_\Phi\|_{W^{N,\infty}(\mathbb{R}\times\mathbb{S}^2,\mathring{g})}^2
+\|(\widetilde{\mathfrak{J}}_k)_i{}^j\|_{W^{N,\infty}(\mathbb{R}\times\mathbb{S}^2,\mathring{g})}^2
+\|\widetilde{\mathfrak{C}}_i\|_{W^{N,\infty}(\mathbb{R}\times\mathbb{S}^2,\mathring{g})}^2 \ud s\\
&+\int^{s_b}_{s_*}
e^{\frac{7}{2}Hs}\big\{\|\widetilde{\mathfrak{J}}_\Phi\|_{W^{N,\infty}(\mathbb{R}\times\mathbb{S}^2,\mathring{g})}^2
+\|(\widetilde{\mathfrak{J}}_k)_i{}^j\|_{W^{N,\infty}(\mathbb{R}\times\mathbb{S}^2,\mathring{g})}^2
+\|\widetilde{\mathfrak{C}}_i\|_{W^{N,\infty}(\mathbb{R}\times\mathbb{S}^2,\mathring{g})}^2\big\}\ud s\bigg]\\
\leq &\, C\big[\mathring{\varepsilon}^2+\widetilde{\varepsilon}^2(s_*-s_0)+e^{-\frac{1}{2}Hs_*}\big] \,.
\end{align*}
Moreover, by Lemmas \ref{lem:Hw}, \ref{lem:norm.equiv} we obtain
\begin{align*}
\|\widehat{g}\|_{W^{2,\infty}(\mathbb{R}\times\mathbb{S}^2,\mathring{g})}\leq&\, Ce^{2Hs}\|\widehat{g}\|_{W^{2,\infty}(\Sigma_s,g)}
    \leq Ce^{2Hs}\|\widehat{g}\|_{H^4(\Sigma_s)}\\
    \leq&\, Ce^{2Hs}\sqrt{\mathcal{E}_N(s)}
\leq Ce^{2Hs}\sqrt{\mathring{\varepsilon}^2+\widetilde{\varepsilon}^2(s_*-s_0)+e^{-\frac{1}{2}Hs_*}},
\end{align*}
and similarly
\begin{align*}
  \|\widehat{g}^{-1}\|_{W^{2,\infty}(\mathbb{R}\times\mathbb{S}^2,\mathring{g})}\leq
Ce^{-2Hs}\sqrt{\mathring{\varepsilon}^2+\widetilde{\varepsilon}^2(s_*-s_0)+e^{-\frac{1}{2}Hs_*}}.
\end{align*}

Now choose the initial data sufficiently close to Kerr de Sitter to begin with, such that $C\mathring{\varepsilon}^2<\varepsilon^2/3$. Also, choose $s_*$ such that $e^{-\frac{1}{2}Hs_*}<\varepsilon^2/3$. Lastly, assume that the two Kerr de Sitter pairs of parameters $(a_1,m_1),(a_2,m_2)$ are sufficiently close
such that $\widetilde{\varepsilon}^2(s_*-s_0)<\varepsilon^2/3$, to deduce that
\begin{align*}
C\big[\mathring{\varepsilon}^2+\widetilde{\varepsilon}^2(s_*-s_0)+e^{-\frac{1}{2}Hs_*}\big] <\varepsilon^2.
\end{align*}
Combining the above inequalities yields an improvement of our the bootstrap assumptions \eqref{Boots}. By standard continuation criteria, we infer that the bootstrap time $s_b=+\infty$. Thus, the energy estimate \eqref{Boots:en.est}, and therefore \eqref{glob:en.est} that we have just derived using the \eqref{Boots:en.est}, hold true for all $s\in[s_0,+\infty)$. In particular, the perturbed solution exists globally.
\end{proof}

\begin{remark}[Local well-posedness]
In the proof of the previous corollary, we tacitly used the local well-posedness of the Einstein vacuum equations in the parabolic gauge \eqref{lapse}, which is proven independently in  Appendix \ref{sec:app}. 
\end{remark}

\section{Future stability estimates}
\label{sec:future}

Our goal in this section is to derive the energy estimates that prove Theorem~\ref{thm:Boots}.
We derive the higher order equations in Section~\ref{sec:higherorder}, and treat the error estimates in Section~\ref{subsec:error.est}.
The overall energy estimate is proven in Proposition~\ref{prop:main.en.est} at the end of Section \ref{subsec:en.est}.
We start as an introduction with the basic energy identity in this gauge.

\subsection{Discussion of the energy identity}\label{subsec:en.id.discussion}

The purpose of this section is to explain that $\mathcal{E}(s):=\mathcal{E}_0(s)$ as defined in \eqref{en.def}, at zeroth order $N=0$ for simplicity, is a suitable energy for the system of equations presented in Section~\ref{sec:vareq}.

We begin with the first variation equations given in Lemma~\ref{lem:hat.eq:1st}. The terms on the RHS of these equations will be collected as an error:
\begin{subequations}
\begin{align}
    \partial_s\hg_{ij}-2H\hg_{ij}=&\,(\Er{\hg})_{ij}\\
    \partial_s\hg^{ij}+2H\hg^{ij}=&\,(\Er{\hg^{-1}})^{ij}
\end{align}
    \end{subequations}
In order to derive an estimate for $\hg$ in $\mathrm{L}^2(\Sigma_s)$
we first derive an equation for
\begin{equation}
    |\hg|_g^2=g^{ii'}g^{jj'}\hg_{ij}\hg_{i'j'}\,.
\end{equation}
To prevent confusion, we point out that this cannot be abbreviated to $\hg^{ij}\hg_{ij}$ because we have defined in \eqref{diff}:
\begin{equation}
    \hg^{ij}=g^{ij}-\tg^{ij}
\end{equation}
In the currrent setting, the first variation equation \eqref{1stvar} can also be written as:
\begin{equation}
    \partial_s g_{ij}-2H g_{ij}=(\Er{g})_{ij}
\end{equation}
Similarly for $g^{-1}$. Indeed, it follows from \eqref{1stvar.inv.2} that
\begin{align}
    \partial_s g^{ij}+2Hg^{ij}=&\, (\Er{g^{-1}})^{ij}\,,\\
    (\Er{g^{-1}})^{ij}=& -2\Phi g^{ia}\hk_{a}{}^j-2\Phi g^{ia}\bigl(\tk_a{}^j-H\delta_a{}^j \Phi^{-1}\bigr)\,.
\end{align}
Therefore,
\begin{equation}
    \label{eq:hg:g}\partial_s|\hg|_g^2=\Er{g^{-1}}\star\hg\star\hg+\Er{\hg}\star\hg
\end{equation}
where we have introduced the schematic notation $\star$ to denote all possible contractions of indices with $g$.
In Lemma~\ref{lem:en.id.g}, a higher order version of this identity will be derived.

Next we derive the energy identity for $\hPhi$, which indicates in particular at which rate $\hPhi$ decays.
Treating the terms on the RHS of Lemma~\ref{lem:hat.2nd} as error terms, we recall from \eqref{Phi.hat.eq} the equation for $\hPhi$:
\begin{equation}
    \partial_s\hPhi-\Delta_g\hPhi+2H\hPhi=\mathrm{Error}_{\widehat{\Phi}}
\end{equation}
After multiplying by $e^{4Hs}\hPhi$, differentiating by parts, and rearranging the terms we obtain 
\begin{equation}  \label{eq:hPhi:Er}
\frac{1}{2}\partial_s\bigl(e^{4Hs}\hPhi^2\bigr)+e^{4Hs}|\nabla \hPhi |^2_g=e^{4Hs}\nabla^i\bigl(\hPhi\partial_i\hPhi\bigr)+(\Er{\hPhi}) e^{4Hs} \hPhi \,.
\end{equation}
The higher order version of this equation is given in \eqref{Phi.hat.eq.diff}.

Moreover, from Lemma~\ref{lem:hat.2nd} we rewrite the equations for $\hGamma$ and $\hk$ in the form
\begin{subequations}
\label{eq:hGamma:hk:Er}
\begin{align}
    \partial_s\hGamma_{ic}^a=&\,\Phi \nabla_i \hk_c{}^a+\Phi \nabla_c\hk_i{}^a-g_{cj}\Phi \nabla^a\hk_i{}^j+(\Er{\hGamma})_{ic}^a   \label{eq:hGamma:Er}\,,\\
    \label{eq:hk:Er}\partial_s\hk_{i}{}^j+3H\hk_i{}^j=&\,g^{cj}\nabla_i\nabla_c\hPhi+\frac{1}{3}\Phi g^{cj}\bigl(\nabla_c\hGamma_{ia}^a-\nabla_a\hGamma_{ci}^a\bigr)+\frac{2}{3}\Phi g^{ab}\bigl(\nabla_a\hGamma_{bi}^j-\nabla_i\hGamma_{ab}^j\bigr)\\
    \notag&+(\Er{\hk})_i{}^j \,.
\end{align}
\end{subequations}
As opposed to the equations for $\hg$ or $\hPhi$, which can be discussed separately,
the equations \eqref{eq:hGamma:hk:Er} need to be considered jointly,  to uncover the hyperbolic structure in this formulation. Similarly to \eqref{eq:hg:g},
we begin by deriving the equation for
\begin{equation}
    |\hGamma|_g^2=g^{ii'}g^{cc'}g_{aa'}\hGamma_{ic}^a\hGamma_{i'c'}^{a'},
\end{equation}
which follows from \eqref{eq:hGamma:Er}:
\begin{equation}
    \frac{1}{2}\partial_s|\hGamma|_{g}^2+H|\hGamma|_g^2=\boxed{2\Phi g^{cc'} g_{aa'}\hGamma_{ic}^a \nabla^i \hk_{c'}{}^{a'}} 
    \boxed{-\Phi  g^{ii'}\hGamma_{ic}^a\nabla_{a}\hk_{i'}{}^{c}}+\Er{g,g^{-1}}\star\hGamma\star\hGamma+\Er{\hGamma}\star\hGamma
\end{equation}
 Moreover it follows from \eqref{eq:hk:Er} that:
\begin{multline}
    \frac{1}{2}\partial_s |\hk|^2_g +3H|\hk|_g^2 =\hk_i{}^j\nabla^i\nabla_j\hPhi +\frac{1}{3}\Phi g^{ii'}\hk_{i}{}^{j}\bigl(\nabla_j\hGamma_{i'a}^a\boxed{-\nabla_a\hGamma_{j i'}^a}\bigr)\\+\frac{2}{3}\Phi g^{ab} g^{ii'} g_{jj'}\hk_{i}{}^j\bigl(\boxed{\nabla_a\hGamma_{bi'}^{j'}}-\nabla_{i'}\hGamma_{ab}^{j'}\bigr)+\Er{g,g^{-1}}\star\hk\star\hk+\Er{\hk}\star\hk
\end{multline}
The crucial observation is the following:
After multiplying the first equation by a factor of $1/3$, the boxed terms add up to a divergence. In fact,
\begin{multline}
    \frac{1}{2}\partial_s\Bigl(e^{2Hs}|\hk|_g^2+\frac{1}{3}e^{2Hs}|\hGamma|_g^2\Bigr)+2H|\hk|_g^2=\frac{2}{3}\Phi g^{ii'}g_{jj'} \nabla^a\bigl(e^{2Hs}\hGamma_{ai'}^{j'} \hk_i{}^j\bigr)-\frac{1}{3}\Phi g^{ii'}\nabla_a\bigl(e^{2Hs}\hGamma^a_{ij}\hk_{i'}{}^j\bigr)\\
    +e^{2Hs}\hk_i{}^j\nabla^i\nabla_j\hPhi +\frac{1}{3}\Phi e^{2Hs} g^{ii'}\hk_{i}{}^{j}\nabla_j\hGamma_{i'a}^a-\frac{2}{3}\Phi e^{2Hs} g^{ab}  g_{jj'}\hk_{i}{}^j\nabla^{i}\hGamma_{ab}^{j'}\\
+e^{2Hs}\Er{g,g^{-1}}\star(\hGamma\star\hGamma+\hk\star\hk)+e^{2Hs}\Er{\hGamma}\star\hGamma\,.
\end{multline}
It remains to treat the terms in the second line of the above equation. After differentiation by parts, these terms produce divergences of $\hk$, and the constraint equations of Lemma~\ref{lemma:constraint:h} come into play:
\begin{align}
\label{mom.const.exp}\nabla_j\widehat k_i{}^j=&\,\partial_i\widehat{\Phi}+(\Er{\mathrm{div}\hk})_i\\
\label{mom.const.exp2}     \nabla^i\hk_i{}^j=&g^{jc}\partial_c \hPhi+(\Er{\mathrm{div}\hk})^j
\end{align}
The result is an equation of the form
\begin{multline}\label{eq:hk:hGamma:Er}
    \frac{1}{2}\partial_s\Bigl(e^{2Hs}|\hk|_g^2+\frac{1}{3}e^{2Hs}|\hGamma|_g^2\Bigr)+2H|\hk|_g^2=\Phi \:\mathrm{div}_g(e^{2Hs}\hGamma\star\hk)+\nabla^i\Bigl(e^{2Hs}\hk_i{}^j\partial_j\hPhi\Bigr)\\
    -e^{2Hs} |\nabla\hPhi|_g^2 +\frac{1}{3}\Phi e^{2Hs}\nabla^i\hPhi\:\hGamma_{ia}^a-\frac{2}{3}\Phi e^{2Hs} g^{ab}\nabla_j\hPhi\:\hGamma_{ab}^j\\
    +e^{2Hs}\Er{\mathrm{div}\hk}\star\nabla\hPhi+e^{2Hs}\Er{\mathrm{div}\hk,\hGamma}\star\hGamma+e^{2Hs}\Er{g,g^{-1}}\star(\hGamma\star\hGamma+\hk\star\hk)
\end{multline}
The higher order version of this equation is the content of Lemma~\ref{lem:en.id}.

It is now clear from \eqref{eq:hg:g}, \eqref{eq:hPhi:Er}, and \eqref{eq:hk:hGamma:Er} that a suitable energy for this system is indeed
\begin{equation} \label{eq:en:zero}
\mathcal{E}(s)=\,\|\hg\|^2_{L^2(\Sigma_s,g)}+\|\hg^{-1}\|^2_{L^2(\Sigma_s,g)}
+e^{3Hs}\|\hPhi\|^2_{L^2(\Sigma_s,g)}+e^{2Hs}\big(\|\hGamma \|^2_{L^2(\Sigma_s,g)}
+\|\hk\|^2_{L^2(\Sigma_s,g)}\big)\,.
\end{equation}
Note however that in comparison to \eqref{eq:hPhi:Er},
the rate for $\hPhi$ included in the energy is not sharp. This gap is needed to close the energy estimates; see for example the error estimates of Lemma~\ref{lem:frak.est} below.
The higher order version of the energy defined above is precisely \eqref{en.def}. Nevertheless, sharp asymptotics for $\widehat{\Phi}$ and the rest of the variables are derived in Section \ref{sec:prec.asym}, after we have completed the overall energy argument.

For the energy \emph{estimates}, 
note already that $\nabla\hPhi$ appears on the RHS of \eqref{eq:hk:hGamma:Er} at one order of differentiability higher than in the energy \eqref{eq:en:zero}. Here the positive term on the LHS of \eqref{eq:hPhi:Er} is used.
However, to proceed we need to estimate in the first place various errors. This will be done systematically in Section~\ref{subsec:error.est}.

\subsection{The differentiated equations and higher order energy identities}
\label{sec:higherorder}

\subsubsection{Higher order equations}

To derive higher order energy estimates,
we first commute the first and second variation equations of Section~\ref{sec:vareq} with tangential derivatives.

\begin{lemma}[Commuted first variation equations]
\begin{subequations}
    \begin{align}
\label{g.hat.eq.diff}
\partial_s\nabla^{(\ell)}\widehat{g}_{ij}-2H\nabla^{(\ell)}\widehat{g}_{ij}=&\,(\mathrm{Error}_{\widehat g,\ell})_{ij},\\
\label{g.inv.hat.eq.diff}\partial_s\nabla^{(\ell)}\widehat{g}^{ij}+2H\nabla^{(\ell)}\widehat{g}^{ij}=&\,(\mathrm{Error}_{\widehat g^{-1},\ell})^{ij}\,,
\end{align}
\end{subequations}
where
\begin{subequations}
\label{Error.gg}
\begin{align}
\label{Error.g}(\mathrm{Error}_{\widehat g,\ell})_{ij}=&\,\nabla^{(\ell)}\big\{2\Phi g_{ja}\widehat{k}_i{}^a+
2\widehat{\Phi}g_{ja}\widetilde{k}_i{}^a
+2\widetilde{\Phi}(\widetilde{k}_i{}^a-H\widetilde{\Phi}^{-1}\delta_i{}^a)\widehat{g}_{ja}\big\}
+[\partial_s,\nabla^{(\ell)}]\widehat{g}_{ij},\\
\label{Error.g.inv}(\mathrm{Error}_{\widehat g^{-1},\ell})^{ij}=&-\nabla^{(\ell)}\big\{2\Phi g^{ia}\widehat{k}_a{}^j
+2\widehat{\Phi}g^{ia}\widetilde{k}_a{}^j
+2\widetilde{\Phi}(\widetilde{k}_a{}^j-H\widetilde{\Phi}^{-1}\delta_a{}^j)\widehat{g}^{ia}\big\}+[\partial_s,\nabla^{(\ell)}]\widehat{g}^{ij}\,.
\end{align}
\end{subequations}

Moreover,
\begin{equation}
\label{Phi.hat.eq.diff}\partial_s\nabla^{(\ell)}\widehat{\Phi}-\Delta_g\nabla^{(\ell)}\widehat{\Phi}+2H\nabla^{(\ell)}\widehat{\Phi}=\,\mathrm{Error}_{\widehat \Phi,\ell},
\end{equation}
where
\begin{equation}
\label{Error.Phi}\mathrm{Error}_{\widehat \Phi,\ell}=\,\nabla^{(\ell)}(\mathfrak{F}+\widetilde{\mathfrak{I}}_\Phi)
+[\partial_s,\nabla^{(\ell)}]\hPhi-[\Delta_g,\nabla^{(\ell)}]\hPhi\,.    
\end{equation}

\end{lemma}

\begin{proof}
    These equations result by  commuting the equations \eqref{g.hat.eq}, \eqref{g.inv.hat.eq}, \eqref{Phi.hat.eq}  with $\nabla^{(\ell)}$.
\end{proof}

\begin{lemma}[Commuted second variation equations] \label{lem:sec.var.commuted}
%
\begin{equation}
\begin{split}
\label{k.hat.eq.diff}\partial_s\nabla^{(\ell)}\widehat{k}_i{}^j+3H\nabla^{(\ell)}\widehat{k}_i{}^j=&
\,\frac{1}{3}\Phi g^{cj}(\nabla_c\nabla^{(\ell)}\widehat\Gamma_{ia}^a-\nabla_a\nabla^{(\ell)}\widehat\Gamma_{ci}^a)+(\mathrm{Error}_{\widehat k,\ell})_i{}^j\\
&+\frac{2}{3}\Phi g^{ab}(\nabla_a\nabla^{(\ell)}\widehat\Gamma^j_{bi}-\nabla_i\nabla^{(\ell)}\widehat\Gamma^j_{ab})
+g^{cj}\nabla_i\nabla^{(\ell)}\nabla_c\widehat\Phi\,,
\end{split}
\end{equation}
where
\begin{align}
\label{Error.k}(\mathrm{Error}_{\widehat k,\ell})_i{}^j
=&\,\nabla^{(\ell)}\big\{\mathfrak{K}_i{}^j-\widetilde{\Phi}(\widetilde{k}_l{}^l-3H\widetilde{\Phi}^{-1})\widehat{k}_i{}^j+(\widetilde{\mathfrak{I}}_k)_i{}^j\big\}\\
\notag&+\frac{1}{3}\Phi g^{cj}\big([\nabla^{(\ell)},\nabla_c]\widehat\Gamma_{ia}^a-[\nabla^{(\ell)},\nabla_a]\widehat\Gamma_{ci}^a\big)
+g^{cj}[\nabla^{(\ell)},\nabla_i]\nabla_c\widehat\Phi\\
\notag&+\frac{2}{3}\Phi g^{ab}\big([\nabla^{(\ell)},\nabla_a]\widehat\Gamma^j_{bi}-[\nabla^{(\ell)},\nabla_i]\widehat\Gamma^j_{ab}\big)+[\partial_s,\nabla^{(\ell)}]\widehat{k}_i{}^j\\
\notag&+\sum_{\ell_1+\ell_2=\ell,\,\ell_2<\ell}\big\{\frac{1}{3}g^{cj}\nabla^{(\ell_1)}\Phi(\nabla^{(\ell_2)}\nabla_c\widehat\Gamma_{ia}^a-\nabla^{(\ell_2)}\nabla_a\widehat\Gamma_{ci}^a)\\ \notag&+\frac{2}{3}g^{ab}\nabla^{(\ell_1)}\Phi (\nabla^{(\ell_2)}\nabla_a\widehat\Gamma^j_{bi}-\nabla^{(\ell_2)}\nabla_i\widehat\Gamma^j_{ab})\big\}\,.
\end{align}
Moreover
\begin{equation}
\label{Gamma.hat.eq.diff}\partial_s\nabla^{(\ell)}\widehat{\Gamma}_{ic}^a=\Phi\nabla_i\nabla^{(\ell)}\widehat{k}_c{}^a
+\Phi\nabla_c\nabla^{(\ell)}\widehat{k}_i{}^a-g^{ab}g_{cj}\Phi\nabla_b\nabla^{(\ell)}\widehat{k}_i{}^j+(\mathrm{Error}_{\widehat \Gamma,\ell})_{ic}^a\,,
\end{equation}
where
\begin{align}
\notag
 (\mathrm{Error}_{\widehat \Gamma,\ell})_{ic}^a=&\,\nabla^{(\ell)}\mathfrak{G}_{ic}^a
+[\partial_s,\nabla^{(\ell)}]\widehat{\Gamma}^a_{ic}
+\Phi[\nabla^{(\ell)},\nabla_i]\widehat{k}_c{}^a
+\Phi[\nabla^{(\ell)},\nabla_c]\widehat{k}_i{}^a\\
   \label{Error.Gamma}&-g^{ab}g_{cj}\Phi[\nabla^{(\ell)},\nabla_b]\widehat{k}_i{}^j
+\sum_{\ell_1+\ell_2=\ell,\,\ell_2<\ell}\big\{\nabla^{(\ell_1)}\Phi\nabla^{(\ell_2)}\nabla_i\widehat{k}_c{}^a
+\nabla^{(\ell_1)}\Phi\nabla^{(\ell_2)}\nabla_c\widehat{k}_i{}^a\\
\notag&-g^{ab}g_{cj}\nabla^{(\ell_1)}\Phi\nabla^{(\ell_2)}\nabla_b\widehat{k}_i{}^j\big\}\,.
\end{align}

\end{lemma}

\begin{proof}
These result by commuting the equations \eqref{k.hat.eq} and \eqref{Gamma.hat.eq} with $\nabla^{(\ell)}$.
\end{proof}

Finally, we also commute  the constraint equations \eqref{mom.const.hat}, \eqref{mom.const.hat.ii} with $\nabla^{(\ell)}$.

\begin{lemma}[Commuted constraint equations]
%
\begin{subequations}
\label{mom.const.diff}
\begin{align}
\nabla_j\nabla^{(\ell)}\widehat k_i{}^j=&\,\nabla_i\nabla^{(\ell)}\widehat{\Phi}+(\mathrm{Error}_{\mathrm{div}\widehat k,\ell})_i,\\
\label{mom.const.diff2}g^{im}\nabla_m\nabla^{(\ell)}\hk_i{}^j=&\,g^{jc}\nabla_c\nabla^{(\ell)} \hPhi+(\mathrm{Error}_{\mathrm{div}\widehat k,\ell})^j,
\end{align}
\end{subequations}
where
\begin{subequations} \label{Error.div.kk}
\begin{align}
\label{Error.divk}(\mathrm{Error}_{\mathrm{div}\widehat k,\ell})_i=&\,\nabla^{(\ell)}\big\{\widetilde{\mathfrak{C}}_i-\widehat{\Gamma}_{jc}^j(\widetilde{k}_i{}^c-H\delta_i{}^c)+\widehat{\Gamma}_{ji}^c(\widetilde{k}_c{}^j-H\delta_c{}^j)\big\}\\
\notag&+[\nabla^{(\ell)},\nabla_i]\widehat{\Phi}+[\nabla_j,\nabla^{(\ell)}]\widehat{k}_i{}^j,\\
\label{Error.divk.2}(\mathrm{Error}_{\mathrm{div}\widehat k,\ell})^j=&\,\nabla^{(\ell)}\big\{\hg^{jc}\tnabla_c(\tk_l{}^l-3H)-\hg^{im}\tnabla_m(\tk_i{}^j-H\delta_i{}^j)
-g^{im}\widehat{\Gamma}_{mc}^j(\widetilde{k}_i{}^c-H\delta_i{}^c)\\
\notag&+g^{im}\widehat{\Gamma}_{mi}^c(\widetilde{k}_c{}^j-H\delta_c{}^j)
+\widetilde{\mathfrak{C}}^j\big\}
+g^{jc}[\nabla^{(\ell)},\nabla_c]\widehat{\Phi}+g^{im}[\nabla_m,\nabla^{(\ell)}]\widehat{k}_i{}^j.
\end{align}
\end{subequations}
\end{lemma}

\subsubsection{Higher order energy identities}
\label{sec:higherorder.en.id}

Next, we write the energy identities for the above system of equations, at the level of pointwise magnitudes $|\nabla^\ell\hg|_g^2$, $|\nabla^\ell \hPhi|_g^2$, and so forth.

\begin{lemma}[Energy identity for $\hg$]
\label{lem:en.id.g}
    \begin{multline} \label{en.id.g}
\sum_{\ell\leq N}\frac{1}{2}\partial_s\Bigl\{e^{2\ell Hs}|\nabla^{(\ell)}\widehat{g}|_g^2
+e^{2\ell Hs}|\nabla^{(\ell)}\widehat{g}^{-1}|_g^2\Bigr\}
\\
=\sum_{\ell\leq N}e^{2\ell Hs}\Bigl\{(\Phi k-H\delta)\star\nabla^{(\ell)}\widehat{g}\star\nabla^{(\ell)}\widehat{g}
+(\Phi k-H\delta)\star\nabla^{(\ell)}\widehat{g}^{-1}\star\nabla^{(\ell)}\widehat{g}^{-1}\\
+\mathrm{Error}_{\widehat{g},\ell}\star\nabla^{(\ell)}\widehat{g}
+\mathrm{Error}_{\widehat{g}^{-1},\ell}\star\nabla^{(\ell)}\widehat{g}^{-1}\Bigr\}\,. 
\end{multline}
\end{lemma}

\begin{proof}
To differentiate 
\begin{equation}
    e^{2\ell Hs}|\nabla^{(\ell)}\widehat{g}|_g^2=e^{2\ell Hs}g^{ii'}g^{jj'}g^{b_1b_1'}\ldots g^{b_\ell b_\ell'}\nabla_{b_1}\ldots\nabla_{b_\ell}\widehat{g}_{ij}\nabla_{b_1'}\ldots\nabla_{b_\ell'}\widehat{g}_{i'j'}
\end{equation}
we use \eqref{1stvar.inv.2} and \eqref{g.hat.eq.diff} and obtain
\begin{equation}
    \begin{split}        
\frac{1}{2}\partial_s\big\{e^{2\ell Hs}|\nabla^{(\ell)}\widehat{g}|_g^2\big\}
=&\,\ell H e^{2\ell Hs}|\nabla^{(\ell)}\widehat{g}|_g^2\\
&-\Phi k_a{}^{b_1'}e^{2\ell Hs}g^{ii'}g^{jj'}g^{b_1a}\ldots g^{b_\ell b_\ell'}\nabla_{b_1}\ldots\nabla_{b_\ell}\widehat{g}_{ij}\nabla_{b_1'}\ldots\nabla_{b_\ell'}\widehat{g}_{i'j'}\\
&-\ldots-\Phi k_a{}^{b_\ell'}e^{2\ell Hs}g^{ii'}g^{jj'}g^{b_1b_1'}\ldots g^{b_\ell a}\nabla_{b_1}\ldots\nabla_{b_\ell}\widehat{g}_{ij}\nabla_{b_1'}\ldots\nabla_{b_\ell'}\widehat{g}_{i'j'}\\
-&\Phi k_a{}^{i'}e^{2\ell Hs}g^{ia}g^{jj'}g^{b_1b_1'}\ldots g^{b_\ell b_\ell'}\nabla_{b_1}\ldots\nabla_{b_\ell}\widehat{g}_{ij}\nabla_{b_1'}\ldots\nabla_{b_\ell'}\widehat{g}_{i'j'}\\
-&\Phi k_a{}^{j'}e^{2\ell Hs}g^{ii'}g^{ja}g^{b_1b_1'}\ldots g^{b_\ell b_\ell'}\nabla_{b_1}\ldots\nabla_{b_\ell}\widehat{g}_{ij}\nabla_{b_1'}\ldots\nabla_{b_\ell'}\widehat{g}_{i'j'}\\
+e^{2\ell Hs}g^{ii'}g^{jj'}&g^{b_1b_1'}\ldots g^{b_\ell b_\ell'}\big\{2H\nabla_{b_1}\ldots\nabla_{b_\ell}\widehat{g}_{ij}+(\mathrm{Error})_{\widehat{g},\ell})_{ij}\big\}\nabla_{b_1'}\ldots\nabla_{b_\ell'}\widehat{g}_{i'j'}\,.
    \end{split}
\end{equation}
Since 
\begin{equation}
    \Phi k_a^b=\Phi\hk_a{}^b+\Phi\tk_a{}^b-H\delta_a{}^b  +H\delta_a{}^b  \,,
\end{equation}
the diagonal terms cancel and we are left with
\begin{equation}
    \frac{1}{2}\partial_s\big\{e^{2\ell Hs}|\nabla^{(\ell)}\widehat{g}|_g^2\big\}=\,e^{2\ell Hs}\Bigl\{(\Phi \hk+\Phi\tk-H\delta)\star\nabla^{(\ell)}\widehat{g}\star\nabla^{(\ell)}\widehat{g}+\mathrm{Error}_{\widehat{g},\ell}\star\nabla^{(\ell)}\widehat{g}\Bigr\}\,.
\end{equation}
Similarly for $\hg^{-1}$.
\end{proof}

%
\begin{lemma}[Energy identity for $\hPhi$]
    \label{lem:en.id.hPhi}
%
\begin{multline}\label{en.id.Phi}
\sum_{\ell\leq N}\bigg[\frac{1}{2}\partial_s\big\{e^{3Hs}e^{2\ell Hs}|\nabla^{(\ell)}\widehat{\Phi}|^2_g\big\}
+e^{3Hs}e^{2\ell Hs}|\nabla^{(\ell+1)}\widehat{\Phi}|^2_g
+2He^{3Hs}e^{2\ell Hs}|\nabla^{(\ell)}\widehat{\Phi}|^2_g\bigg]\\
=\sum_{\ell\leq N}e^{3Hs}e^{2\ell Hs}\big\{\nabla^i\big[(\nabla_i\nabla^{(\ell)}\widehat{\Phi})\nabla^{(\ell)}\widehat{\Phi}\big]
+(\Phi k-H\delta)\star\nabla^{(\ell)}\widehat{\Phi}\star \nabla^{(\ell)}\widehat{\Phi}+\mathrm{Error}_{\widehat{\Phi},\ell}\star\nabla^{(\ell)}\widehat{\Phi}\big\}
\end{multline}
\end{lemma}

\begin{proof}
    We multiply equation \eqref{Phi.hat.eq.diff} with $e^{3Hs}e^{2\ell Hs}\nabla^{(\ell)}\widehat{\Phi}$, differentiate by parts in $\partial_s$, and differentiate by parts in $\nabla_i$ the term with $\ell+2$ spatial derivatives;
 contract all corresponding pairs of indices using the metric, and sum in $\ell\leq N$.
The correction term with factor $\Phi k- H\delta$ arises in the integration by parts in $\partial_s$ as in the proof of Lemma~\ref{lem:en.id.g}.
\end{proof}

\begin{lemma}[Energy identity for $\widehat{\Gamma}$  and $\widehat{k}$]
\label{lem:en.id}
%
%
\begin{align}\label{en.id.Gamma.k}
\notag\sum_{\ell\leq N}\bigg[\frac{1}{2}\partial_s\big\{\frac{1}{3}e^{2Hs}e^{2\ell Hs}|\nabla^{(\ell)}\widehat{\Gamma}|^2_g
+e^{2Hs}e^{2\ell Hs}|\nabla^{(\ell)}\widehat{k}|^2_g\big\}
+2He^{2Hs}e^{2\ell Hs}|\nabla^{(\ell)}\widehat{k}|^2_g\bigg]\\
\notag=\sum_{\ell\leq N}\mathrm{Div}_{\hGamma,\hk,\ell}
+\sum_{\ell\leq N}e^{2Hs}e^{2\ell Hs}\bigg[\frac{2}{3}\Phi g^{ab}g_{jj'}\big\{g^{j'c}\nabla_c\nabla^{(\ell)} \hPhi+(\mathrm{Error}_{\mathrm{div}\widehat k,\ell})^{j'}\big\}\star\nabla^{(\ell)}\widehat{\Gamma}_{ab}^j\\
-\frac{1}{3}\Phi g^{ii'}\big\{\nabla_{i'}\nabla^{(\ell)}\widehat{\Phi}+(\mathrm{Error}_{\mathrm{div}\widehat k,\ell})_{i'}\big\}\star\nabla^{(\ell)}\widehat{\Gamma}_{ia}^a
-\big\{\nabla^c\nabla^{(\ell)} \hPhi+(\mathrm{Error}_{\mathrm{div}\widehat k,\ell})^c\big\}\star\nabla^{(\ell)}\nabla_c\widehat{\Phi}\bigg]\\
\notag+\sum_{\ell\leq N}e^{2Hs}e^{2\ell Hs}\big\{(\Phi k-H\delta)\star\nabla^{(\ell)}\widehat{k}\star\nabla^{(\ell)}\widehat{k}+(\Phi k-H\delta)\star\nabla^{(\ell)}\widehat{\Gamma}\star\nabla^{(\ell)}\widehat{\Gamma}\\
\notag+\mathrm{Error}_{\widehat{\Gamma},\ell}\star\nabla^{(\ell)}\widehat{\Gamma}+\mathrm{Error}_{\widehat{k},\ell}\star\nabla^{(\ell)}\widehat{k}\big\},
\end{align}
where
\begin{multline}\label{eq:Div.l}
\mathrm{Div}_{\hGamma,\hk,\ell}=e^{2Hs}e^{2\ell Hs}\bigg[
\frac{2}{3}\Phi g^{ii'}g^{cc'}g_{aa'}\nabla_i\big[\nabla^{(\ell)}\widehat{\Gamma}_{i'c'}^{a'}\star\nabla^{(\ell)}\widehat{k}_c{}^a\big]\\
-\frac{1}{3}\Phi g^{ii'}\nabla_b\big[\nabla^{(\ell)}\widehat{\Gamma}_{i'j}^{b}\star\nabla^{(\ell)}\widehat{k}_i{}^j\big]
+\frac{1}{3}\Phi g^{ii'}\nabla_c\big[\nabla^{(\ell)}\widehat{k}_{i'}{}^{c}\star\nabla^{(\ell)}\widehat{\Gamma}_{ia}^a\big]\\-\frac{2}{3}\Phi g^{ab}g^{ii'}g_{jj'}\nabla_i\big[\nabla^{(\ell)}\widehat{k}_{i'}{}^{j'}\star\nabla^{(\ell)}\widehat{\Gamma}_{ab}^j\big]
+g^{ii'}g_{jj'}g^{cj}\nabla_i\big[\nabla^{(\ell)}\widehat{k}_{i'}{}^{j'}\star\nabla^{(\ell)}\nabla_c\widehat{\Phi}\big]\,.
\bigg]
\end{multline}
Here the $\star$ symbol in all schematic expressions above signifies that all relevant indices in these terms are contracted.
\end{lemma}
\begin{proof}
This energy identity follows from the higher order equations of Lemma~\ref{lem:sec.var.commuted}:
First multiply
\[\text{\eqref{Gamma.hat.eq.diff}}\times\frac{1}{3}e^{2Hs}e^{2\ell Hs}\nabla^{(\ell)}\widehat{\Gamma}_{i'c'}^{a'}\]
and contract all corresponding pairs of indices $(i;i')$, $(j;j')$, $(c;c')$  using the metric.
Similarly multiply and contract
\[\text{\eqref{k.hat.eq.diff}}\times e^{2Hs}e^{2\ell Hs}\nabla^{(\ell)}\widehat{k}_{i'}{}^{j'}\,,\] 
and sum up the resulting equations.
After differentiating by parts in $\partial_s$,
we obtain the principal terms on the LHS of \eqref{en.id.Gamma.k},
and the error terms involving $\Phi k-H\delta$ on the RHS.

It remains to differentiate by parts the terms in the RHS which contain $\ell+1$ spatial derivatives of $\widehat{\Gamma}$ and $\widehat{k}$.  This produces on one hand the divergence terms $(\mathrm{Div}_{\hGamma,\hk,\ell})$ in \eqref{eq:Div.l},
and on the other hand, divergences of $\nabla^{(\ell)}k$ which we can replace using the higher order  constraint equations \eqref{mom.const.diff}.  For instance,
\begin{align*}
&g^{ii'}g_{jj'}e^{2Hs}e^{2\ell Hs}\nabla^{(\ell)}\widehat{k}_{i'}{}^{j'}\frac{1}{3}\Phi g^{cj}\nabla_c\nabla^{(\ell)}\widehat{\Gamma}_{ia}^a=\\
=&\,e^{2Hs}e^{2\ell Hs}\big\{\frac{1}{3}\Phi g^{ii'}\nabla_c\big[\nabla^{(\ell)}\widehat{k}_{i'}{}^{c}\star\nabla^{(\ell)}\widehat{\Gamma}_{ia}^a\big]-\frac{1}{3}\Phi g^{ii'}(\nabla_c\nabla^{(\ell)}\widehat{k}_{i'}{}^{c})\star\nabla^{(\ell)}\widehat{\Gamma}_{ia}^a\big\}\\
=&\,e^{2Hs}e^{2\ell Hs}\bigg[\frac{1}{3}\Phi g^{ii'}\nabla_c\big[\nabla^{(\ell)}\widehat{k}_{i'}{}^{c}\star\nabla^{(\ell)}\widehat{\Gamma}_{ia}^a\big]-\frac{1}{3}\Phi g^{ii'}\big\{\nabla_{i'}\nabla^{(\ell)}\widehat{\Phi}+(\mathrm{Error}_{\mathrm{div}\widehat k,\ell})_{i'}\big\}\star\nabla^{(\ell)}\widehat{\Gamma}_{ia}^a\bigg]
\end{align*}
The rest of the computations are similar and straightforward.
\end{proof}

\subsection{Error estimates}
\label{subsec:error.est}

In this section we estimate the error terms in \eqref{Error.gg},
\eqref{Error.Phi}, \eqref{Error.k} and in \eqref{Error.div.kk}, assuming the bootstrap assumptions \eqref{Boots} are valid. For this purpose, we first derive commutator estimates.
\begin{lemma}\label{lem:comm.est}
Let $\mathcal{T}$ be $\Sigma_s$-tangent $(n,m)$ tensor and let $\ell\leq N$. Then it satisfies:
\begin{subequations}
    \label{comm.est}
\begin{align}
e^{\ell Hs}\|[\partial_s,\nabla^{(\ell)}]\mathcal{T}\|_{L^2(\Sigma_s,g)}\leq&\,Ce^{-Hs}\|\mathcal{T}\|_{H^{\ell-1}(\Sigma_s,g)},\label{comm.est.s}\\
e^{\ell Hs}\|[\nabla_i,\nabla^{(\ell)}]\mathcal{T}\|_{L^2(\Sigma_s,g)}\leq&\,Ce^{-Hs} \|\mathcal{T}\|_{H^{\ell-1}(\Sigma_s,g)},\label{comm.est.nabla}\\
 e^{(\ell+1)Hs}\|[\Delta_g,\nabla^{(\ell)}]\mathcal{T}\|_{L^2(\Sigma_s,g)}\leq&\,Ce^{-Hs}\|\mathcal{T}\|_{H^\ell(\Sigma_s,g)}, \label{comm.est.lap}
\end{align}
\end{subequations}
for all $s\in[s_0,s_b)$.
\end{lemma}
\begin{proof}
First, we derive formulas for the commutators.
Commuting $\partial_s$ with $\nabla$ applied to $\mathcal{T}$ gives:
\begin{align*}
\partial_s\nabla_b\mathcal{T}^{j_1\ldots j_n}_{i_1\ldots i_m}
=\,\partial_s\big\{\partial_b\mathcal{T}^{j_1\ldots j_n}_{i_1\ldots i_m}+&\Gamma_{bc}^{j_1}\mathcal{T}^{c\ldots j_n}_{i_1\ldots i_m}+\ldots+\Gamma_{bc}^{j_n}\mathcal{T}^{j_1\ldots c}_{i_1\ldots i_m}\\
&-\Gamma_{bi_1}^c\mathcal{T}^{j_1\ldots j_n}_{c\ldots i_m}-\ldots-\Gamma_{bi_m}^c\mathcal{T}^{j_1\ldots j_n}_{i_1\ldots c}\big\}\\
=\,\nabla_b\partial_s\mathcal{T}^{j_1\ldots j_n}_{i_1\ldots i_m}
+&\partial_s\Gamma_{bc}^{j_1}\mathcal{T}^{c\ldots j_n}_{i_1\ldots i_m}+\ldots+\partial_s\Gamma_{bc}^{j_n}\mathcal{T}^{j_1\ldots c}_{i_1\ldots i_m}\\
&-\partial_s\Gamma_{bi_1}^c\mathcal{T}^{j_1\ldots j_n}_{c\ldots i_m}-\ldots-\partial_s\Gamma_{bi_m}^c\mathcal{T}^{j_1\ldots j_n}_{i_1\ldots c}.
\end{align*}
Using \eqref{Gamma.eq} to replace $\partial_s\Gamma$, we have schematically
\begin{align*}
[\partial_s,\nabla]\mathcal{T}
=\nabla(\Phi k)\star\mathcal{T},
\end{align*}
and by induction on $\ell$:
\begin{align*}
[\partial_s,\nabla^{(\ell)}]\mathcal{T}
=\sum_{\ell_1+\ell_2=\ell,\,\ell_2<\ell}\nabla^{\ell_1}(\Phi k)\star\nabla^{\ell_2}\mathcal{T}
\end{align*}
For the estimate \eqref{comm.est.s},
let us make a case distinction for the terms in this sum, depending on $\ell_1$.
For $0<\ell_1\leq N-2$ we have
\begin{equation}
   e^{\ell Hs} \|\nabla^{\ell_1}(\Phi k)\star\nabla^{\ell_2}\mathcal{T}\|_{L^2}\leq e^{\ell_1 Hs}\|\nabla^{(\ell_1)}(\Phi k)\|_{L^\infty} \|\mathcal{T}\|_{H^{\ell_2}}\leq C e^{-Hs} \|\mathcal{T}\|_{H^{\ell-1}}\,,
\end{equation}
where we have used Lemma~\ref{lem:WN-2.infty}, Lemma~\ref{lem:ref.properties}, and Lemma \ref{lem:norm.equiv}.
For $\ell_1>N-2$ the estimate still holds for those terms in 
\begin{equation}
    \nabla^{(\ell_1)}\bigl(\Phi k\bigr)= \nabla^{(\ell_1)}\bigl(\hPhi \hk+\tPhi \hk+\hPhi \tk+\tPhi \tk\bigr)\,,
\end{equation}
involving at most $N-2$ derivatives of $\hPhi$,
and $\hk$. For the remaining terms, with more that $N-2$ derivatives of $\hPhi$ or $\hk$,
we use the bootstrap assumption \eqref{Boots} on the energy,  and Lemma~\ref{Sob} for $\mathcal{T}$.
For example:
\begin{equation}
    e^{\ell Hs} \|(\nabla^{(\ell_1)}\hPhi) \tk \star\nabla^{\ell_2}\mathcal{T}\|_{L^2}\leq \|\hPhi\|_{H^{\ell_1}}\|\mathcal{T}\|_{W^{\ell_2,\infty}}\leq C e^{-\frac{3}{2}Hs}\sqrt{\mathcal{E}_{N}(s)}\|\mathcal{T}\|_{H^{\ell_2+2}}\,,
\end{equation}
where $\ell_1\ge N-1\ge3$ and $\ell_2+2=\ell-\ell_1+2\leq\ell-1$.

For the commutation of spatial derivatives we write schematically: 
\begin{align}
[\nabla_i,\nabla^{(\ell)}]\mathcal{T}=&\,\sum_{\ell_1+\ell_2=\ell,\,\ell_2<\ell}\nabla^{(\ell_1)}\Gamma\star\nabla^{(\ell_2)}\mathcal{T}
+\!\!\!\!\sum_{\substack{\ell_1+\ell_2+\ell_3=\ell-1\\\ell_3<\ell}}\!\!\!\!\nabla^{(\ell_1)}\Gamma\star\nabla^{(\ell_2)}\Gamma\star\nabla^{(\ell_3)}\mathcal{T}\\
[\Delta_g,\nabla^{(\ell)}]\mathcal{T}=&\,\nabla^i[\nabla_i,\nabla^{(\ell)}]\mathcal{T}+[\nabla^i,\nabla^{(\ell)}]\nabla_i\mathcal{T}
\end{align}
The estimates \eqref{comm.est.nabla} and \eqref{comm.est.lap} then follow by estimating the above expressions using the bootstrap assumptions \eqref{Boots} and the Lemma \ref{lem:Hw}. 
\end{proof}
Next, we estimate the error terms coming from \eqref{frak.G}, \eqref{frak.K}, \eqref{frak.F}.
\begin{lemma}\label{lem:frak.est}
The expressions $\mathfrak{G}_{ic}^a,\mathfrak{K}_i{}^j,\mathfrak{F}$ satisfy:
\begin{subequations} \label{frak.est}
\begin{align}
 e^{Hs}e^{\ell Hs}\|\nabla^{(\ell)}\mathfrak{G}\|_{L^2(\Sigma_s,g)}\leq&\,Ce^{-\frac{1}{2}Hs}\sqrt{\mathcal{E}_N(s)}+Ce^{-\frac{1}{2}Hs}e^{\frac{3}{2}Hs}e^{\ell Hs}\|\nabla^{(\ell+1)}\widehat{\Phi}\|_{L^2(\Sigma_s,g)}\label{frak.est.G}\\
e^{Hs}e^{\ell Hs}\|\nabla^{(\ell)}\mathfrak{K}\|_{L^2(\Sigma_s,g)}\leq&\,Ce^{-\frac{1}{2}Hs}\sqrt{\mathcal{E}_N(s)}\label{frak.est.K}\\
 e^{\frac{3}{2}Hs}e^{\ell Hs}\|\nabla^{(\ell)}\mathfrak{F}\|_{L^2(\Sigma_s,g)}\leq&\,Ce^{-\frac{1}{2}Hs}\sqrt{\mathcal{E}_N(s)} \label{frak.est.F}
\end{align}
\end{subequations}
for all $s\in[s_0,s_b)$ and $\ell\leq N$.
\end{lemma}
\begin{proof}
We begin with \eqref{frak.est.G}.
Going back to \eqref{frak.G}, we notice the following cancellations in the first line:
\begin{equation}
\begin{split}
\Phi\widehat\nabla_i\widetilde k_c{}^a
&+\Phi\widehat\nabla_c\widetilde k_i{}^a-g^{ab}g_{cj}\Phi\widehat\nabla_b \widetilde k_i{}^j=\\
=&\,\Phi\widehat{\Gamma}_{ib}^a\widetilde{k}_c{}^b
-\Phi\widehat{\Gamma}_{ic}^b\widetilde{k}_b{}^a
+\Phi\widehat{\Gamma}_{cb}^a\widetilde{k}_i{}^b
-\Phi\widehat{\Gamma}_{ic}^b\widetilde{k}_b{}^a
-g^{ab}g_{cj}\Phi(\widehat{\Gamma}_{bc}^j\widetilde{k}_i{}^c-\widehat{\Gamma}_{bi}^c\widetilde{k}_c{}^j)\\
=&\,\Phi\widehat{\Gamma}_{ib}^a(\widetilde{k}_c{}^b-H\delta_c{}^b)
-\Phi\widehat{\Gamma}_{ic}^b(\widetilde{k}_b{}^a-H\delta_b{}^a)
+\Phi\widehat{\Gamma}_{cb}^a(\widetilde{k}_i{}^b-H\delta_i{}^b)
-\Phi\widehat{\Gamma}_{ic}^b(\widetilde{k}_b{}^a-H\delta_b{}^a)\\
&-g^{ab}g_{cj}\Phi\big[\widehat{\Gamma}_{bc}^j(\widetilde{k}_i{}^c-H\delta_i{}^c)-\widehat{\Gamma}_{bi}^c(\widetilde{k}_c{}^j-H\delta_c{}^j)\big]
\end{split}
\end{equation}
Using the bootstrap assumptions \eqref{Boots} and Lemma~\ref{lem:WN-2.infty},
as well as the properties of the reference metric in Lemma~\ref{lem:ref.properties},
we deduce that
\begin{align*}
e^{Hs}e^{\ell Hs}\|\nabla^{(\ell)}(\Phi\widehat\nabla_i\widetilde k_c{}^a
+\Phi\widehat\nabla_c\widetilde k_i{}^a-g^{ab}g_{cj}\Phi\widehat\nabla_b \widetilde k_i{}^j)\|_{L^2(\Sigma_s,g)}\leq Ce^{-2Hs}\sqrt{\mathcal{E}_N(s)},
\end{align*}
which of course is much better than the asserted bound for $\mathfrak{G}_{ic}^a$. The least decaying terms are in the third line of \eqref{frak.G}, which we can write as
\begin{equation}
\begin{split}
    k_c{}^a\nabla_i\widehat\Phi
+k_i{}^a\nabla_c\widehat\Phi-g^{ab}g_{cj}k_i{}^j\nabla_b\widehat\Phi=&\hk_c{}^a\nabla_i\widehat\Phi
+\hk_i{}^a\nabla_c\widehat\Phi-g^{ab}g_{cj}\hk_i{}^j\nabla_b\widehat\Phi\\
+(\tk_c{}^a-H\delta_c{}^a)\nabla_i\widehat\Phi
&+(\tk_i{}^a-H\delta_i{}^a)\nabla_c\widehat\Phi-g^{ab}g_{cj}(\tk_i{}^j-H\delta_i{}^j)\nabla_b\widehat\Phi\\
&+\delta_c{}^a\nabla_i\widehat\Phi
+\delta_i{}^a\nabla_c\widehat\Phi-g^{ab}g_{ci}\nabla_b\widehat\Phi\,.
\end{split}
\end{equation}
Thus,
\begin{multline}
e^{Hs}e^{\ell Hs}\|\nabla^{(\ell)}(k_c{}^a\nabla_i\widehat\Phi
+k_i{}^a\nabla_c\widehat\Phi-g^{ab}g_{cj}k_i{}^j\nabla_b\widehat\Phi)\|_{L^2(\Sigma_s,g)}\\
\leq\, Ce^{-Hs}\sqrt{\mathcal{E}_N(s)}
+Ce^{-\frac{1}{2}Hs}e^{\frac{3}{2}Hs}e^{\ell Hs}\|\nabla^{(\ell+1)}\widehat{\Phi}\|_{L^2(\Sigma_s,g)}\,.
\end{multline}
The rest of the terms in \eqref{frak.G} satisfy better higher order estimates and are treated similarly.

For the inequality \eqref{frak.est.K}, the least decaying terms come from the first line of \eqref{frak.K}:
\begin{align*}
e^{Hs}e^{\ell Hs}\|\nabla^{(\ell)}(\widehat{\Phi}k_l{}^lk_i{}^j+\widetilde{\Phi}\widehat{\Phi}k_i{}^j-\Lambda\delta_i{}^j\widehat{\Phi})\|_{L^2(\Sigma_s,g)}\leq Ce^{-\frac{1}{2}Hs}\sqrt{\mathcal{E}_N(s)}
\end{align*}
The rest of the terms in \eqref{frak.K} satisfy better higher order estimates.

Finally, we turn to the estimate \eqref{frak.est.F} with $\mathfrak{F}$ given by \eqref{frak.F}.
While for the first two terms in $\mathfrak{F}$,
\begin{multline*}
    e^{3Hs}e^{2\ell Hs}\int_{\Sigma_s}f^2(t) |\nabla^{(\ell)}(\hPhi \hk_i{}^j \tk_j{}^i +\hPhi\tk_i{}^j\hk_j{}^i)|_g^2  e^{-3Hs}\vol{g}\leq\\
    \leq C e^{Hs}\|\hPhi\|^2_{W^{N-2,\infty}} e^{2Hs}\|\hk\|_{H^N}^2+C \|\hk\|^2_{W^{N-2,\infty}} e^{3Hs}\|\hPhi\|_{H^N}^2\leq e^{-2Hs}\mathcal{E}_N(s)\,,
\end{multline*}
we encounter the term decaying the least in $\Phi \hk_i{}^j \hk_j{}^i=\hPhi \hk_i{}^j \hk_j{}^i+\tPhi \hk_i{}^j \hk_j{}^i$, and obtain:
    \begin{equation*}
        e^{3Hs}e^{2\ell Hs}\int_{\Sigma_s} f^2(t)|\nabla^{(\ell)}(\tPhi \hk_i{}^j \hk_j{}^i)|_g^2 e^{-3Hs}\vol{g}\leq C e^{Hs}\|\hk\|_{W^{N-2,\infty}}^2 e^{2Hs}\|\hk\|_{H^N}^2\leq C e^{-Hs} \mathcal{E}_N(s)\,.
    \end{equation*}
Here and above we have used Lemma~\ref{lem:WN-2.infty} for the pointwise estimates and \eqref{Boots}. The remaining terms in $\nabla^{(\ell)}\mathfrak{F}$ are estimated similarly now using the pointwise estimates of Lemma~\ref{lem:ref.properties} for the reference metric.

\end{proof}
\begin{proposition}\label{prop:error.est}
Assume the bootstrap assumptions \eqref{Boots} are satisfied for some $N\geq 4$.

For all $s\in[s_0,s_b)$,
\begin{enumerate}
    \item the error terms in \eqref{Error.g} and \eqref{Error.g.inv} satisfy the estimates:
\begin{subequations}
\label{error.est.g}
\begin{align}
     e^{\ell Hs}\|\mathrm{Error}_{\widehat g,\ell}\|_{L^2(\Sigma_s,g)}\leq&\,Ce^{-Hs}\sqrt{\mathcal{E}_N(s)}\\
 e^{\ell Hs}\|\mathrm{Error}_{\widehat {g}^{-1},\ell}\|_{L^2(\Sigma_s,g)}\leq&\,Ce^{-Hs}\sqrt{\mathcal{E}_N(s)}\,,
\end{align}
\end{subequations}

\item the error in \eqref{Error.Phi} satisfies:
\begin{equation}
     e^{\frac{3}{2}Hs}e^{\ell Hs}\|\mathrm{Error}_{\widehat \Phi,\ell}\|_{L^2(\Sigma_s,g)}\leq \,Ce^{-\frac{1}{2}Hs}\sqrt{\mathcal{E}_N(s)}+Ce^{\frac{3}{2}Hs}\|\widetilde{\mathfrak{I}}_\Phi\|_{W^{\ell,\infty}(\mathbb{R}\times\mathbb{S}^2,\mathring{g})}\,.
\end{equation}

\item the error terms in \eqref{Error.k}, \eqref{Error.Gamma} satisfy:
\begin{align}
e^{Hs}e^{\ell Hs}\|\mathrm{Error}_{\widehat k,\ell}\|_{L^2(\Sigma_s,g)}\leq&\,Ce^{-\frac{1}{2}Hs}\sqrt{\mathcal{E}_N(s)}+Ce^{Hs}
\|(\widetilde{\mathfrak{I}}_k)_i{}^j\|_{W^{\ell,\infty}(\mathbb{R}\times\mathbb{S}^2,\mathring{g})}\\
   e^{Hs}e^{\ell Hs}\|\mathrm{Error}_{\widehat \Gamma,\ell}\|_{L^2(\Sigma_s,g)}\leq& \,Ce^{-\frac{1}{2}Hs}\sqrt{\mathcal{E}_N(s)}+Ce^{-\frac{1}{2}Hs}e^{\frac{3}{2}Hs}e^{\ell Hs}\|\nabla^{(\ell+1)}\widehat{\Phi}\|_{L^2(\Sigma_s,g)}
\end{align}

\item and the error terms in \eqref{Error.div.kk} satisfy:
\begin{equation}
 e^{Hs}e^{\ell Hs}\|\mathrm{Error}_{\mathrm{div}\widehat k,\ell}\|_{L^2(\Sigma_s,g)}\leq \,Ce^{-Hs}\sqrt{\mathcal{E}_N(s)}+C\|\widetilde{\mathfrak{C}}_i\|_{W^{\ell,\infty}(\mathbb{R}\times\mathbb{S}^2,\mathring{g})}\,.
\end{equation}

\end{enumerate}

\end{proposition}
\begin{proof}
For \emph{(I)} consider \eqref{Error.g},
\begin{align}
    \Er{\hg,\ell}=&\nabla^{(\ell)}\Er{\hg}+[\partial_s,\nabla^{(\ell)}]\hg\,,\\
    (\Er{\hg})_{ij}=&2\Phi g_{ja}\widehat{k}_i{}^a+
2\widehat{\Phi}g_{ja}\widetilde{k}_i{}^a
+2\widetilde{\Phi}(\widetilde{k}_i{}^a-H\widetilde{\Phi}^{-1}\delta_i{}^a)\widehat{g}_{ja}\,.
\end{align}
For each term in $\Er{\hg}$, we separate the differences, for example:
\begin{equation}
    \Phi g_{ja}\hk_i{}^a=\hPhi \hg_{ja}\hk_i{}^a+\tPhi \hg_{ja}\hk_i{}^a+\hPhi \tg_{ja}\hk_i{}^a+\tPhi \tg_{ja}\hk_i{}^a
\end{equation}
Then with $\ell$ derivatives falling on $\Er{\hg}$, we estimate the highest order of derivatives of $\hk$, $\hg$, or $\hPhi$ in energy, and apply the Sobolev inequalities of Lemma~\ref{lem:WN-2.infty} to the lower orders, exploiting the decay of all quantities except $\hg$, in either norm:
\begin{multline}
    e^{\ell Hs}\|\nabla^{(\ell)}(\hPhi\hg\hk)\|_{L^2}\leq e^{-Hs}\Bigl(\|\hPhi\|_{W^{N-2,\infty}}+\|\hg\|_{W^{N-2,\infty}}
    +\|\hk\|_{W^{N-2,\infty}}\Bigr)e^{Hs}\Bigl(\|\hk\|_{H^N}+\|\hPhi\|_{H^N}\Bigr)\\
    +e^{-Hs}\Bigl(e^{Hs}\|\hPhi\|_{W^{N-2,\infty}}+e^{Hs}\|\hk\|_{W^{N-2,\infty}}\Bigr)\|\hg\|_{H^N}\leq C e^{-Hs} \sqrt{\mathcal{E}_N(s)}\,.
\end{multline}
For the terms involving derivatives of the reference metric, namely $\nabla^{(\ell)}\tPhi$ or $\nabla^{(\ell)}\tg$, we also separate differences using $\nabla=\hnabla+\tnabla$ thus
introducing $\hGamma$ as a quantity that can be treated alongside $\hk$ as above. The remaining terms $\tnabla^{(\ell)}\tPhi$ and $\tnabla^{(\ell)}\tg$ can always be estimated in $L^\infty$. For example,
\begin{equation}
    e^{\ell Hs}\|(\tnabla^{(\ell)}\tPhi)\hg\hk\|_{L^2}\leq \|\tPhi\|_{W^{N,\infty}}\|\hg\|_{L^\infty}\|\hk\|_{L^2}\leq C e^{-Hs} \sqrt{\mathcal{E}_N(s)}\,.
\end{equation}
In this way, all terms in $\nabla^{(\ell)}\Er{\hg}$ can be as asserted. Together with Lemma~\ref{lem:comm.est}, this implies \eqref{error.est.g}.

For \emph{(II)} consider \eqref{Error.Phi}.
The commutator terms in $\Er{\hPhi,\ell}$ are dealt with using Lemma~\ref{lem:comm.est},
and the estimate for $\nabla^{(\ell)}\mathfrak{F}$ is given in Lemma~\ref{lem:frak.est}. Also, the $L^2(\Sigma_s,g)$ norm of $\nabla^{(\ell)}\widetilde{\mathfrak{J}}_\Phi$ can be replaced by the $W^{\ell,\infty}(\mathbb{R}\times\mathbb{S}^2,\mathring{g})$ norm using Lemma \ref{lem:norm.equiv} and the observation that $\widetilde{\mathfrak{J}}_\Phi$ is supported in $\{-1\leq t\leq 1\}$, cf. \eqref{eq:J:k.Phi} and \eqref{approx.sol.supp}.

\emph{(III)} now follows directly by employing Lemma~\ref{lem:frak.est}, and the commutation estimates of Lemma~\ref{lem:comm.est}, together with the bootstrap assumptions \eqref{Boots}, and the Lemma \ref{lem:WN-2.infty}.

\emph{(IV)} follows similarly, using in addition Lemma~\ref{lem:ref.properties} for the behavior of the reference variables in \eqref{Error.div.kk}.
\end{proof}

\subsection{Main energy estimates}\label{subsec:en.est}

In this section we derive the main energy estimates for the variables $\widehat{g},\widehat{g}^{-1},\widehat{k},\widehat{\Gamma}$, using the error estimates in Section~\ref{subsec:error.est} and the energy identities in Section~\ref{sec:higherorder.en.id}
\begin{proposition}\label{prop:main.en.est}
Assume that the bootstrap assumptions \eqref{Boots} are valid for some $N\ge 4$.

Then the following energy estimates hold for all $s\in[s_0,s_b)$:

\begin{enumerate}
    \item For $\widehat{g},\widehat{g}^{-1}$,
\begin{equation}\label{main.en.est.g}
\partial_s\bigl(\|\widehat g\|^2_{H^N(\Sigma_s,g)}+\|\widehat g^{-1}\|^2_{H^N(\Sigma_s,g)}\bigr)\leq Ce^{-Hs}\mathcal{E}_N(s),
\end{equation}
\item for $\widehat{\Phi}$,
\begin{multline}\label{main.en.est.Phi}
\partial_s\big\{e^{3Hs}\|\widehat{\Phi}\|_{H^N(\Sigma_s,g)}^2\big\}
+e^{3Hs}\|\nabla\widehat{\Phi}\|_{H^N(\Sigma_s,g)}^2 \leq\\
\leq\, Ce^{-\frac{1}{2}Hs}\mathcal{E}_N(s)
+C e^{\frac{7}{2}Hs}\|\widetilde{\mathfrak{I}}_\Phi\|_{W^{N,\infty}(\mathbb{R}\times\mathbb{S}^2,\mathring{g})}^2,
\end{multline}
\item and finally for $\widehat{\Gamma},\widehat{k}$,
\begin{multline}\label{main.en.est.Gamma.k}
\partial_s\big\{\frac{1}{3}e^{2Hs}\|\widehat \Gamma\|^2_{H^N(\Sigma_s,g)}+e^{2Hs}\|\widehat k\|^2_{H^N(\Sigma_s,g)}\big\}\leq\\
\leq\,Ce^{-\frac{1}{2}Hs}\mathcal{E}_N(s)
+4e^{-\frac{1}{2}Hs}e^{3Hs}e^{2NHs}\|\nabla^{N+1}\hPhi\|^2_{L^2(\Sigma_s,g)}\\
+Ce^{\frac{5}{2}Hs}\|(\widetilde{\mathfrak{I}}_k)_i{}^j\|_{W^{N,\infty}(\mathbb{R}\times\mathbb{S}^2,\mathring{g})}^2
+Ce^{\frac{1}{2}Hs}\|\widetilde{\mathfrak{C}}_i\|_{W^{N,\infty}(\mathbb{R}\times\mathbb{S}^2,\mathring{g})}^2\,.
\end{multline}

\end{enumerate}

\end{proposition}
\begin{proof}

For the derivation of the energy estimates we frequently use that the second fundamental form $k$ of the solution is diagonal up to a decaying remainder:
\begin{equation}
    |\Phi k_i{}^j-H\delta_i{}^j|\leq C e^{-Hs} \label{k.diag}
\end{equation}
This is proven as follows. Since
\begin{equation}
    \Phi k-H\delta=\hPhi \hk+\tPhi \hk+\hPhi\tk+\tPhi\tk-H\delta\,,
\end{equation}
and for the reference solution $\tPhi\tk_i{}^j=H\delta_i{}^j+\Os{-2}$ by Lemma~\ref{lem:ref.properties},
it remains to bound the differences.
In view of the bootstrap assumptions,
this follows from Lemma~\ref{lem:WN-2.infty}:
\begin{equation}
    e^{\frac{3}{2}Hs}|\hPhi|+e^{Hs}|\hk|\leq C\varepsilon
\end{equation}

Each of the estimates \emph{(I)-(III)} is derived by multiplying the corresponding energy identity in Section~\ref{sec:higherorder.en.id} with the weight $f^2(t)$ in \eqref{Hw} and then integrating on $\Sigma_s$ with respect to   the volume form $e^{-3Hs}\mathrm{vol}_g$. A correction term is generated when $\partial_s$ falls on $e^{-3Hs}\mathrm{vol}_g$, namely
\begin{align}\label{ds.vol}
\partial_s(e^{-3Hs}\mathrm{vol}_g)=(\Phi\,\text{tr}_g k-3H)e^{-3Hs}\mathrm{vol}_g,\qquad |\Phi\,\text{tr}_g k-3H|\leq Ce^{-Hs},
\end{align}
which again follows from \eqref{k.diag}.

\smallskip

\emph{(I)}. Applying the above procedure to the
energy identity of Lemma~\ref{en.id.g}, we have
\begin{multline}
    \sum_{\ell\leq N}e^{2\ell Hs}
    \int_{\Sigma_s}\Bigl\lvert(\Phi k-H\delta)\star\nabla^{(\ell)}\widehat{g}\star\nabla^{(\ell)}\widehat{g}
+(\Phi k-H\delta)\star\nabla^{(\ell)}\widehat{g}^{-1}\star\nabla^{(\ell)}\widehat{g}^{-1} \Bigr\rvert f^2(t) e^{-3Hs}\mathrm{vol}_g\leq\\\leq C e^{-Hs}\mathcal{E}_N(s)
\end{multline}
and
\begin{equation}
\sum_{\ell\leq N}e^{2\ell Hs}
    \int_{\Sigma_s}\Bigl\lvert\mathrm{Error}_{\widehat{g},\ell}\star\nabla^{(\ell)}\widehat{g}
+\mathrm{Error}_{\widehat{g}^{-1},\ell}\star\nabla^{(\ell)}\widehat{g}^{-1}\Bigr\}\Bigr\rvert  f^2(t) e^{-3Hs}\mathrm{vol}_g \leq C e^{-Hs}\mathcal{E}_N(s)
\end{equation}
which follow by Cauchy-Schwarz, the bound \eqref{k.diag}, and the error estimate \eqref{error.est.g} in Proposition~\ref{prop:error.est}.

\smallskip
\emph{(II)}.
For the energy estimate \eqref{main.en.est.Phi}, we repeat the previous argument, using instead the energy identity \eqref{en.id.Phi}, and in addition integrate by parts the terms which take the form of a divergence:
\begin{multline}
\sum_{\ell\leq N}\int_{\Sigma_s}e^{3Hs}e^{2\ell Hs}\nabla^i\big[(\nabla_i\nabla^{(\ell)}\widehat{\Phi})\nabla^{(\ell)}\widehat{\Phi}\big] f^2(t)e^{-3Hs}\mathrm{vol}_g\\
=-\sum_{\ell\leq N}\int_{\Sigma_s}e^{3Hs}e^{2\ell Hs}g^{ii'}(\nabla_i\nabla^{(\ell)}\widehat{\Phi})\nabla^{(\ell)}\widehat{\Phi}\big[\nabla_{i'}f^2(t)\big]e^{-3Hs}\mathrm{vol}_g\\
\leq \,Ce^{-Hs}\mathcal{E}_N(s)+\frac{1}{2}e^{-Hs}e^{3Hs}e^{2N Hs}\|\nabla^{(N+1)}\widehat{\Phi}\|_{L^2(\Sigma_s,g)}^2\,,
\end{multline}
where we used that $|\partial f(t)|\leq (\alpha_1+\alpha_2)f(t)$.
The term with $(N+1)$ derivatives of $\hPhi$ can be absorbed in the LHS thanks to the corresponding favorable term in \eqref{en.id.Phi}.
The stated estimate then follows from Proposition~\ref{prop:error.est} \emph{(II)} and Young's inequality.

\smallskip
\emph{(III)}. For the last energy estimate \eqref{main.en.est.Gamma.k}, we argue similarly, using instead the energy identity \eqref{en.id.Gamma.k}. Integrating by parts the divergence terms in \eqref{eq:Div.l} produces error terms which are controlled as above:
\begin{equation}
 \sum_{\ell\leq N}\int_{\Sigma_s} (\mathrm{Div}_{\hGamma,\hk,\ell}) f^2(t) e^{-3Hs}\vol{g}\leq 
 Ce^{-Hs}\mathcal{E}_N(s)+e^{-Hs}e^{3Hs}e^{2NHs}\|\nabla^{(N+1)}\hPhi\|_{L^2}^2
\end{equation}
For the overall $e^{Hs}$ powers that appear in these estimates using Cauchy-Schwarz, it is useful to recall Lemma~\ref{lem:T.norm}.
For example,
\begin{multline}
\int_{\Sigma_s} e^{2Hs}e^{2\ell Hs}\Phi g^{ii'}g^{cc'}g_{aa'}\big[\nabla^{(\ell)}\widehat{\Gamma}_{i'c'}^{a'}\star\nabla^{(\ell)}\widehat{k}_c{}^a\big]\nabla_i f^2(t)e^{-3Hs}\vol{g}\leq\\
\leq e^{-Hs} \|\hGamma\|_{H^\ell(\Sigma_s,g)} \|\hk\|_{H^\ell(\Sigma_s,g)}\,.
\end{multline}
Furthermore, we have in the second line of the RHS of \eqref{en.id.Gamma.k},
\begin{multline}
    \sum_{\ell\leq N}\int_{\Sigma_s} e^{2Hs}e^{2\ell Hs}\Bigl[\Phi g^{ab}g_{jj'}\big\{\nabla^{j'}\nabla^{(\ell)} \hPhi+(\mathrm{Error}_{\mathrm{div}\widehat k,\ell})^{j'}\big\}\star\nabla^{(\ell)}\widehat{\Gamma}_{ab}^j\Bigr]f^2(t)e^{-3Hs}\vol{g}\leq\\
    \leq  e^{Hs}\Bigl(\|\nabla \hPhi\|_{H^{N}(\Sigma_s,g)} +\sum_{\ell\leq N} e^{\ell Hs}\|\Er{\mathrm{div}\hk,\ell}\|_{L^2(\Sigma_s,g)}\Bigr)e^{Hs}\|\hGamma\|_{H^{N}(\Sigma_s,g)}\\
    \leq C e^{-\frac{1}{2}Hs} \mathcal{E}_N(s)
    +e^{-\frac{1}{2}Hs} e^{3Hs}e^{2NHs}\|\nabla^{N+1}\hPhi\|_{L^2(\Sigma_s,g)}^2
    +C e^{\frac{1}{2}Hs}\|\widetilde{\mathfrak{C}}_i\|^2_{W^{N,\infty}(\mathbb{R}\times\mathbb{S}^2,\mathring{g})}\,.
\end{multline}
where we have used the error estimate of Proposition~\ref{prop:error.est} \emph{(IV)};
similarly in the third line of RHS of \eqref{en.id.Gamma.k}
\begin{multline}
\sum_{\ell\leq N}\int_{\Sigma_s} e^{2Hs}e^{2\ell Hs}\Bigl[
\big\{\nabla^c\nabla^{(\ell)} \hPhi+(\mathrm{Error}_{\mathrm{div}\widehat k,\ell})^c\big\}\star\nabla^{(\ell)}\nabla_c\widehat{\Phi}\Bigr]f^2(t)e^{-3Hs}\vol{g}\leq\\
\leq e^{2Hs} \|\nabla\hPhi\|_{H^N(\Sigma_s,g)}^2+\sum_{\ell\leq N} e^{Hs} e^{\ell Hs}\|\Er{\mathrm{div}\hk,\ell}\|_{L^2(\Sigma_s,g)} e^{Hs} \|\nabla\hPhi\|_{H^N(\Sigma_s,g)}\\
 \leq 2 e^{-Hs}e^{3Hs}e^{2NHs} \|\nabla^{N+1}\hPhi\|_{L^2(\Sigma_s,g)}^2+ C e^{-Hs} \mathcal{E}_N(s)+ C \|\widetilde{\mathfrak{C}}_i\|^2_{W^{N,\infty}(\mathbb{R}\times\mathbb{S}^2,\mathring{g})}\,.
\end{multline}
Finally, in the fourth line of \eqref{en.id.Gamma.k} we can apply Proposition~\ref{prop:error.est} \emph{(III)}
to obtain
\begin{multline}
    \sum_{\ell\leq N}\int_{\Sigma_s}e^{2Hs}e^{2\ell Hs} \Bigl\{\Er{\hGamma,\ell}\star \nabla^{(\ell)}\hGamma+\Er{\hk,\ell}\star\nabla^{(\ell)}\hk\Bigr\} f^2(t) e^{-3Hs}\vol{g}\leq\\
    \leq C e^{-\frac{1}{2}Hs}\mathcal{E}_N(s)+ e^{-\frac{1}{2}Hs}e^{3Hs} e^{2NHs}\|\nabla^{N+1}\hPhi\|_{L^2(\Sigma_S,g)}^2+C e^{\frac{5}{2}Hs}\|(\widetilde{\mathfrak{J}}_k)_i{}^j\|^2_{W^{N,\infty}(\mathbb{R}\times\mathbb{S}^2,\mathring{g})}
\end{multline}
This concludes the proof of the main energy estimates.
\end{proof}

\section{Precise asymptotics of the perturbed solution}\label{sec:prec.asym}

Now that we have established the global stability estimate \eqref{glob:en.est}, we can derive the precise asymptotic behavior of all variables.
\begin{proposition}\label{prop:k.Phi.ref.est}
The sharp estimate
\begin{align}\label{k.Phi.ref.est}
\|f(t)\widehat{k}\|_{W^{N-4,\infty}(\Sigma_s,g)}+\|f(t)\widehat{\Phi}\|_{W^{N-4,\infty}(\Sigma_s,g)}\leq C\mathring{\varepsilon}e^{-2Hs},
\end{align}
holds for all $s\in[s_0,+\infty)$.
Moreover, the following expansions are valid for $\widehat{g}_{ij},\widehat{\Phi}$:
\begin{align}\label{g.Phi.asym}
\begin{split}
\widehat{g}_{ij}(s,x)=\widehat{g}_{ij}^\infty(x)e^{2Hs}+\widehat{h}_{ij}(s,x),\\
\widehat{\Phi}(s,x)=\widehat{\Phi}^\infty(x)e^{-2Hs}+\widehat{\Psi}(s,x)\,,
\end{split}
\end{align}
where $\widehat{g}^\infty_{ij}(x),\widehat{h}_{ij}(s,x),\widehat{\Phi}^\infty(x),\widehat{\Psi}(s,x)$ satisfy:
\begin{align}\label{g.Phi.rem.est}
\begin{split}
\|f(t)\widehat{g}^\infty(x)\|_{W^{N-4}(\mathbb{R}\times\mathbb{S}^2,\mathring{g})}\leq C\mathring{\varepsilon},\qquad\|f(t)\widehat{h}(s,x)\|_{W^{N-4}(\mathbb{R}\times\mathbb{S}^2,\mathring{g})}\leq C\mathring{\varepsilon},\\
\|f(t)\widehat{\Phi}^\infty(x)\|_{W^{N-6}(\mathbb{R}\times\mathbb{S}^2,\mathring{g})}\leq C\mathring{\varepsilon},\qquad\|f(t)\widehat{\Psi}(s,x)\|_{W^{N-6}(\mathbb{R}\times\mathbb{S}^2,\mathring{g})}\leq C\mathring{\varepsilon}e^{-4Hs}\,,
\end{split}
\end{align}
where the functions $\widehat{\Phi}^\infty(x),\widehat{\Psi}(s,x)$ are well-defined for $N\ge6$.
\end{proposition}
\begin{proof}
Using the global energy estimate \eqref{glob:en.est} to control the RHS of \eqref{k.hat.eq}, we deduce that
\begin{align*}
\|f(t)(\partial_s\widehat{k}+3H\widehat{k})\|_{W^{N-4,\infty}(\Sigma_s,g)}\leq C\mathring{\varepsilon}e^{-\frac{3}{2}Hs}.
\end{align*}
Hence, we have 
\begin{align*}
\bigg\|\int_{s_0}^s\partial_\tau(f(t)e^{3H\tau}\widehat{k})\ud\tau\bigg\|_{W^{N-4,\infty}(\Sigma_s,g)}\leq &\,\int_{s_0}^se^{3H\tau}\|f(t)(\partial_\tau \widehat{k}+3H\widehat{k})\|_{W^{N-4,\infty}(\Sigma_\tau,g)} \ud\tau\leq C\mathring{\varepsilon}e^{\frac{3}{2}Hs}\\
\Rightarrow\quad \|f(t)\widehat{k}\|_{W^{N-4,\infty}(\Sigma_s,g)}\leq&\, e^{-3Hs}e^{3Hs_0}\|f(t)\widehat{k}\|_{W^{N-4,\infty}(\Sigma_{s_0},g)}+ C\mathring{\varepsilon}e^{-\frac{3}{2}Hs}\leq C\mathring{\varepsilon}e^{-\frac{3}{2}Hs},
\end{align*}
for all $s\in[s_0,+\infty)$. Going back to the equation \eqref{Phi.hat.eq} for $\widehat{\Phi}$,
we employ the latter improved estimate for $\widehat{k}$, together with the global estimate \eqref{glob:en.est} to infer that 
\begin{align*}
\|f(t)(\partial_s\widehat{\Phi}+2H\widehat{\Phi})\|_{W^{N-4,\infty}(\Sigma_s,g)}\leq C\mathring{\varepsilon}e^{-3Hs}.
\end{align*}
Repeating the above argument, integrating in $[s_0,s]$, gives 
\begin{align*}
\|f(t)\widehat{\Phi}\|_{W^{N-4,\infty}(\Sigma_s,g)}\leq C\mathring{\varepsilon}e^{-2Hs}.
\end{align*}
Using the latter to bound the RHS of \eqref{k.hat.eq} once more, we obtain the improved estimate
\begin{align*}
\|f(t)(\partial_s\widehat{k}+3H\widehat{k})\|_{W^{N-4,\infty}(\Sigma_s,g)}\leq C\mathring{\varepsilon}e^{-2Hs}.
\end{align*}
Integrating in $[s_0,s]$ and repeating the above argument gives
\begin{align*}
\|f(t)\widehat{k}\|_{W^{N-4,\infty}(\Sigma_s,g)}\leq C\mathring{\varepsilon}e^{-2Hs},
\end{align*}
which completes the proof of \eqref{k.Phi.ref.est}.

Next, we employ the already derived \eqref{k.Phi.ref.est}, together with \eqref{glob:en.est}, to estimate the RHS of \eqref{g.hat.eq}:
\begin{align*}
\|f(t)(\partial_s\widehat{g}-2H\widehat{g})\|_{W^{N-4,\infty}(\mathbb{R}\times\mathbb{S}^2,\mathring{g})}\leq C\mathring{\varepsilon},
\end{align*}
for all $(s,x)\in\mathcal{M}$.
Hence, it follows that 
\begin{align*}
\bigg\|\int_{s_1}^{s_2}\partial_s(f(t)e^{-2Hs}\widehat{g})\ud s\bigg\|_{W^{N-4,\infty}(\mathbb{R}\times\mathbb{S}^2,\mathring{g})}\leq &\,C\mathring{\varepsilon}e^{-2Hs_1},\qquad s_1<s_2.
\end{align*}
This implies that $e^{-2Hs}\widehat{g}_{ij}$ has a $W^{N-4,\infty}(\mathbb{R}\times\mathbb{S}^2,\mathring{g})$ limit, as $s\rightarrow+\infty$, denoted by $\widehat{g}_{ij}^\infty(x)$. 

On the other hand, integrating in $[s,+\infty)$ gives:
\begin{equation}
\bigg\|\int_s^{+\infty}\partial_s(f(t)e^{-2Hs}\widehat{g})ds\bigg\|_{W^{N-4,\infty}(\mathbb{R}\times\mathbb{S}^2,\mathring{g})}\leq C\mathring{\varepsilon}e^{-2Hs}
\end{equation}
Then with $\widehat{h}_{ij}=\hg_{ij}-\hg_{ij}^\infty e^{2Hs}$, it follows
\begin{align*}
\Rightarrow\quad&\left\{\begin{array}{ll}
\|f(t)e^{-2Hs}\widehat{h}\|_{W^{N-4,\infty}(\mathbb{R}\times\mathbb{S}^2,\mathring{g})}\leq C\mathring{\varepsilon}e^{-2Hs}  \\
\|f(t)\widehat{g}^{\infty}\|_{W^{N-4,\infty}(\mathbb{R}\times\mathbb{S}^2,\mathring{g})}\leq e^{-2Hs}\|f(t)\widehat{g}\|_{W^{N-4,\infty}(\mathbb{R}\times\mathbb{S}^2,\mathring{g})}+C\mathring{\varepsilon}e^{-2Hs}
\end{array}\right.\\
\Rightarrow\quad &\left\{\begin{array}{ll}
\|f(t)\widehat{h}\|_{W^{N-4,\infty}(\mathbb{R}\times\mathbb{S}^2,\mathring{g})}\leq C\mathring{\varepsilon}  \\
\|f(t)\widehat{g}^{\infty}\|_{W^{N-4,\infty}(\mathbb{R}\times\mathbb{S}^2,\mathring{g})}
\leq C\mathring{\varepsilon}\,. 
\end{array}\right.
\end{align*}

The expansion for $\widehat{\Phi}$ is derived similarly, using the equation \eqref{Phi.hat.eq}.
The refined bounds~\eqref{k.Phi.ref.est}, together with the global estimate, imply that
\begin{align*}
\|f(t)(\partial_s\widehat{\Phi}+2H\widehat{\Phi})\|_{W^{N-6}(\mathbb{R}\times\mathbb{S}^2,\mathring{g})}\leq C\mathring{\varepsilon}e^{-4Hs}.
\end{align*}
Hence, it follows that 
\begin{align*}
\bigg\|\int_{s_1}^{s_2}\partial_s(f(t)e^{2Hs}\widehat{\Phi})\ud s\bigg\|_{W^{N-6,\infty}(\mathbb{R}\times\mathbb{S}^2,\mathring{g})}\leq &\,C\mathring{\varepsilon}e^{-2Hs_1},\qquad s_1<s_2.   
\end{align*}
Hence, $e^{2Hs}\widehat{\Phi}$ has a $W^{N-6,\infty}(\mathbb{R}\times\mathbb{S}^2,\mathring{g})$ limit, as $s\rightarrow+\infty$, denoted by $\widehat{\Phi}^\infty(x)$. 
Moreover, with $\widehat{\Psi}=\hPhi-\hPhi^\infty e^{-2Hs}$,
integrating in $[s,+\infty)$ gives:
\begin{align*}
\bigg\|\int_s^{+\infty}\partial_s(f(t)e^{2Hs}\widehat{\Phi})\ud s&\bigg\|_{W^{N-6,\infty}(\mathbb{R}\times\mathbb{S}^2,\mathring{g})}\leq C\mathring{\varepsilon}e^{-2Hs}\\
\Rightarrow\quad&\left\{\begin{array}{ll}
\|f(t)e^{2Hs}\widehat{\Psi}\|_{W^{N-6,\infty}(\mathbb{R}\times\mathbb{S}^2,\mathring{g})}\leq C\mathring{\varepsilon}e^{-2Hs}  \\
\|f(t)\widehat{\Phi}^{\infty}\|_{W^{N-6,\infty}(\mathbb{R}\times\mathbb{S}^2,\mathring{g})}\leq e^{2Hs}\|f(t)\widehat{\Phi}\|_{W^{N-6,\infty}(\mathbb{R}\times\mathbb{S}^2,\mathring{g})}+C\mathring{\varepsilon}e^{-2Hs}
\end{array}\right.\\
\Rightarrow\quad &\left\{\begin{array}{ll}
\|f(t)\widehat{\Psi}\|_{W^{N-6,\infty}(\mathbb{R}\times\mathbb{S}^2,\mathring{g})}\leq C\mathring{\varepsilon}e^{-4Hs} \\
\|f(t)\widehat{\Phi}^{\infty}\|_{W^{N-6,\infty}(\mathbb{R}\times\mathbb{S}^2,\mathring{g})}
\leq C\mathring{\varepsilon} \,.
\end{array}\right.
\end{align*}
This completes the proof of the proposition.
\end{proof}

\begin{remark}[Asymptotic expansion] \label{rmk:expansion}
Higher order terms in the asymptotic expansion can be obtained --- assuming successively more derivatives for the solution --- by plugging \eqref{g.Phi.asym} into the equations \eqref{g.hat.eq}, \eqref{Phi.hat.eq} and repeating the arguments in the proof of the previous proposition. 
Similarly, we can derive asymptotic expansions for $\widehat{\Gamma}_{ic}^a,\widehat{k}_i{}^j$; see for example \cite[Theorem 1.3]{fournodavlos:FLRW}. 
A complete asymptotic expansion, and smoothness of the asymptotically rescaled metric across the conformal boundary at infinity $\mathcal{I}^+=\{s=\infty\}$ has recently been obtained in \cite{hintz:vasy:csm:24}.
\end{remark}

\appendix

\section{Local well-posedness}
\label{sec:app}

For completeness, we outline a proof of the local well-posedness of the Einstein vaccum equations with cosmological constant in parabolic gauge.

The argument for \emph{local} existence is qualitatively different from the proof of \emph{global} existence for the following reason: In the derivation of the energy identity \eqref{eq:hk:hGamma:Er} for the variables $\widehat{k},\widehat{\Gamma}$ (which lies at the heart of the global existence argument), we use the momentum constraint \eqref{mom.const.exp}, \eqref{mom.const.exp2}; see discussion in Section~\ref{subsec:en.id.discussion}. Unfortunately, such an identity is not available in a contraction mapping argument (with the aim of proving local well-posedness), since the constraint equations are not valid at each step of an iteration scheme. A remedy for this problem is to consider a system of wave equations for $k_{ij}$, instead of the propagation equation \eqref{eq:secondvariationalformula} in the first place. 

For the resulting \emph{reduced} system of equations,
standard results for quasilinear wave type systems can be used to establish local well-posedness. After showing that the solution produced is actually a solution to the Einstein vacuum equations \eqref{eq:EVE}, we infer as a corollary the local well-posedness of the original evolution equations in parabolic gauge (Section \ref{sec:evoleqs}), i.e. \emph{existence, uniqueness and continuous dependence of solutions to initial data}.

This type of argument was already employed in \cite[Section 10.2.2]{ch:kl} to prove the local well-posedness of the first and second variation equations \eqref{1stvar}, \eqref{2ndvar}, in the maximal gauge, where the corresponding equation for the lapse is elliptic:\[\tr k=0\qquad \triangle\Phi=|k|^2\Phi\,,\] see also \cite{FSm:maximal} for the extension to an initial-\emph{boundary} value problem. Similarly this argument was adapted in \cite[Section 14]{rod:speck:bang} for the CMC gauge $\tr k=f(t)$. 

Now consider the equations \eqref{2ndvar.2}, \eqref{Gamma.eq} for $k_i{}^j,\Gamma_{ic}^a$ which follow from \eqref{1stvar}, \eqref{2ndvar}. The parabolic gauge \eqref{lapse} can be written in the form $\tr k=\Phi+\text{a known source term}$. Since it makes no difference for local well-posedness, we drop the source term and consider for convenience the gauge condition:
\begin{align}\label{lapse:app}
\Phi=\tr k. 
\end{align}
First we derive a second order (in time) equation satisfied by $k_{ij}$.
Here the Einstein equations are \emph{not} imposed. Instead we derive a geometric condition, adapted to \eqref{lapse:app} which is equivalent to a wave equation for $k_{ij}$.
\begin{proposition}\label{lem:boxk}
Let $s$ be a time function for the spacetime metric $\g$,
and $g,k$ the first and second fundamental forms as in  \eqref{metric} and \eqref{eq:firstvariation}.
Then $k_{ij}$ satisfies the wave equation
\begin{equation}
    \label{boxk}
    e_0^2k_{ij}-\Delta_g k_{ij} = I_{ij}+H_{ij}\,,
\end{equation}
where 
\begin{subequations} \label{eq:boxk:H}
\begin{align}
H_{ij}=&-\Phi^{-3}\partial_s\Phi\nabla_i\nabla_j\Phi+\Phi^{-2}\nabla_i\nabla_j\partial_s\Phi-\nabla_i\nabla_j\Phi-\Phi^{-2}\partial_s\Gamma^l_{ij}\partial_l\Phi-e_0(k_{ij}\mathrm{tr}k-2{k_i}^lk_{jl})\\
\notag&+\Phi^{-1}k_{ij}\Delta_g\Phi
-\Phi^{-1}k_i{}^a\nabla_a\nabla_j\Phi
-\Phi^{-1}k_j{}^a\nabla_a\nabla_i\Phi\\
\notag&-\Phi^{-1}\nabla^a\Phi(\nabla_jk_{ia}+\nabla_ik_{ja}-2\nabla_ak_{ij})
+\Phi^{-1}\mathrm{tr}k\nabla_j\nabla_i\Phi\\
\notag&+\Phi^{-1}\nabla_j\Phi(\nabla_i\mathrm{tr}k-\nabla_ak_i{}^a)
+\Phi^{-1}\nabla_i\Phi(\nabla_j\mathrm{tr}k-\nabla_ak_j{}^a)\,,\\
I_{ij}=&-3(k_{ci}\mathrm{Ric}_j{}^c+k_{cj}\mathrm{Ric}_i{}^c) +2\mathrm{tr}k\mathrm{Ric}_{ji}
+2 g_{ji}\mathrm{Ric}_a{}^ck_c{}^a+(k_{ij}-g_{ji}\mathrm{tr}k)R+2\Lambda k_{ij} \,,
\end{align}
\end{subequations}
if and only if the following propagation equation holds:
\begin{align}\label{e0Rij4}
e_0({\RIC}_{ij}-\Lambda g_{ij})=\nabla_i\mathcal{G}_j+\nabla_j\mathcal{G}_i-\nabla_i\nabla_j(\Phi-\mathrm{tr}k),\qquad\mathcal{G}_i:={\RIC}_{0i},\qquad e_0=\Phi^{-1}\partial_s\,.
\end{align}
\end{proposition}
\begin{proof}
The wave equation for $k_{ij}$ is obtained from the second variation equation \eqref{2ndvar}. Indeed, multiplying \eqref{2ndvar} by $\Phi^{-1}$, and differentiating in $s$, it remains to compute the time derivative of the Ricci curvature, $\partial_s\Ric_{ij}$, directly from \eqref{eq:Ricci:local}. This is done in the proof of \cite[Proposition 2.1]{FSm:maximal} (as in \cite[Proof 10.2.2]{ch:kl}), and leads to \cite[(2.16)]{FSm:maximal}. (Note that \cite{FSm:maximal} uses the opposite sign convention in the definition of $k_{ij}$.) The only difference is that here we have written \eqref{e0Rij4} as a propagation equation for $\RIC_{ij}-\Lambda g_{ij}$, in terms of the gauge quantity $\tr k-\Phi$, with corresponding terms in \eqref{eq:boxk:H}.
\end{proof}
We couple the wave equation \eqref{boxk} to the following evolution equations:
\begin{align}
\label{1stvar.app}\partial_sg_{ij}=&\,2\Phi k_{ij}\,,\\
\label{lapse.eq.app}    \partial_s\Phi=&\,\Delta_g\Phi+\Lambda\Phi-\Phi |k|^2\,.
\end{align} 
Together, they form closed system for $(\Phi,g_{ij},k_{ij})$.
We now turn to the local well-posedness of the system (\ref{boxk}, \ref{1stvar.app}, \ref{lapse.eq.app}).
\begin{proposition}\label{prop:loc.well}
The system \eqref{boxk}, \eqref{1stvar.app}, \eqref{lapse.eq.app} is locally well-posed in the spaces
\begin{subequations} \label{loc.well.sp}
\begin{align}
 &k_{ij}\in C^0\big([s_0,s_1];H^N(\Sigma_s,g)\big)\cap C^1\big([s_0,s_1];H^{N-1}(\Sigma_s,g)\big),\\ 
&g_{ij}\in C^0\big([s_0,s_1];H^N(\Sigma_s,g)\big),\\ 
&\Phi\in C^0\big([s_0,s_1];H^N(\Sigma_s)\big)\cap L^2\big([s_0,s_1];H^{N+2}(\Sigma_s)\big),
\end{align}
\end{subequations}
for $N\ge4$.
\end{proposition}
\begin{proof}[Sketch of the proof]
The solution is obtained from a Picard iteration in the given spaces; see \cite[Proof 10.2.2]{ch:kl} and \cite[Section 3]{FSm:maximal} for the maximal case. Each iterate $(k^n,g^n,\Phi^n)$ is then subject to a linear equation. Note that $N\ge 4$ allows for control of up to two spatial derivatives of $k_{ij},g_{ij},\Phi$ in $L^\infty$ (by the Sobolev inequality of Lemma~\ref{lem:WN-2.infty}), which is sufficient to handle the coefficients of the quasi-linear system.

The argument primarily relies on the energy estimate for $k_{ij}$.
Once $(k_{ij},\partial_s k_{ij})$ is controlled in $H^N\times H^{N-1}$ on $\Sigma_s$, then from the transport equation \eqref{1stvar.app}, and the parabolic equation \eqref{lapse.eq.app}, we obtain control over $(g_{ij},\Phi)$ in the stated norms through energy and parabolic estimates. It remains to show that the energy estimate for $k_{ij}$ can be completed at the given Sobolev regularity. To obtain an $H^N$ energy estimate for $k_{ij}$, we commute the equation \eqref{boxk} with $\nabla^{(\ell)}$, $\ell\leq N-1$, multiply the commuted equation with $e_0\nabla^{(\ell)}k_{ij}$, and integrate in spacetime. All terms are estimated directly in energy, apart from the ones involving the spatial curvature and the Hessian of $\partial_s\Phi$ in \eqref{eq:boxk:H}, at top order $\ell=N-1$. These are of the form
\begin{align}\label{prob.terms.well}
\nabla^N\Gamma\, e_0\nabla^{N-1}k,\qquad \nabla^{N+1}\partial_s\Phi \,e_0\nabla^{N-1}k\,,
\end{align}
and cannot be directly estimated in energy, since they contain $N+1$ derivatives of $g$ and $N+3$ spatial derivatives of $\Phi$, if plug in \eqref{lapse.eq.app}. However, they can be treated as follows: $(i)$~In the first term, we integrate by parts, once in $\nabla$ and once in $e_0$, to produce a spacetime integrand of the form $e_0\nabla^{N-1}\Gamma\, \nabla^N k$. Now we observe that \eqref{1stvar.app} can be used to eliminate the $e_0$ derivative and obtain a quadratic term in $\nabla^Nk$, which is at the level of the $H^N$ energy of $k$. The boundary terms generated from integrating by parts in $e_0$ can be absorbed using Young's inequality.
$(ii)$~For the second term in \eqref{prob.terms.well} we differentiate \eqref{lapse.eq.app} in $\partial_s$, instead of replacing $\partial_s\Phi$ with $\Delta_g\Phi$. Then we observe that $\partial_s\Phi$ satisfies a parabolic equation with coefficients that contain at most one derivative of $k$. Hence, the standard parabolic estimate gives us control over $\partial_s\Phi$ in $L^2_sH^{N+1}$, i.e. two orders more regular in space than $\nabla k,e_0k\in H^{N-1}$. 
This is sufficient to complete the energy estimate for $k$ in $H^N\times H^{N-1}$. 
\end{proof}
We conclude this appendix with the argument that solutions to the reduced system of equations, with specific initial data, are in fact solutions to the Einstein vacuum equations with positive cosmological constant. Initial data for the reduced system of equations treated in Proposition~\ref{prop:loc.well} consists of $(\Phi,g_{ij},k_{ij},\partial_s k_{ij})$ on $\Sigma_{s_0}$.
In order to obtain a solution to Einstein's equations in parabolic gauge, the lapse $\Phi$ needs to satisfy \eqref{lapse:app} initially, and $g$ and $k$ must satisfy the constraint equations \eqref{mom.const}, and \eqref{Ham.const}.
Moreover $\partial_s k_{ij}$ must be prescribed in agreement with the second variation equation \eqref{2ndvar}:
\begin{align}\label{EVE.init.app}
\begin{split}
\Phi=\mathrm{tr}k,\qquad
 \partial_s k_{ij}=\nabla_i\nabla_j\Phi-\Phi(\Ric_{ij}+\tr k\, k_{ij}-2k_i{}^mk_{mj}-\Lambda g_{ij}),\\
 R-\lvert k\rvert^2+(\tr k)^2=2\Lambda,\qquad
\mathrm{div}_g\, k - \ud \tr_g k=0.\qquad\qquad\qquad
\end{split}
\end{align}
Finally we show that for such initial data a solution to the Einstein equations $G_{\mu\nu}=0$ is obtained, where $G_{\mu\nu}$ is the Einstein tensor:
\begin{align}\label{Gij}
G_{\mu\nu}={\RIC}_{\mu\nu}-\frac{1}{2}{\g}_{\mu\nu}{\R}+\Lambda {\g}_{\mu\nu}\,.
\end{align}
\begin{proposition}\label{prop:EVE}
Suppose $\g$ is a metric of the form \eqref{metric},
and $(g_{ij},\Phi, k_{ij})$ is a solution to the reduced system  \eqref{1stvar.app}, \eqref{lapse.eq.app}, and \eqref{boxk}. 
Then 
\begin{align}\label{EVE.rec.syst1a}
{\bf Ric}_{00}+\Lambda=&\,e_0(\Phi-\mathrm{tr}k),
\end{align}
\begin{align}\label{EVE.rec.syst1b}
e_0^2(\Phi-\mathrm{tr}k)-\Delta_g(\Phi-\mathrm{tr}k)=&-2\mathrm{tr}k\,e_0(\Phi-\mathrm{tr}k)
+4\Phi^{-1}\nabla^i\Phi\,\mathcal{G}_i,
\end{align}
\begin{align}\label{EVE.rec.syst2}
\notag e_0G_{ij}=&\,\nabla_i\mathcal{G}_j+\nabla_j\mathcal{G}_i-g_{ij}\nabla^a\mathcal{G}_a\\
&+\frac{1}{2}g_{ij}\,e_0^2(\Phi-\mathrm{tr}k)
+\frac{1}{2}g_{ij}\,\Delta_g(\Phi-\mathrm{tr}k)-\nabla_i\nabla_j(\Phi-\mathrm{tr}k)\\
\notag&+(2k_{ij}-g_{ij}\mathrm{tr}k)G_a{}^a+g_{ij} k^{ab}G_{ab}+
(g_{ij}\mathrm{tr}k-2k_{ij})\,e_0(\Phi-\mathrm{tr}k),
\end{align}
and
\begin{align}\label{EVE.rec.syst3}
e_0^2\mathcal{G}_i-\Delta_g \mathcal{G}_i= L_i\bigl[\nabla^{a_1} e_0^{a_2}(\Phi-\mathrm{tr}k),\nabla^{a_3}\mathcal{G},\nabla^{a_3}G: a_1+a_2\leq 2,\, a_3\leq1 \bigr],
\end{align}
where $L_i$ is linear in the given variables, with coefficients that depend on up to two derivatives of the metric $\g$.

Furthermore with initial data satisfying \eqref{EVE.init.app}, the solution to the reduced system  gives rise to a trivial solution of the system (\ref{EVE.rec.syst1b}-\ref{EVE.rec.syst3}).
In particular, $G_{\mu\nu}=0$.

%
%
\end{proposition}
\begin{proof}
The first equation  \eqref{EVE.rec.syst1a} 
follows from the propagation equation for the lapse  \eqref{lapse.eq.app}  and the propagation equation \eqref{eq:trk} for $\tr k$.
The equation \eqref{EVE.rec.syst1b} follows by taking the trace of the wave equation for $k_{ij}$ \eqref{boxk}, and using equations \eqref{1stvar.app}, \eqref{lapse.eq.app}, along with the identities \eqref{1stvar.inv.2} and \eqref{eq:perppara}:
\begin{align*}
e_0 g^{ij}=-2 k^{ij},\qquad \mathcal{G}_i=\nabla^jk_{ij}-\nabla_i\mathrm{tr}k.
\end{align*}
The computations are long, but note for example that $g^{ij} I_{ij}=2\Lambda \tr k$. 
To derive the equation \eqref{EVE.rec.syst2}, we use the propagation equation \eqref{e0Rij4}, together with \eqref{EVE.rec.syst1a}:
\begin{align*}
e_0G_{ij}=&\,e_0({\bf Ric}_{ij}-\Lambda g_{ij})-\frac{1}{2}g_{ij}\,e_0({\bf R}-4\Lambda)-k_{ij}{\bf R}+4\Lambda k_{ij}\\
=&\,\nabla_i\mathcal{G}_j+\nabla_j\mathcal{G}_i
-\nabla_i\nabla_j(\Phi-\mathrm{tr}k)
+g_{ij} k^{ab}(\RIC_{ab}-\Lambda g_{ab})\\
&+\frac{1}{2}g_{ij}\,e_0^2(\Phi-\mathrm{tr}k)
-g_{ij}\nabla^a\mathcal{G}_a
+\frac{1}{2}g_{ij}\,\Delta_g(\Phi-\mathrm{tr}k)
+k_{ij}(4\Lambda-{\bf R}),
\end{align*}
and then substitute the algebraic relations
\begin{multline*}
4\Lambda-{\bf R}=G_\mu{}^\mu,\qquad
G_{\alpha\beta}-\frac{1}{2}{\boldsymbol g}_{\alpha\beta}G_\mu{}^\mu={\bf Ric}_{\alpha\beta}-\Lambda {\boldsymbol g}_{\alpha\beta},\qquad
G_{00}+\frac{1}{2}G_\mu{}^\mu={\bf Ric}_{00}+\Lambda,\\
\Longrightarrow\quad G_{00}=2e_0(\Phi-\mathrm{tr}k)-G_a{}^a\,,\qquad 4\Lambda-{\bf R}=2G_a{}^a-2e_0(\Phi-\mathrm{tr}k).
\end{multline*}
For the derivation of \eqref{EVE.rec.syst3} we use the second Bianchi identity:
\begin{equation}\label{Einstein.Bianchi}
{\Nabla}_0 G_{0i}={\boldsymbol\nabla}^j G_{ji}
\quad\Rightarrow\quad e_0\mathcal{G}_i=\nabla^jG_{ji}-\mathrm{tr}k\,\mathcal{G}_i+
\Phi^{-1}\nabla^j\Phi\, G_{ji}+\Phi^{-1}\nabla_i\Phi\, G_{00}+k_i{}^j\mathcal{G}_j
\end{equation}
Taking the divergence of \eqref{EVE.rec.syst2} and plugging in \eqref{Einstein.Bianchi}, \eqref{EVE.rec.syst1b} gives \eqref{EVE.rec.syst3}.

Initial data for the equations (\ref{EVE.rec.syst1b}-\ref{EVE.rec.syst3}) consists of $\Phi-\mathrm{tr}k$, $e_0(\Phi-\tr k)$, $G_{ij}$, $\mathcal{G}_i$, $e_0\mathcal{G}_i$ on $\Sigma_{s_0}$. For a solution to (\ref{boxk}, \ref{1stvar.app}, \ref{lapse.eq.app}) whose initial data satisfies \eqref{EVE.init.app}, we have that
\begin{align}\label{EVE.init.app2}
{\RIC}_{\mu\nu}=\Lambda {\g}_{\mu\nu}\Longrightarrow e_0(\Phi-\mathrm{tr}k)=0\,, G_{ij}=0\,, \mathcal{G}_i=0,\quad \text{on $\Sigma_{s_0}$}.
\end{align}
Going back to \eqref{Einstein.Bianchi}, we then infer that $e_0\mathcal{G}_i=0$ on $\Sigma_{s_0}$ as well. Hence, the initial data for the system \eqref{EVE.rec.syst1b}-\eqref{EVE.rec.syst3} is trivial.

Finally, 
we observe that solutions to the system (\ref{EVE.rec.syst1b})-(\ref{EVE.rec.syst3}) are unique, with $(G_{ij},\mathcal{G}_i,\Phi-\tr k)$ in $L^2\times H^1 \times H^2$.
Indeed, we note that two derivatives of $\Phi-\mathrm{tr}k$ can be controlled from one spatial derivative of $\mathcal{G}_i$, at the level of an $H^2$ energy estimate using \eqref{EVE.rec.syst1b}. In turn, $G_{ij}$ is controlled in $\mathrm{L}^2$ by the $\mathrm{H}^1$ energy of $\mathcal{G}_i$ and the $\mathrm{H}^2$ energy of $\Phi-\mathrm{tr}k$, using \eqref{EVE.rec.syst2}. From \eqref{EVE.rec.syst3}, we notice that the $\mathrm{H}^1$ energy of $\mathcal{G}_i$ is at the level of the $\mathrm{H}^2$ energy of $\Phi-\mathrm{tr}k$ and the $\mathrm{H}^1$ energy of $G_{ij}$. The latter term loses a derivative in an energy argument, since we only control it in $\mathrm{L}^2$. To circumvent the problem, after multiplying \eqref{EVE.rec.syst3} with $e_0\mathcal{G}_i$ and integrating in spacetime, we replace the term $\nabla G e_0\mathcal{G}_i$ with $e_0G\nabla\mathcal{G}_i$ by integrating by parts. Then, we use \eqref{EVE.rec.syst2} to eliminate $e_0 G$ in favor of derivatives of $\mathcal{G}$, and $\Phi-\mathrm{tr}k$. 

Note that in the process, we only assume that up to two derivatives of ${\boldsymbol g}$ are bounded. This is consistent with the global existence statement in Theorem \ref{thm:fs}, cf. Lemma \ref{lem:WN-2.infty}, $N\ge4$.
\end{proof}

\section{Einstein equations in parabolic gauge}
\label{app.eqs}

We collect the equations derived in Section~\ref{sec:eqs} for a solution $\g$ to the Einstein vacuum equations in parabolic gauge, with respect to an arbitrary reference metric $\widetilde{\g}$.

\smallskip
\textit{Metric:}
\begin{equation}\label{app.metric}
\g=-\Phi^2 \ud s^2+ {g}_{ij}\ud x^i \ud x^j
\end{equation}
The spacetime metric $\g$ satisfies $\RIC[\g]=\Lambda \g$.
Here $\Lambda>0$, and we set $H=\sqrt{\frac{\Lambda}{3}}$.

\smallskip
\textit{Reference metric:}
\begin{align}\label{app.metric.ref}
\widetilde{\g} =-\widetilde{\Phi}^2\ud s^2+ \widetilde{g}_{ij}\ud x^i\ud x^j
\end{align}
The components of the Ricci curvature of $\widetilde{\g}$ are denoted by $\widetilde{\RIC}_{\mu\nu}=\RIC[\widetilde{\g}]_{\mu\nu}$.

\smallskip
\textit{Differences:}
\begin{align}\label{app.diff}
\begin{split}
\widehat\Phi=\Phi-\widetilde\Phi,\qquad \widehat{g}_{ij}=g_{ij}-\widetilde{g}_{ij},\qquad \widehat{g}^{ij}=g^{ij}-\widetilde{g}^{ij},\\
\widehat\nabla=\nabla-\widetilde{\nabla},\qquad\widehat\Gamma_{ij}^a=\Gamma_{ij}^a-\widetilde{\Gamma}_{ij}^a,\qquad\widehat{k}_i{}^j=k_i{}^j-\widetilde{k}_i{}^j\,.
\end{split}
\end{align}
Here $k_i{}^j=g^{cj}k_{ic}$ and $\widetilde{k}_i{}^j=\widetilde{g}^{cj}\widetilde{k}_{ic}$.

\smallskip
\textit{Parabolic gauge:}
\begin{equation}
    \label{app.lapse}
    \widehat\Phi=\widehat{k}_l{}^l
\end{equation}

\smallskip
\textit{First variation formula:}
\begin{equation}
\label{app.g.hat.eq}\partial_s\widehat{g}_{ij}-2H\widehat{g}_{ij}=\,2\Phi g_{ja}\widehat{k}_i{}^a+
2\widehat{\Phi}g_{ja}\widetilde{k}_i{}^a
+2\widetilde{\Phi}(\widetilde{k}_i{}^a-H\widetilde{\Phi}^{-1}\delta_i{}^a)\widehat{g}_{ja}  
\end{equation}
Similarly, $\widehat{g}^{ij}$ satisfies \eqref{g.inv.hat.eq}.

\smallskip
\textit{Lapse equation:}
\begin{equation}
\label{app.Phi.hat.eq}\partial_s\widehat{\Phi}-\Delta_g\widehat{\Phi}+2H\widehat{\Phi}=\mathfrak{F}+\tPhi\Bigl(\widetilde{\RIC}_{00}+\Lambda\Bigr)
  \end{equation}
Here $\mathfrak{F}$ is given by \eqref{frak.F}.

\smallskip
\textit{Second variation equation:}
  \begin{align}
\label{app.k.hat.eq}\partial_s\widehat{k}_i{}^j+3H\widehat{k}_i{}^j=&
\,g^{cj}\nabla_i\nabla_c\widehat\Phi-\widetilde{\Phi}(\widetilde{k}_l{}^l-3H\widetilde{\Phi}^{-1})\widehat{k}_i{}^j+\mathfrak{K}_i{}^j-\widetilde{\Phi}\Bigl(\widetilde{\RIC}_i{}^j-\Lambda \delta_i{}^j\Bigr)\\
\notag&+\frac{1}{3}\Phi g^{cj}(\nabla_c\widehat\Gamma_{ia}^a-\nabla_a\widehat\Gamma_{ci}^a)
+\frac{2}{3}\Phi g^{ab}(\nabla_a\widehat\Gamma^j_{bi}-\nabla_i\widehat\Gamma^j_{ab})
  \end{align} 
The expression for the remaining terms  $\mathfrak{K}_i{}^j$ is given in \eqref{frak.K}.

\smallskip
\textit{Connection:}
\begin{equation}
    \label{app.Gamma.hat.eq} \partial_s\widehat{\Gamma}_{ic}^a=\,\Phi\nabla_i\widehat{k}_c{}^a
+\Phi\nabla_c\widehat{k}_i{}^a-g^{ab}g_{cj}\Phi\nabla_b\widehat{k}_i{}^j+\mathfrak{G}_{ic}^a
\end{equation}
The expression for $\mathfrak{G}_{ic}^a$ is given in \eqref{frak.G}.

\smallskip
\textit{Momentum constraint:}
\begin{equation}
    \label{app.mom.const.hat}
\nabla_j\widehat k_i{}^j=\,\partial_i\widehat{\Phi}-\widehat{\Gamma}_{jc}^j(\widetilde{k}_i{}^c-H\delta_i{}^c)+\widehat{\Gamma}_{ji}^c(\widetilde{k}_c{}^j-H\delta_c{}^j)
-\widetilde{\RIC}_{0i}
\end{equation}

\input{main_v2.bbl}

\end{document}

%% file: main_v2.bbl
\providecommand{\bysame}{\leavevmode\hbox to3em{\hrulefill}\thinspace}
\providecommand{\MR}{\relax\ifhmode\unskip\space\fi MR }
\providecommand{\MRhref}[2]{%
  \href{http://www.ams.org/mathscinet-getitem?mr=#1}{#2}
}
\providecommand{\href}[2]{#2}